\documentclass[11pt]{amsart}
\usepackage{header}
\usepackage{macros}

\title{Decentralized Distributed Graph Coloring II: \\ $degree+1$-Coloring Virtual Graphs}
\date{}
\usepackage[foot]{amsaddr}
\author{Maxime flin}
\email{maximef@ru.is}
\author{Magn\'us M. Halld\'orsson}
\email{mmh@ru.is}
\author{Alexandre Nolin}
\email{alexandre.nolin@cispa.de}
\address{(M. Flin and M.M. Halld\'orsson) Reykjavik University} 
\address{(A. Nolin) CISPA Helmholtz Center for Information Security}
\thanks{M. Flin was supported by the Icelandic Research Fund (grant 2310015). M.M. Halld\'orsson was partially supported by the Icelandic Research Fund (grant 217965).}

\begin{document}
\begin{abstract}
Graph coloring is fundamental to distributed computing.
We give the first general treatment of the coloring of virtual graphs, where the graph $H$ to be colored is locally embedded within the communication graph $G$.
Besides generalizing classical distributed graph coloring (where $H=G$), this captures other previously studied settings, including cluster graphs and power graphs. 

We find that the complexity of coloring a virtual graph depends on the edge congestion of its embedding. The main question of interest is how fast we can color virtual graphs of constant congestion.
We find that, surprisingly, these graphs can be colored nearly as fast as ordinary graphs. Namely, we give a $O(\log^4\log n)$-round algorithm for the deg+1-coloring problem, where each node is assigned more colors than its degree. 

This can be viewed as a case where a distributed graph problem can be solved even when the operation of each node is decentralized.
\end{abstract}

\maketitle

\vspace{-1em}
\setcounter{tocdepth}{1}
\tableofcontents
\newpage

\section{Introduction}
\label{sec:intro}

Most distributed graph algorithms assume communication on the input graph. Namely, that the graph that forms the input to the computational problem at hand is equivalent to the communication network infrastructure.
In the \local model, this is often without loss of generality, as simulating a round of \local on $H$ while communicating on $G=(V_G,E_G)$ without bandwidth restriction is trivial as long as adjacent vertices in $H$ are $O(1)$-hops away in $G$.
When we restrict message size, however, naive simulation is prohibitively inefficient. The delivery of individual messages to each neighbor of a node can slow down the algorithm by a factor proportional to degrees, which might be as high as $n = |V_H|$. Handling cases where $H\neq G$ is an overarching issue in the design of \congest algorithms (e.g., in \cite{ghaffari2013cut,ghaffari2016distributed,GKKLP18,RG20,GGR20,FGLPSY21,GZ22,RozhonGHZL22,GHIR23}) that is salient when using a \congest algorithm as a subroutine (e.g., local rounding \cite{FGGKR23} used in \cite{GHIR23}) or when modifying the input graph (e.g., contracting edges \cite{GKKLP18,FGLPSY21}). 
We attempt to study \emph{how bandwidth constraints affect distributed algorithms solving problems on graphs whose description is itself distributed on a communication network}.
In this paper, we focus on symmetry breaking and thus ask

\medskip
\begin{quote}
    \emph{How efficiently can $H$ be colored when distributed on a network $G$?}
\end{quote}
\medskip

Coloring problems are of fundamental importance to distributed graph algorithms. In fact, in its seminal paper \cite{linial92}, Linial studied the locality of $3$-coloring cycles. A long line of work \cite{linial92,luby86,SW10,BEPSv3,HSS18,CLP20,RG20} showed that $\Delta+1$-coloring could be achieved in $\poly(\log\log n)$ rounds of \local. Further work extended the result to local list sizes \cite{HKNT22}, and small messages \cite{GGR20,HKMT21,HNT22}. 
We extend these results to embedded graphs in nearly the same number of rounds while using local color lists (in a slightly weaker sense than in \cite{HKNT22}). 

\subsection{Virtual Graphs}
\label{sec:intro-virtual-graph}

\begin{wrapfigure}{r}{0.3\linewidth}
\centering
    \includegraphics[width=.8\linewidth,page=1]{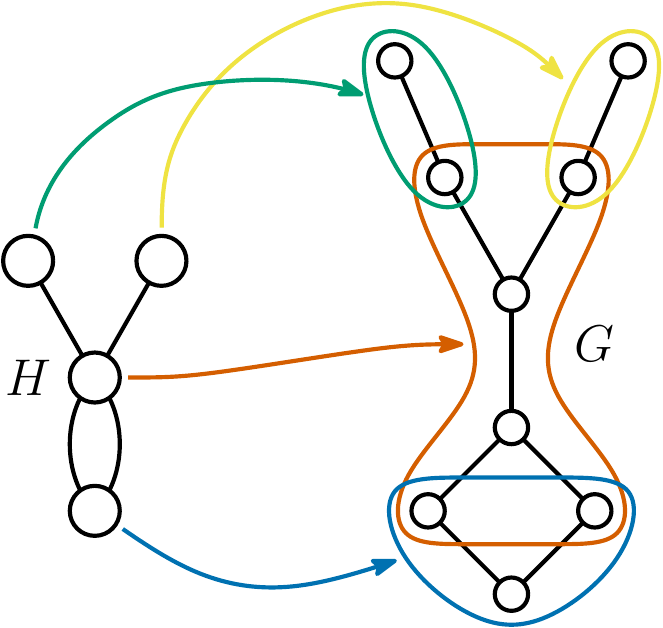}
    \caption{\footnotesize A virtual graph $H$ (on the left) embedded on a network $G$ (on the right). 
On this example, there is a unique choice of support trees; they have 
congestion $\congestion=1$ and dilation $\dilation = 3$.
    \label{fig:example-embedding}}
\end{wrapfigure}

Before answering our research question, we clarify the meaning of \emph{embedding} a graph $H$ into a network $G$. We give here a high-level definition and expound on the formal definitions in \cref{sec:model}. For clarity, we refer to $H=(V_H,E_H)$ as the input or virtual graph while $G=(V_G,E_G)$ is the communication network. We call elements of $V_H$ vertices or nodes while elements of $V_G$ are machines; elements of $E_H$ are edges or conflicts while elements of $E_G$ are links. 

We set the definition of embedded virtual graphs forth by specifying which machine knows about which vertex and edge of $H$.
Each vertex $v\in V_H$ is mapped to a set $V(v) \subseteq V_G$ of machines such that \emph{vertices $u,v\in H$ are adjacent (in $H$) only if their support intersect}, i.e., $V(v) \cap V(u) \neq \emptyset$. We also assume that each support $V(v)$ is equipped with a spanning tree $T(v)$ (called support tree) that can be used to perform aggregation. We assume that machines $w \in V_G$ know about all the supports they belong to --- the set of $v$ such that $V(v) \ni w$ --- as well as which support tree their adjacent links belong to. Each edge $uv \in E_H$ is mapped to a machine $w\in V(u)\cap V(v)$ in the intersection of the two nodes' supports, which knows about the existence of that edge. \Cref{fig:example-embedding} exemplifies such an embedding.

It is convenient to design algorithms for $H$ as a sequence of (virtual) rounds with the same three-step structure\footnote{we emphasize, however, that algorithms are not limited to this scheme and can communicate on the network more cleverly.}: first, broadcast a message to all vertices on the support; second, machines at intersections of supports perform local computations; third, converge-cast the result of these computations on the support trees. Naturally, the efficiency of any such algorithm is limited by (1) the diameter of the support trees and (2) the number of trees using the same edge. We call the former the \emph{dilation} and the latter the \emph{congestion}. In some cases, most of the effort is in computing a good embedding, meaning with small enough dilation and congestion. 
For instance, in \cite{FGLPSY21}, the struggle is in finding $n^{o(1)}$-congestion embeddings for various sparsifiers. In this paper, besides direct applications, we assume the embedding is given as part of the input.

Last but not least, we allow $H$ to be a \emph{multi}-graph (without self-loops) to capture the fact that supports can intersect in multiple places. For instance, in \cref{fig:example-embedding}, the central vertex is adjacent to the bottom vertex through two paths in the network. While distinguishing between the number of incident edges and adjacent vertices is not always necessary, it is crucial for graph coloring, especially when --- like in this paper --- the number of colors used by each vertex depends on its degree.

\subsection{Our Contributions}
Our conceptual contribution is an explicit formalization of the notion of \emph{virtual graphs} that captures the aforementioned examples. 
We show that the key parameters of congestion $\congestion$ and the dilation $\dilation$ essentially capture the hardness of the coloring problem. On one hand, they limit the efficiency of any $\deg+1$-coloring algorithm:
\medskip

\begin{result}
\begin{restatable}{theorem}{TheoremLowerBound}
    \label{thm:lower-bounds}
    Any constant-error algorithm for $3$-coloring a $2$-regular virtual graph $H$ embedded on a network with bandwidth $\bandwidth$, congestion $\congestion$, and dilation $\dilation$, requires $\Omega(\frac{\congestion}{\bandwidth} + \dilation\cdot \log^*n)$ rounds in the worst-case.
\end{restatable}
\end{result}

\medskip
We emphasize that the lower bound applies to algorithms working for any given embedding. 
It applies to all such algorithms, and not just those following the three-step process described in \cref{sec:intro-virtual-graph}.

Conversely, we provide a nearly optimal upper bound for coloring virtual graphs. Applied to the \congest model --- when $H = G$ --- its complexity nearly matches the state-of-the-art $O(\log^3\log n)$ round complexity of \cite{HKMT21,HNT22}.
\medskip

\begin{result}
\begin{restatable}{theorem}{ThmRandGeneral}
    \label{thm:rand-general}
    Let $H$ be a virtual graph on network $G$ with $|V_G| = n$ machines, bandwidth $\bandwidth=O(\log n)$, congestion $\congestion \le n$ and dilation $\dilation$.
    There exists an algorithm to $\deg+1$-color $H$ in $O(\congestion \dilation \cdot \log^4\log n)$ rounds. More precisely, at the end of the algorithm, each vertex $v\in V_H$ has a color $\col(v) \in \set{1, 2, \ldots, \deg(v)+1}$ where $\deg(v)$ is the number of edges incident to $v$ in $H$.
\end{restatable}
\end{result}
\medskip

A key reason for considering the $\deg+1$-coloring problem is that we forgo 
using some frequently assumed global knowledge --- here, the maximum degree $\Delta$.
This is the source of substantial technical challenges, sketched in \cref{sec:intro-tech-overview}. That virtual nodes can be connected with multiplicity breaks several classic arguments, hence requires novel ideas to reach the usual goals of providing nodes with excess colors, and classifying them according to their potential in that respect. Our adaptation of the Ghaffari-Kuhn algorithm (\cref{sec:finish-coloring-low-deg}) to our distributed paradigm might be of independent interest.

\subsection{Technical Overview}
\label{sec:intro-tech-overview}

\paragraph{The Lower Bound.}
We prove lower bounds on the congestion and dilation separately. Since a $o(\dilation\log^*n)$ round algorithm for coloring virtual graphs implies a $o(\log^*n)$ round \local algorithm for coloring cycles, the lower bound on the dilation follows from \cite{linial92,naor95}. To prove the lower bound on the congestion, we provide a probability distribution on gadgets (a 2-regular 16-vertex graph) where vertices are partitioned between two sets $V_A$ and $V_B$. The gadget is such that if Alice (respectively Bob) knows all vertices and edges incident to $V_A$ (respectively $V_B$), then for Alice and Bob to assign colors to their vertices such that the coloring is proper, they must communicate $\Omega(1)$ bits. A classic direct sum argument shows that solving $k$ independent copies of this communication problem requires $\Omega(k)$ bits of communication. Finally, we embed the coloring problem on a graph where Alice's vertices are separated from Bob's through a bridge, causing congestion to be $\congestion = k$.

\paragraph{The Upper Bound: Inaccurate Degrees.}
The main challenge for coloring virtual graphs is that vertices do not have direct access to their list of available colors (or palette). Previous work \cite{HKMN20,FHN23} demonstrated that it was not necessary if vertices could instead estimate certain local density parameters. While in \cite{HKM20,FHN23,parti} these density parameters were defined in term of $\Delta$ --- the \emph{globally known} maximum degree --- in this paper, we assume no such global knowledge and aim to use local list sizes; hence, we require a different notion of local sparsity/density. 
We adapt our definition of embedding to encompass each vertex's local view of its degree. Concretely, we color a multi-graph $H$ where each vertex uses one more color than it has incident edges. We call a vertex \emph{inaccurate} if its number of incident edges is a constant factor larger than its number of adjacent neighbors. Inaccurate vertices require special treatment, for they can skew estimates of local sparsity. Since we use a number of colors dependent on the number of incident edges while each neighbor blocks at most one color, inaccurate vertices are always guaranteed to have an abundance of free colors. After detecting them, we defer coloring inaccurate vertices to the very end of the algorithm.

\paragraph{The Upper Bound: Providing Enough Colors.}
Every sublogarithmic randomized coloring algorithm \cite{HSS18,CLP20,HKNT22} has three phases. First, they compute a partial coloring where each vertex has either \emph{low degree} or \emph{many excess colors compared to its uncolored number of neighbors.} Second, they use randomization and symmetry-breaking techniques to take advantage of this excess and color high-degree vertices ultrafast. Third, low-degree vertices are handled fast due to their low degree. In \cite{HSS18,CLP20,HKNT22}, the algorithm produces excess colors by a single-round randomized color trial. When vertices cannot access their palette \cite{ACK19,FHN23,FGHKN24}, they resort to approximations that require generating more excess colors in the densest regions on the graph. We follow the same general approach with some major modifications. First, the use of local-list size partially breaks the analysis of slack generation from \cite{HKMT21} (and the one of \cite{HKNT22} cannot be implemented fast on virtual graphs). Our main technical contribution is to provide sufficient assumptions for a color trial algorithm to generate enough excess colors even when vertices can have small lists (see \cref{sec:slack-generation}). In general, these added assumptions introduce substantial modifications to the accounting of colors throughout the algorithm (see \cref{sec:proof-accounting}).

\paragraph{The Upper Bound: Low-Degree Vertices.}
Contrary to previous work \cite{HKMN20,FHN23}, all high-degree --- larger than some $\poly(\log n)$ --- vertices are colored with high probability (rather than reducing uncolored degrees to $O(\log n)$). This implies that, for low-degree vertices, colors can be represented using $O(\log \deg) = O(\log\log n)$ bits. The algorithm for coloring low-degree nodes follows the shattering framework of \cite{BEPSv3}. First, vertices try random colors for $O(\log\log n)$ rounds. This reduces the uncolored parts of the graph to $\poly(\log n)$-sized components. Since nodes do not know their palette, we provide an algorithm for sampling colors likely-enough to succeed. Then, uncolored vertices learn a list of uncolored-degree+1 colors from their palette with an algorithm similar to a binary search. Finally, we simulate the deterministic algorithm of \cite{GK21} efficiently and complete the coloring. Our main contributions, our algorithms for sampling colors and learning palettes, are in \cref{sec:low-deg-sampling,sec:low-deg-learncolors}, respectively.

\subsection{Related Work}
Distributed coloring has been intensively studied. See, e.g., \cite{linial92,barenboimelkin_book,SW10,BEPSv3,disc16_coloring,fraigniaud16,HSS18,CLP20,MT20,HKM20,GK21,HKMT21,HKNT22,FK23} and references therein. The focus is usually on simple graphs, where the degree refers to the number of neighbors. The state-of-the-art \local algorithm for degree+1-coloring (in terms of $n$ only) is the $\tilde{O}(\log^2\log n)$-round algorithm obtained by plugging the $\tilde{O}(\log^2 n)$-round deterministic algorithm of \cite{GG23} into the shattering framework of \cite{HKNT22}. In \congest, authors of \cite{HNT22} show how to implement shattering with small messages; hence, using the $O(\log^3 n)$-round deterministic algorithm of \cite{GK21}, the resulting complexity is $O(\log^3\log n)$.
Besides degrees being defined slightly differently, results of \cite{HKNT22,GK21,GG23} are also more general in the sense that vertices can use \emph{any list} of degree+1-colors (not necessarily $\set{1, 2, \ldots, \deg(v)+1}$). Handling less constrained lists of colors in virtual graphs appears out of reach of current techniques; in fact, the problem has yet to be tackled in the simpler settings of cluster graphs and power graphs.

\paragraph{Virtual Graphs.}
Virtual graphs are ubiquitous in distributed graph algorithms and we make no attempt to be exhaustive. They refer to cases where the input graph differs from the communication network, though the formalism varies by use case. Here, we list occurrences of greatest relevance.
\begin{enumerate}
    \item Many algorithms modify the input graph --- e.g., by contracting an edge or removing a vertex and adding an edge between each neighbor --- throughout the execution. This happens, e.g., in \cite{GKKLP18,FGLPSY21}. In such cases, the algorithms embed the modified graph into the network while ensuring low congestion. Authors of \cite{ghaffari2016distributed,RozhonGHZL22,ALHZG_dc23} show that under some assumptions on the graph (e.g., planarity or excluded minor) then low-congestion shortcuts can be found efficiently, leading to drastic improvements on the round complexity.
    \item Recent network decomposition algorithms \cite{RG20,GGR20,GHIR23} compute clusters --- i.e., sets of vertices --- by growing increasingly large sets of vertices. Hence, computations are held through the three-step aggregation process described in \cref{sec:intro-virtual-graph}. That is, these algorithms are computing sequences of virtual graphs (including support trees) with $\poly(\log n)$ dilation and congestion.
    \item Finally, the local rounding framework introduced in \cite{Fischer17} and perfected in \cite{FGGKR23} runs a defective-coloring subroutine on virtual graphs. They describe d2-multigraphs, a special case of virtual graphs used to implement their algorithm in \congest. They care for parallel edges since they compute a coloring, like us. Besides, their rounding algorithm has been used by network decomposition algorithms \cite{GHIR23,GG23} and thus had to be implemented on virtual graphs.
\end{enumerate}
Our formalism for virtual graph captures all mentioned examples (with \& without congestion, with \& without parallel edges).

\paragraph{Scheduling \& Routing.}
Congestion and dilation are natural parameters in routing problems, where they measure the maximum overlap and length of the delivery paths of a set of packets. Scheduling, in this context, refers to organizing the packets' delivery along their paths, taking into account congestion constraints.
Naive scheduling leads to a $O(\congestion\dilation)$ delivery time, which can be hard to improve upon distributedly. 
Asymptotically optimal $\Theta(\congestion + \dilation)$ schedules exist and can be computed efficiently given global knowledge of the paths \cite{LMR_combinatorica94,LMR_combinatorica99}.

The routing literature is expansive and growing to this day \cite{LMRR_jal94,Lenzen_podc13,G15,GHZ_stoc21,HPRSZ_arxiv24}. While parallel delivery of information is crucial to our virtual graph algorithms, our problems are quite distinct from typical routing questions, as we usually aggregate and broadcast information rather than deliver it from a single source to a single target. In particular, we often change the information during its delivery.  
Even for our more complex tasks, a naive scheduling in $O(\congestion\dilation)$ remains possible.
We leave open the question of whether the $O(\congestion\dilation)$ dependency can be improved to $O(\congestion + \dilation)$ (see \cref{open-problem:scheduling} in \cref{sec:open} for more).

\paragraph{Power Graphs.}
Recently, there has been a growing interest in bandwidth-efficient algorithms for power graphs \cite{HKM20,HKMN20,BCMPP20,MPU23,FHN23,BG23}.
\cref{thm:rand-general} improves on previous work about distance-2 coloring \cite{HKM20,HKMN20,FHN23} by handling a more general problem (see \cref{sec:applications}), by reducing the number of colors used by each vertex to its pseudo-degree (rather than, say, using $\Delta^2+1$ colors which depends on a global parameter), and by improving the runtime by several $O(\log\log n)$ factors.

\paragraph{Other Models.}
The \textsf{Congested Clique} \cite{LPP06} can be viewed as a virtual graph model on the opposite end of the spectrum, where the communication graph is a clique. It has a $O(1)$-round deterministic algorithm for $\deg+1$-list-coloring \cite{CCDM_icalp23}, building on similar results for $\Delta+1$-coloring \cite{CFGUZ19,CDP21}. The supported \CONGEST model interpolates between \CONGEST and \textsf{Congested Clique} \cite{SS13}.

\paragraph{Sibling Paper.}
In a sibling paper \cite{parti}, we treat cluster graphs, a particular type of virtual graphs, focusing on high-degree graphs. We give a $O(\log^* n)$-round algorithm for $\Delta+1$-coloring cluster graphs when $\Delta = \Omega(\log^{21} n)$. A key technical contribution is coloring so-called put-aside sets in extremely dense subgraphs, which we build on in this paper. That paper introduces essential primitives that apply to general virtual graphs, particularly operations on the communication backbone, including broadcast, aggregation, and palette queries. It also contains a fingerprinting technique for approximating the sizes of neighborhoods.

\subsection{Outline of Paper}
In the next section, we describe the modeling of virtual graphs and show how they capture two important settings.
We present the main ideas behind the lower bound in \cref{sec:lower-bound-overview}. The high-level view of the algorithm is given in \cref{sec:high-level} along with key definitions, before describing some open questions in \cref{sec:open}. 

The detailed descriptions of various parts of the algorithm follow. We start with a result on slack generation, generalizing previous arguments to $deg+1$-colorings (of both sparse and dense nodes). The coloring of different parts of the graph is split into several sections: the dense-but-not-too-dense part in \cref{sec:color-non-cabals},  the low-degree nodes in \cref{sec:low-deg}, while the extremely-dense are in \cref{sec:color-cabals} as it builds heavily on the sibling paper \cite{parti}. The details of the lower bound are in \cref{sec:lower-bound}. Further appendices feature various algorithmic steps that are non-trivial adaptations or modifications of previous work, including almost-clique decomposition in \cref{sec:ACD}.

\subsection{Preliminaries}
\label{sec:prelim}

\paragraph{Mathematical Notation.}
For an integer $t \ge 1$, let $[t] \eqdef \set{1, 2, \ldots, t}$. For a function $f : \calX \to \calY$, when $X \subseteq \calX$, we write $f(X) \eqdef \set{f(x): x \in X}$; and when $Y \subseteq \calY$, we write $f^{-1}(Y) \eqdef \set{x\in X: f(x) \in Y}$. We abuse notation and write $f^{-1}(y) \eqdef f^{-1}(\set{y})$. For $X \subseteq \calX$, we write $f_{|X} : X \to \calY$ for the restriction of $f$ to $X$.
For an ordered family of sets $(S_1,\ldots,S_t)$, we write $S_{\geq i} = \bigcup_{j = i}^t {S}_{j}$ for the union of sets of index at least $i$.
For a real number $x \in \mathbb{R}$, let $x^+ = \max(x, 0)$ be the positive part of $x$. Throughout the paper, we hide overhead due to congestion $\congestion$ and dilation $\dilation$ by writing $\Ohat(f)$ for $O(\congestion \dilation \cdot f)$.

\paragraph{Graphs \& Multi-Graphs.}
A multi-graph $H=(V_H, E_H)$ is defined by a set of vertices $V_H$ and sets $E_H(u,v)$ describing all edges between $u$ and $v$ (and $E_H(u,v) = \emptyset$ if $u$ and $v$ are not adjacent). When each set $E_H(u,v)$ contains at most one edge ($H$ has no parallel edges), we say the graph is \emphdef{simple}.
The neighbors of $v$ in $H$ are $N_H(v) \eqdef \set{u\in V_H: E_H(u,v) \neq\emptyset}$. The \emphdef{pseudo-degree} of $v$ in $H$ is $\deg(v; H) \eqdef \sum_{u\in V_H} \card{E_H(u,v)}$, its number of incident edges. Its \emphdef{degree} counts its neighbors $|N_H(v)|$. When $H$ is clear from context, we drop the subscript and write $N(v) = N_H(v)$ and $\deg(v) = \deg(v; H)$. 
An unordered pair $\set{u,v} \subseteq V_H$ is called an \emphdef{anti-edge} or \emphdef{non-edge} if 
$E_H(u,v) = \emptyset$.

\paragraph{Colorings.}
For any integer $q \geq 1$, a \emphdef{partial $q$-coloring} is a function $\col : V_H \to [q]\cup\set{\bot}$ where $\bot$ means ``not colored''.
The domain $\dom \col \eqdef \set{v\in V_H: \col(v) \neq \bot }$ of $\col$ is the set of colored nodes. 
A coloring $\col$ is \emphdef{total} when all nodes are colored, i.e., $\dom\col =  V_H$; and we say it is \emphdef{proper} if $\bot \in \col(\set{u,v})$ or $\col(v) \neq \col(u)$ whenever $E_H(u,v)\neq \emptyset$.
We write that $\psi \succeq \col$ when a partial coloring $\psi$ \emphdef{extends} $\col$: for all $v\in \dom \col$, we have $\psi(v) = \col(v)$. 
The \emphdef{uncolored degree} $|N_\col(v)| \eqdef |N(v) \setminus \dom\col|$ of $v$ with respect to $\col$ is the number of uncolored neighbors of $v$. The \emphdef{uncolored pseudo-degree} $\deg_\col(v)$ of $v$ with respect to $\col$ counts its number of incident edges to uncolored neighbors.
The \emphdef{palette} of $v$ with respect to a coloring $\col$ is $L_\col(v) = [\deg(v)+1] \setminus \col(N(v))$, the set of colors we can use to extend $\col$ at $v$.

\section{Virtual Graphs}
\label{sec:model}
In distributed algorithmics, we consider \emphdef{communication graphs} or \emphdef{networks} $G=(V_G, E_G)$ where elements of $V_G$ are 
\emphdef{machines} that communicate by sending messages on incident \emphdef{links} --- unordered pairs of $E_G$ --- simultaneously in synchronous rounds. We assume machine $v\in V_G$ is provided a $O(\log |V_G|)$-bits unique identifiers $\ID_v$ to break symmetry. For randomized algorithms, they can also access local random bits. Messages are limited to $\bandwidth$ bits, where $\bandwidth$ is called the \emphdef{bandwidth} of the network. Unless stated explicitly, it is assumed that $\bandwidth=\Theta(\log |V_G|)$.

We define our notion of virtual graphs formally.
We shall always refer to the conflict graph by $H$ and to the communication graph by $G$. Vertices/nodes and edges refer only to elements of $H$, while machines and links are used for $G$.

\begin{definition}[Virtual Graph]
\label{def:virtual}
Let $G=(V_G, E_G)$ be a \emph{simple} graph. 
A virtual graph on $G$ is a \emph{multi}-graph $H=(V_H, E_H)$ where each vertex $v\in V_H$ is mapped to a set $V(v) \subseteq V_G$ of machines called the \emphdef{support} of $v$. Whenever two nodes are adjacent in $H$ their supports intersect, i.e., if $E_H(u,v) \neq \emptyset$ then $V(u) \cap V(v) \ne \emptyset$.
Each machine $w \in V_G$ knows the set $V^{-1}(w)$ of vertices whose supports contains it.
\end{definition}

When bandwidth is not an issue, we can work directly with the representation of \cref{def:virtual}. We can compute a breadth-first spanning tree $T(v) \subseteq E_G$ on each support $V(v)$ for distributing information, and then simulate a local algorithm on this support structure. 
With bandwidth constraints, we need to be more careful.

\begin{definition}[Embedded Virtual Graph]
\label{def:embedded}
Let $H$ be a virtual graph on $G$ such that $|V_H| \le \poly(|V_G|)$.
Suppose that (1) for each vertex $v\in V_H$, there is a tree $T(v) \subseteq E_G$ spanning $V(v)$; and (2) for each edge $e\in E_H(u,v)$ there is a machine $m(e) \in V(u) \cap V(v)$. 
We call $T(v)$ the \emphdef{support tree} of $v$ and $m(e)$ the machine \emphdef{handling} edge $e$.
Each machine $w$ knows the set of edges $m^{-1}(w)$ it handles as well as, for each incident link $\set{w,w'} \in E_G$, the set $T^{-1}(ww')$ of support trees it belongs to.
\end{definition}

Given support trees, it is convenient to design our algorithms as a sequence of rounds each consisting of a three-step process: broadcast, local computation on edges, followed by converge-cast.
We use two parameters to quantify the overhead cost of aggregation on support trees.
The \emphdef{congestion} $\congestion$ of $H$ is the maximum number of trees using the same link. The \emphdef{dilation} $\dilation$ is the maximum height of a tree $T(v)$ in $G$.
Formally,
\begin{equation}
\label{eq:def-congestion-dilation}
   \congestion \eqdef \max_{e\in E_G} |T^{-1}(e)|
\quad\text{and}\quad
   \dilation \eqdef \max_{v\in V_H} \parens*{\max_{u\in T(v)}\operatorname{dist}_{T(v)}(v, u)} \ .
\end{equation}
Congestion and dilation 
are natural bottlenecks for virtual graphs. 
In \cref{thm:lower-bounds}, we show that $\Omega(\congestion/\bandwidth + \dilation\log^* n)$ rounds are needed for our coloring task given $\bandwidth$ bandwidth in the communication graph.
Conversely, \cref{thm:rand-general} shows that coloring in $O(\congestion\dilation\cdot \log^4\log n)$ rounds is feasible for any embedding.

\begin{remark}
    A few remarks are in order.
\begin{enumerate}
    \item 
The degrees in $H$ can be computed as $\deg(v) = \sum_{w\in T(v)} |m^{-1}(w)|$ by aggregation on support trees, which is why we ask that edges have designated handlers.
    Counting exactly the number of distinct neighbors for all vertices appears to be challenging (i.e., requires $\Omegatilde(\card{N_H(v)})$ rounds).

    \item By running a BFS from a single source (or from multiple sources but in \emph{disjoint} subgraphs), we can count the exact the number of neighbors the source has. However, running this algorithm from multiple vertices creates congestion proportional to that number of vertices.

    \item 
It is, per se, easy to compute low-diameter support trees for all vertices, e.g., by BFS, but a poor selection of edges could easily lead to high congestion. It is an open question if trees of both low diameter and congestion can be computed efficiently (see \cref{sec:open}).

    \end{enumerate}
\end{remark}

\subsection{Implications}
\label{sec:applications}
Our framework captures several models and problems studied in the distributed graph literature. We review them quickly.

It is helpful to see the communication network $G=(V_G, E_G)$ through its \emphdef{subdivision graph}: the bipartite graph $S_G=(V_G, E_G, \set{\set{u,e}: u \in e \in E_G})$ with machines on the left, links on the right, and a link between $v \in V_G$ and $e \in E_G$ if and only if $v$ is an endpoint of $e$. Simulating a round of communication on $S_G$ takes one round of communication of $G$ (conversely, one round on $G$ takes two rounds of $S_G$). Defining our virtual graphs on $S_G$ rather than $G$ allows us to put conflict on links. We call the links of $S_G$ \emphdef{half-links}.
\footnote{
A common alternative representation is to stipulate that edges are between vertices with adjacent supports, i.e., $uv \in E_H$ implies that $\exists w \in V(v), x \in V(u)$ s.t.\ $wx \in E_G$. 
If we extend each support $V(v)$ in $G$ to include also the incident link nodes in $S_G$, then two supports in $S_G$ intersect whenever they are adjacent in $G$.  Hence, our formulation encompasses this variant.
}

\subsubsection{Application 1: Cluster Graphs}
A \emphdef{cluster graph} $\mathcal{C}$ on a communication graph $G=(V_G, E_G)$ is a graph where vertices are disjoint sets $C_x \subseteq V_G$ called \emphdef{clusters} with a designated machine $\mathsf{leader}(x)\in C_x$. Each cluster $C_x$ induces a connected graph of small diameter in $G$. Two clusters are adjacent if and only if they are connected by a link. A round of communication on $H$ consists of 1) broadcasting a $\bandwidth$-bit message from $\mathsf{leader}(x)$ to all machines in $C_x$, 2) communication on inter-cluster links, and 3) aggregate a $\poly\log n$-bit message (e.g., a sum or a min) to $\mathsf{leader}(x)$. They appear in several places, from maximum flow algorithms \cite{GKKLP18,FGLPSY21} to network decomposition \cite{RG20, GGR20}.

Clearly, cluster graphs are captured by our definition of virtual graphs: for $C_x\in V_H$, let $V(C_x)$ be $C_x$ plus the half-links going out of $C_x$ and $T(C_x)$ be a BFS tree spanning $V(C_x)$. \Cref{thm:rand-general} implies we can color cluster graphs fast:

\begin{corollary}
\label{cor:cluster}
Cluster graphs can be $\deg+1$-colored, w.h.p., in $O(\dilation \cdot \log^4\log n)$ \CONGEST rounds where $\dilation$ is the maximum (strong) diameter of a cluster, i.e., of $H[C_x]$ over all $C_x$.
\end{corollary}

In \cite{parti}, we show that $\Delta+1$-coloring high-degree cluster graphs (where $\Delta \leq \poly(\log n)$) can be done in $O(\log^*n)$ rounds. \cref{cor:cluster} is the first non-trivial distributed algorithm for degree+1-coloring cluster graphs.

\subsubsection{Application 2: Coloring Power Graphs}
For some integer $t \ge 1$, the \emphdef{$t$-th power graph} of $G$ is the graph $G^t$ on vertices $V_G$ where there is an edge $\set{u,v}$ when $\operatorname{dist}_G(u,v) \le t$. A \emphdef{distance-$t$ coloring} of $G$ is a coloring of $G^t$. Concretely, it is a coloring where nodes receive colors different from the ones in their $t$-hop neighborhood.
Our framework provides a unified view of distance-$t$ colorings: the same algorithm provides fast algorithms for all values of $t \ge 1$.
\footnote{
Another way to generalize both distance-1 and distance-2 coloring is through the \emph{relay model}. The communication graph $G$ is bipartite with vertices (of $H$) on one side and \emph{relay} nodes on the other side. Each vertex $v$ has as support tree all the incident edges in $G$. This is a star graph (i.e., of radius 1). The congestion on an edge is 1 and dilation is 2. All pairs of nodes whose support trees intersect form an edge in $H$. Thus, both the support trees and edge handling are implicitly given.
}

\begin{lemma}
\label{lem:distance-t}
Let $t \ge 1$ and $G=(V_G, E_G)$ be a graph with maximum (distance-1) degree $\Delta$. 
We can define a virtual graph $H=(V_H, E_H)$ on the subdivision graph $S_G$ of $G$ such that $V_H = V_G$ and a $\deg+1$-coloring of $H$ is a $\Delta^t+1$-coloring of $G^t$. 
Moreover, the congestion is $\congestion = O(\Delta^{\floor{\frac{t-1}{2}}})$, the dilation is $\dilation = t$, and the embedding is computable in $O(t\congestion)$ rounds.
\end{lemma}

\begin{proof}
For each node $v \in V_G$, its support tree $T(v)$ in $G$ is set to be the \emph{$t$-hop BFS tree in the subdivision graph $S_G$ rooted at $v$}.
For any pair $u,v\in V_H$, the edge set $E_H(u,v) = \emptyset$ if and only if $\dist_G(u,v) > t$. Otherwise, $E_H(u,v)$ contains an edge for each \emph{simple $uv$-paths} in $T(u)\cup T(v)$ in $G$. As there are at most $\sum_{i=1}^{t-1} \Delta (\Delta-1)^i \leq \Delta^t$ simple paths of length at most $t$ starting from $v$ in $G$. Hence, each vertex is incident to at most $\Delta^t$ edges in $H$. Thus, any $\deg+1$-coloring on $H$ is a distance-$t$ coloring of $G$ with $\Delta^t +1$ colors and from the definition of edges in $H$, a proper coloring on $H$ is also proper on $G^t$.

The bound on the dilation is immediate. To verify the congestion on a half-link $ev$, observe that there are at most $\Delta^{\floor{\frac{t-1}{2}}}$ nodes (of $G$) that are within distance $t$ of $v$ in $S_G$, and therefore at most that many support trees using that half-link.

We map each simple $uv$-path in $T(u) \cup T(v)$ to its middle machine in $S_G$. It is unique, as $S_G$ is bipartite and $u,v$ are on the same side. Each machine $w \in T(u) \cap T(v)$ knows its distances to $u$ in $T(u)$ and to $v$ in $T(v)$, and thereby knows if it is the middle machine.
To compute the embedding, we have each machine learn its distance-$t$ neighborhood in $S_G$, with the distance it has to each machine in it. This is done as follows: initially, each machine $v$ prepares a message $(\ID_v,1)$, which it sends to its $\Delta$ direct neighbors in $G$. Then, for each positive integer $i \leq \floor{\frac{t-1}{2}}$, each machine sends a message of the form $(\ID_u,i+1)$ to its direct neighbors for each message $(\ID_u,i)$ it has received, of which there are at most $\Delta^i$. Sending all messages for a fixed $i$ takes $O(\Delta^{\floor{\frac{t-1}{2}}}) = O(\congestion)$ rounds, hence a total $O(t\congestion)$ complexity. At the end of this process, each machine $v$ knows to which support trees $T(u)$ it belongs, and for each simple path of length at most $t$ in $S_G$ between two nodes $u,u'$ s.t.\ $v\in T(u)\cap T(u')$, $v$ knows whether it is its midpoint and should thus handle the edge.
\end{proof}

For any $t \ge 1$, \cref{thm:rand-general,lem:distance-t} imply that there is a distributed algorithm communicating on $S_G$ with $O(\log n)$ bandwidth that computes a $\Delta^t+1$-coloring of $G^t$. Since a round of communication on $S_G$ can be simulated in one round of communication on $G$, it shows the following corollary.

\begin{corollary}\
\label{cor:power-graphs}
For any $t\ge 1$, there is a randomized \congest algorithm for $\Delta^t+1$-coloring $G^t$ that runs in $O(t\Delta^{\floor{(t-1)/2}}\log^4\log n)$ rounds w.h.p.
\end{corollary}

Furthermore, the specific structure of power graphs allows for broadcast and aggregation over support trees to be done in only $O(\Delta^{\floor{(t-1)/2}}) = O(\congestion + \dilation)$ rounds instead of $O(\congestion\dilation) = O(t\Delta^{\floor{(t-1)/2}})$. The runtime in \cref{cor:power-graphs} can be improved to $O(\Delta^{\floor{(t-1)/2}}\log^4\log n)$ as a result.

It is not too difficult to see that --- by a reduction to set disjointness --- verifying an \emph{arbitrary (or random)} distance-$t$ coloring needs $\tilde\Omega(\Delta^{\floor{\frac{t-1}{2}}})$ rounds in \congest\cite{FHN20}. However, no super-constant lower bounds are known for computing distance-$t$ colorings in \congest when $t \ge 3$ and $\Delta \gg \log n$.

\section{Lower Bounds: Overview}
\label{sec:lower-bound-overview}
In this section, we sketch the main ideas behind our lower bound. 
We show the following theorem:

\TheoremLowerBound*

This implies as immediate corollary the same lower bound for the more general problem of $\deg+1$-coloring virtual graphs. The single statement is actually the concatenation of two independent lower bounds, one relative to congestion and bandwidth, and the other relative to dilation.

The dilation lower bound is straightforward, following directly from the seminal $\Omega(\log^*n)$ lower bounds on $3$-coloring~\cite{linial92,naor95}. We refer readers to \cref{sec:lower-bound-dilation}.

As the congestion lower bound makes lengthy use of technical tooling from communication complexity literature largely unrelated to the rest of the paper, we defer most details to \cref{sec:lower-bound}. Here, we chiefly describe the virtual graphs used for our lower bound and give intuition behind the complexity of coloring them.

\paragraph{Proof Structure of the Congestion Lower Bound.} 
The congestion-related part of our lower bound is obtained through a reduction from communication complexity. Our overall proof plan is as follows:
\begin{itemize}
    \item We introduce a 2-player communication complexity task in which said players must coordinate to $3$-color a $16$-node $2$-regular graph. Each player only knows the edges incident to half of the vertices and is in charge of outputting half of the coloring.
    \item We show that this task is nontrivial, and in particular, that it has $\Omega(1)$ information complexity, a complexity measure which lower bounds communication complexity.
    \item From known direct-sum results on information complexity, we get that solving $\congestion/8$ independent copies of the task has information complexity $\Omega(\congestion)$. \item We introduce a virtual graph of congestion $\congestion$ and constant dilation in which we embed  $\congestion/8$ instances of the task, i.e., $\deg+1$-coloring the virtual graph solves the $\congestion / 8$ instances.
    \item We observe: any $T$-round algorithm for $\deg+1$-coloring virtual graphs over communication graphs with congestion $\congestion$ given bandwidth $\bandwidth$ implies a $O(T \bandwidth)$ communication complexity algorithm for solving $\congestion/8$ copies of the nontrivial task.
    \item We conclude: the round complexity $T$ of any such distributed algorithm for $\deg+1$-coloring must necessarily be at least $\Omega(\congestion / \bandwidth)$.
\end{itemize}

\paragraph{The Communication Complexity Gadget.} 
We define the communication complexity task we use in \cref{def:congestion-gadget-task}. While this definition is a generalization with an arbitrary even number of nodes on both sides, for our purposes, we will only use the gadget with the parameter $k$ set to $k=4$, i.e., with $8$ nodes on Alice and Bob's sides. See \cref{fig:gadget-graph} for a illustration of our gadget.

\begin{restatable}[Matching 3-coloring task]{definition}{MatchingGadgetTask}
\label{def:congestion-gadget-task}
    In the $\MCOL_k$ task, a $4k$-node graph is initially uncolored. Its nodes are split into two equal parts -- left and right -- given to Alice and Bob. Alice and Bob receive a perfect matching over their respective sets of $2k$ nodes. For each $i \in [2k]$, the $i$th left node is connected to the $i$th right node. At the end of the communication protocol, Alice must output a color in $\set{1,2,3}$ for each left node, and Bob must do the same for the right nodes, such that the coloring is valid with respect to the graph received as input.
\end{restatable}

\begin{figure}[ht]
    \centering
    \includegraphics[page=2,width=0.33\linewidth]{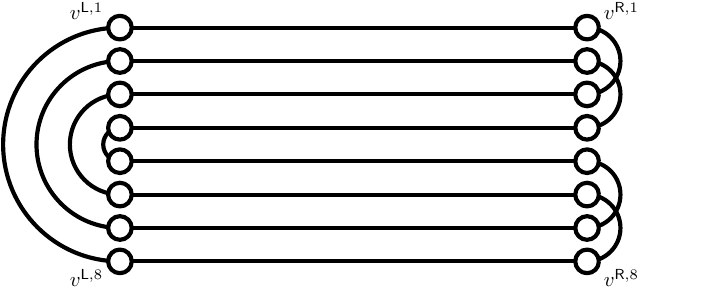}%
    \includegraphics[page=3,width=0.33\linewidth]{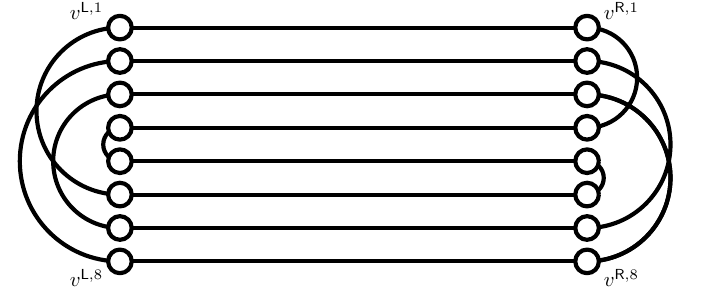}%
    \includegraphics[page=4,width=0.33\linewidth]{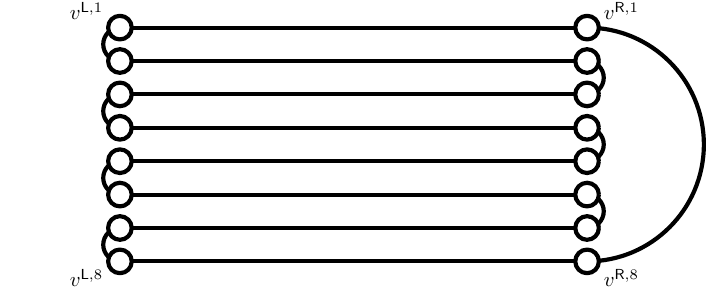}%
    \caption{Three possible inputs to the communication complexity task.}
    \label{fig:gadget-graph}
\end{figure}

The crux of the argument is to show that the $\MCOL_4$ task cannot be solved without communication. This can be intuitively seen by noticing that there can be at most $3$ nodes on which Alice \emph{always} outputs the same color regardless of her input matching (same on Bob's side). Indeed, as there are only $3$ colors, a fourth node with a fixed color means that two nodes would receive the same color on Alice's side regardless of her matching. This implies an error when said two nodes are connected in Alice's matching. Generalizing this idea to randomized algorithms allows us to show that an algorithm without communication necessarily makes an error with some constant probability
\footnote{
This intuition also explains why we take gadgets with $8$ nodes on each side and not less: a smaller gadget would be solvable without communication by fixing the color of (up to) $3$ nodes on each side.}
.

\begin{restatable}{lemma}{NoCommunicationError}
    \label{lem:gadget-no-communication-error}
    Any zero communication protocol for $\MCOL_4$ fails with probability at least $\frac{1}{196}$ over the uniform input distribution.
\end{restatable}

\paragraph{Embedding the Gadget.}
Embedding the gadgets into a virtual graph is then done with the following communication network: we consider two stars (depth-1 trees) with $\congestion$ leaves; we connect the two stars by a single link between their roots $\wLcom$ and $\wRcom$. The support of each node on the left is made of an edge of the left star with the central edge, while the support of each node on the right is just an edge in the right star. $\wLcom$ handles the edges in the left matching, while $\wRcom$ handles the edges of the right matching as well as the edges between the left and right sides of the virtual graph.
See \cref{fig:congestion-lb-graph} for an illustration.

\begin{figure}[t]
    \centering
    \includegraphics[page=3,width=0.33\textwidth]{congest_lb_virtual.pdf}\includegraphics[page=1,width=0.33\textwidth]{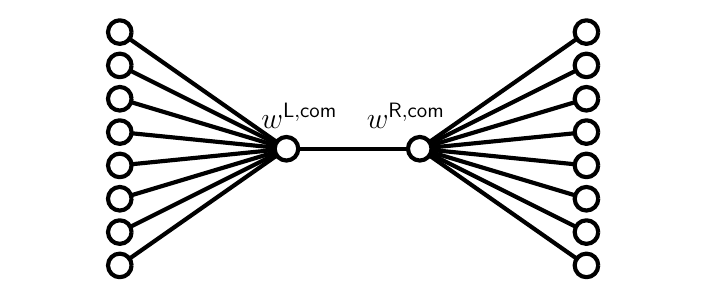}\includegraphics[page=8,width=0.33\textwidth]{congest_lb_communication.pdf}\caption{Examples of a virtual graph $H$ with a single gadget  (left), a communication network $G$ (middle) in which $H$ can be embedded, and the support of the top left virtual node (right).}
    \label{fig:congestion-lb-graph}
\end{figure}

The proof of \cref{lem:gadget-no-communication-error}, with its implication for the information complexity of the task, and ultimately, our $\Omega(\congestion/\bandwidth)$ lower bound for $3$-coloring graphs of degree $2$ (\cref{thm:lower-bounds}), are all deferred to \cref{sec:lower-bound-congestion}.

\section{Coloring Algorithm}
\label{sec:coloring-alg}
\label{sec:coloring}
The goal of this section is to present the main technical ideas behind \cref{thm:rand-general}.

\ThmRandGeneral*

We give necessary definitions and self-contained statements of each of the main steps of our algorithm. 
First, we discuss the concept of \emph{slack} and present the means by which we measure and produce it. We then introduce a version of the sparse-dense decomposition tailored to our needs. Finally, we describe the main steps of our algorithm.
Proofs for the most technical parts are deferred to later sections (slack generation in \cref{sec:slack-generation} and coloring of non-cabals in \cref{sec:color-non-cabals}) to preserve the flow of the paper.
Some of the preliminary results, such a randomized color trial, are given in \cref{sec:background}.

\subsection{Slack}
Intuitively, the slack measures how easily a vertex gets colored. More formally, it is used to bound from below the number of colors available to a vertex when its neighbors are trying to get colored.
There are several types of slack that occur.

\paragraph{Savings.}
Whenever a neighbor uses a color that is either outside $v$'s palette or the same color as another neighbor, then $v$ \emph{saves} a color. Under a given partial coloring $\col$, this is quantified by
the \emphdef{savings slack} of $v$ from coloring $\col$:
\begin{equation}
    \label{eqdef:savings}
    \xi_\col(v, S) \eqdef |S \cap \dom\col| - |\col(S) \cap [\deg(v)+1]|
\end{equation}
We write $\xi_\col(v)$ for $\xi_\col(v, N(v))$.

\paragraph{Redundancy.}
In degree+1-coloring (unlike $\Delta+1$-coloring), slack can also occur when $v$ has a shortage of neighbors with a high enough degree. We measure this with the \emphdef{redundancy} of $v$
defined as
\begin{equation}
    \label{eqdef:redundancy}
    \red_v \eqdef \max_{t \leq |N(v)|/12} \card{N_H(v)} - t - |\set{u \in N(v): \deg(u) + 1 > t}|\ .
\end{equation}
In other words, there is a $t \leq |N(v)| / 12$ such that, even if all high-degree neighbors (larger than $t$) use different colors, at least $\red_v$ colors remain available to $v$.

\paragraph{Inaccuracy}. The difference between the palette size and the number of neighbors is the \emphdef{inaccuracy} of the node:
\begin{equation}
    \label{eqdef:inaccuracy}
     \delta_v = \deg(v) - |N_H(v)| \ . 
\end{equation}
In our setting, this is caused by parallel edges. 
A vertex with $\delta_v > \delta$ is called \emphdef{$\delta$-inaccurate} and \emphdef{$\delta$-accurate} otherwise.

\paragraph{Permanent \& Temporary slack}. 
The aforementioned forms of slack (savings, redundancy, and inaccuracy) are \emphdef{permanent}, meaning that they do not decrease as we extend the coloring. Another way to provide slack to a vertex is by keeping some of its uncolored neighbors inactive. This artificially reduces degrees -- thus contention -- without reducing the number of available colors, thereby providing slack. This is called \emphdef{temporary slack} as it perishes when we eventually color the inactive neighbors.

\paragraph{Slack Generation.}
While redundancy and inaccuracies do not depend on the coloring, vertices get savings only if we manage to same-color its neighbors.
We show in \cref{sec:slack-generation} that a classic one-round algorithm of ``trying a random color'' creates enough slack for deg+1-colorings. This generalizes results for $\Delta+1$-coloring \cite{Reed98,EPS15,HKMT21}. 
It also generalizes a method of \cite[Lemma 4.10]{AA20} for deg+1-coloring that applies to the sparse and uneven nodes (assuming $\deg(v) = |N_H(v)|$).
\slackgeneration creates color savings probabilistically. The savings expected from a random color trial are measured by the unevenness and sparsity, which we now define.

The savings we expect due to high-degree neighbors using colors beyond $\deg(v)+1$ is captured by the \emphdef{unevenness} of $v$. Within a subgraph induced by a set $S \subseteq V_H$, it is defined as
\begin{equation}
        \label{eqdef:unevenness}
        \discr_v(S) \eqdef \sum_{u\in S} \frac{[\deg(u)+1] \setminus [\deg(v)+1]}{[\deg(u)+1]} 
    = \sum_{u\in S} \frac{(\deg(u) - \deg(v))^+}{\deg(u)+1} \ . 
    \end{equation}
    We write $\discr_v = \discr_v(N(v))$ for succinctness.
A vertex such that $\discr_v > \eta$ is called \emphdef{$\eta$-uneven} and \emphdef{$\eta$-balanced} otherwise.

    The savings we expect from colors reused by multiple neighbors
    is quantified by 
    the \emphdef{sparsity} of $v$. 
The sparsity of $v$ is defined as
    \begin{equation} 
        \label{eqdef:sparsity}
\spar_v \eqdef \frac{1}{|N_H(v)|} \parens*{\binom{|N_H(v)|}{2}
        - \frac{1}{2} \sum_{u \in N_H(v)} \card{N_H(u) \cap N_H(v)} } \ . 
    \end{equation}
    Note that $\frac{1}{2} \sum_{u \in N_H(v)} \card{N_H(u) \cap N_H(v)}$ counts the number of edges in $N_H(v)$ \emph{without multiplicity}, even if $H$ is not simple. Hence $\spar_v \cdot \card{N_H(v)}$ counts the number of edges missing in $N_H(v)$, without multiplicity. 
A vertex such that $\spar_v > \spar$ is called \emphdef{$\spar$-sparse} and \emphdef{$\spar$-dense} otherwise.

\newcommand{\CSlackNode}{\gamma_{\ref{lem:slack-generation-node}}}
\begin{restatable}[Slack Generation]{lemma}{LemSlackGen}
\label{prop:slack-generation}
\label{lem:slack-generation-node}
Let $\Vsg \subseteq V_H$ and let $\colsg$ be the coloring produced by running \cref{alg:slack-generation} in $H[\Vsg]$ avoiding colors $\leq r$.
Let $v \in \Vsg$ be a node satisfying
$\deg(v)\le 3|N_H(v)|/2$, 
$|N(v) \setminus \Vsg| < (\zeta_v+\discr_v)/4$, 
$\zeta_v \ge 48r$, and $\red_v \le (\spar_v+\discr_v)/12$.
Then
\[ 
    \xi_{\colsg}(v) \ge \CSlackNode \cdot (\spar_v + \discr_v)
    \quad\text{with probability}\quad
    1- \exp \parens*{-\Theta(\spar_v+\discr_v)}
\]
\end{restatable}

\subsection{Almost-Clique Decomposition}
\label{sec:background-ACD}
In \cref{prop:AC}, we describe a structural decomposition partitioning vertices according to their ability to receive slack and of which type. All sub-logarithmic distributed coloring algorithms \cite{HSS18,CLP20,HKNT22,parti} use such a decomposition. We adapt \cite{AA20} to account for inaccuracies in degrees (\cref{part:ACD-accurate}).
\cref{prop:AC} partitions vertices into high- and low-degree vertices based on the threshold $\Deltalow = \Theta\parens*{\log^{21} n}$. Each requires a different approach and, in particular, we do not need to argue that low-degree vertices receive slack.
We prove \cref{prop:AC} in \cref{sec:ACD} to preserve the flow of the paper.

\newcommand{\Cstar}{\gamma_{\ref{prop:AC}}}
\begin{restatable}{lemma}{LemACD}
    \label{prop:AC}
    There exists an algorithm that,
    for any multi-graph $H=(V_H,E_H)$ and $\epsilon \in (0,1/100)$, computes in $\Ohat(1/\epsilon^6)$ rounds an $\epsilon$-\emphdef{almost-clique decomposition}: a partition $V_H = \Vlow \cup \Vhigh$ and $\Vhigh = \Vin \cup \Vstar \cup \Vdense$ such that 
    \begin{enumerate}
        \item \label[part]{part:ACD-low}
        each $v\in \Vlow$ has $\deg(v) \leq 2\Deltalow$ and $v\in \Vhigh$ has $\deg(v) \geq \Deltalow$;
        \item \label[part]{part:ACD-accurate}
        each $v\in \Vin$ is $\Omega(\eps^3 |N(v)|)$-inaccurate and each $v \in \Vhigh \setminus \Vin$ has $\deg(v) \leq (1 + \eps^3)|N(v)|$;
        \item \label[part]{part:ACD-star}
        each $v\in \Vstar$ has $\spar_v + \discr_v + |N(v) \cap \Vin| \geq \Cstar \cdot \deg(v)$ for a constant $\Cstar = \Cstar(\eps)\in(0,1)$; 
        \item \label[part]{part:ACD-dense}
        $\Vdense$ is partitioned into $\epsilon$-almost-cliques: sets $K \subseteq \Vhigh$ such that
        \begin{enumerate}
            \item $|N_H(v) \cap K| \ge (1-\epsilon)|K|$, for each $v\in K$,
            \item $\deg(v) \le (1+\epsilon)|K|$, for each $v\in K$, and
            \item\label[part]{part:ACD-ext} $|N_H(v) \setminus K| \leq O_\eps(\spar_v + \discr_v + |N(v) \cap \Vin|)$.
        \end{enumerate}
    \end{enumerate}
    
\end{restatable}

Let $\Delta_K \eqdef \max_{v\in K} \deg(v)$.
From \cref{prop:AC}, it holds for each almost-clique $K$ that $(1-\epsilon)|K| \leq \Delta_K \le (1+\epsilon)|K|$, and that for every $v\in K$, $\deg(v) \ge (1-2\epsilon)\Delta_K$. Every pair of vertices in $K$ has $(1-2\epsilon)|K|$ neighbors in common in $K$, and hence $H[K]$ has (strong-)diameter at most two.

For $v \in \Vdense$, let $K_v$ denote the almost-clique containing $v$. We denote by $A_v=K_v\setminus N_H(v)$ its \emphdef{anti-neighborhood} and by $a_v = |K_v| - \deg(v, H\cap K_v)$ its \emphdef{pseudo-anti-degree}. 
We call $E_v = N_H(v) \setminus K_v$ the \emphdef{external-neighborhood} and $e_v = \deg(v,H\setminus K_v)$ its \emphdef{pseudo-external-degree}. 
Importantly, \emph{pseudo-external and pseudo-anti-degrees count multiplicities of edges in the conflict graph}. 
We split the contribution to $\delta_v$ (\cref{eqdef:inaccuracy}) between external and internal neighbors:
\begin{equation}
\label{eq:def-degree-slack}
    \delta_v = \delta_v^e + \delta_v^a,\quad \text{where}\quad 
    \delta_v^e \eqdef e_v - |E_v|\quad \text{and}\quad 
    \delta_v^a \eqdef |A_v| - a_v \ .
\end{equation}
For almost-clique $K$, we denote average values by
$a_K = \sum_{v\in K} a_v/|K|$ and $e_K = \sum_{v\in K} e_v/|K|$.
We state the following facts for future use.
\begin{fact}
    For every $v \in \Vdense$,
    $|E_v| \leq 2\eps|K|$ and $|A_v| \leq \eps|K|$. Thus $e_v \leq 2\eps |K|$ and $a_v \leq \eps|K|$.
\end{fact}
\begin{fact}\label{fact:count-degree}
For every $v \in \Vdense$, $\deg(v)+1 = |K| + e_v - a_v$.
\end{fact}
\begin{proof}
$\deg(v) = \delta_v + |N_H(v)| = \delta_v + |N_H(v) \cap K_v| + |E_v| = \delta_v + |K| + |E_v| - |A_v| - 1$.
\end{proof}

\subsection{The High-Level Algorithm}
\label{sec:high-level}
We can now describe the main steps of our algorithm. At high level, we compute the decomposition of \cref{prop:AC}, run slack generation in $\Vsg = \Vhigh \setminus (\Vcabal \cup \Vin)$ and color each part of the decomposition in a precise order. Necessary conditions and guarantees for each step of \cref{alg:d1-col} are given in the corresponding propositions.

\begin{algorithm}
    \caption{The $\deg+1$-coloring algorithm\label{alg:d1-col}}
    
    \computeACD \label[line]{line:ACD} \hfill (\cref{prop:AC})

    \slackgeneration in $\Vsg = \Vhigh \setminus (\Vcabal \cup \Vin)$ 
    without using colors $[r]$ \label[line]{line:SG}
    \hfill (\cref{prop:slack-generation})

    \alg{ColoringVstar} without using colors $[r]$ 
    \label[line]{line:d1-col-star}
    \hfill (\cref{prop:color-sparse})

    \ColoringNonCabals
    \label[line]{line:d1-col-non-cabals}
    \hfill (\cref{prop:non-cabals,sec:color-non-cabals})

    \ColoringCabals
    \hfill (\cref{prop:cabals})

    \alg{ColoringInaccurate}
    \hfill (\cref{prop:inaccurate})

    \alg{ColoringLowDegree} \label[line]{line:low-degree}
    \hfill (\cref{prop:low,sec:low-deg})
\end{algorithm}

\paragraph{Parameters.}
Let $C_1$ be some large universal constant. Let us set the following parameters
\begin{equation}
    \label{eq:params}
    \eps = 1/2000 \ ,\quad
    \ell = C_1 \cdot \log^{1.2} n \ , \quad
    \text{and}\quad
    r = C_1 \cdot \log^{1.1} n \ ,
\end{equation}
where $\ell$ is chosen to asymptotically dominate\footnote{We could set $\ell = \Theta(\log^{1.1} n)$ for some fine-tuned hidden constant. We set $\ell$ to be larger to ensure that for any constant $\gamma > 0$ we have $\gamma \ell \geq r$ for a large enough $n$.} $\Theta(\log^{1.1} n)$, which is the minimum palette size for \refMCT, and $r$ sets the number of reserved colors.
We call colors from $[r] = \set{1, 2, \ldots, r}$ reserved because we use them
exclusively for multicolor trials.
Let $\Vlow, \Vin, \Vstar, \Vdense$ be an $\eps$-almost-clique decomposition of the high-degree vertices.
We define
\[
    \Kcabal = \set{K: e_K < \ell} \ ,\quad
    \Vcabal = \set{v \in \Vdense: K_v \in \Kcabal  } 
    \quad\text{and}\quad
    \Vsg = \Vhigh \setminus (\Vcabal \cup \Vin) \ .
\]

After running \refSlackGen in $\Vsg$, w.h.p., all the vertices in $\Vstar$ have enough slack to get colored by \refMCT.
\cref{prop:color-sparse} achieves this coloring with additional post-conditions necessary for coloring non-cabals (\cref{prop:non-cabals}). In words, we extend the coloring $\colsg$ produced by slack generation such that $\Vstar$ is totally colored, the coloring in $V_H \setminus \Vstar$ coincides with $\colsg$ and reserved colors are not used (not even in $\Vstar$). 

\begin{proposition}[Coloring $\Vstar$]
    \label{prop:color-sparse}
    Suppose $\colsg$ is the coloring produced by slack generation.
    In $\Ohat(\log^*n)$ rounds, we compute $\col \succeq \colsg$ such that, w.h.p., we have $\Vstar \subseteq \dom\col$, $\col_{|V_H \setminus \Vstar} = {\colsg}_{|V_H \setminus \Vstar}$ and $\col(V_H) \cap [r] = \emptyset$. 
\end{proposition}

\begin{proof}
    Let $\calC(v) = [\deg(v) + 1] \setminus [r]$ for all $v \in \Vstar$. 
    First, we argue there exists a universal constant $\gamma = \gamma(\eps) \in (0,1)$ such that
    \begin{equation}
        \label{eq:slack-sparse}
        \text{for all $\col \succeq \colsg$ such that $\dom\col \subseteq \Vsg$,}\ 
        \text{each $v\in \Vstar$ has }
        |L_\col(v) \cap \calC(v)| \geq \gamma \cdot \deg(v) \ .
    \end{equation}
    Fix $v$ and $\col$ such as described in \cref{eq:slack-sparse}.
    Observe that $|L_\col(v)| \geq |N(v) \setminus \Vsg|$ because nodes outside $\Vsg$ remain uncolored. The other case being trivial, we now assume $|N(v) \setminus \Vsg| \leq \Cstar/10 \cdot \deg(v)$. Since $\Vin \cap \Vsg = \emptyset$, we have that $|N(v) \cap \Vin| \leq |N(v) \setminus \Vsg|$ and thus $\spar_v + \discr_v \geq (9/10)\Cstar \cdot \deg(v)$ by \cref{part:ACD-star}.
By definition of redundancy (\cref{eqdef:redundancy}), $|L_\col(v)| \geq \red_v$, hence we may assume $\red_v \leq (\spar_v + \discr_v) / 12$. Finally,  by \cref{prop:slack-generation} (recall $v\in \Vstar$ is accurate), w.h.p., $v$ saves $\xi_{\colsg}(v) \geq \Omega_\eps(\spar_v + \discr_v)$ colors after slack generation. In any case, the number of available colors is $|L_\col(v)| \geq \Omega_\eps(\spar_v + \discr_v) \geq \Omega_\eps(\deg(v))$. \cref{eq:slack-sparse} follows from $\Vstar \subseteq \Vhigh$ and our choice of $r \ll \Deltalow$.

    The algorithm coloring $\Vstar$ has two phases. First, vertices try random colors in $\calC(v)$ for $T = O(\gamma^{-4}\log \gamma^{-1}) = O(1)$ iterations. W.h.p., uncolored degrees in $H[\Vstar]$ decrease by a $O(1 - \Theta(\gamma^4))$ factor each trial as one can verify that \cref{eq:slack-sparse} implies the hypotheses of \cref{lem:try-color}. At the end of this phase, w.h.p., all vertices have uncolored degree $\deg_\col(v, \Vstar) \leq \gamma / 4 \cdot \deg(v)$.
    In particular, \cref{part:mct-slack} of \cref{lem:mct} is verified and we complete the coloring of $\Vstar$ in $\Ohat(\log^*n )$ rounds with high probability using MultiColor Trial.

    Let us conclude the proof by verifying the properties claimed in \cref{prop:color-sparse}.
    We only color vertices of $\Vstar$, hence $\col_{|V_H \setminus \Vstar} = {\colsg}_{|V_H \setminus \Vstar}$. Further, sets $\calC(v)$ do not include reserved colors and hence we do not use them. Since slack generation does not use them either, $\col(V_H) \cap [r] = \emptyset$.
\end{proof}

Non-cabal dense vertices are colored by \cref{alg:non-cabals} in \cref{sec:color-non-cabals}.
They are colored immediately after coloring $\Vstar$, and the conditions needed for \cref{prop:non-cabals} follow from those guaranteed by \cref{prop:color-sparse}. The algorithm combines primitives from various recent randomized coloring algorithms \cite{HKNT22,FGHKN23,FHN23} (and \cite{parti}), all needing non-trivial adaptation to the current setting. 
Instead of applying MultiColor Trials directly after the synchronized color trial, we use the slower $O(\log\log n)$-round Slice Color algorithm of \cite{FHN23} to find an orientation where nodes have $O(\log n)$ uncolored out-neighbors. This allows us to use a fixed number of only $r = \Theta(\log^{1.1} n)$ reserved colors in the final application of MultiColor Trials, simplifying the (already intricate) treatment. Finally, a significant effort is needed to add up all sources of slack and show that dense vertices always have enough colors in the clique palette (see \cref{sec:proof-accounting}).

\begin{restatable}[Coloring Non-Cabals]{proposition}{PropColoringNonCabals}
    \label{prop:non-cabals}
    Suppose $\col$ is a coloring such that $\dom\col \subseteq \Vsg$, $\col_{|\Vdense} = {\colsg}_{|\Vdense}$ and $\col(V_H) \cap [r] = \emptyset$.
    In $\Ohat(\log\log n \cdot\log^* n)$ rounds, we color all vertices in $\Vdense \setminus \Vcabal$.
\end{restatable}

Cabals are colored by \cref{alg:coloring-cabals} in \cref{sec:color-cabals}.
The approach to color $\Vcabal$ is similar to \cref{prop:non-cabals} except for two major differences.
First, vertices do not receive slack from slack generation, so we instead resort to \emph{put-aside sets} \cite{HKNT22}. Second, coloring put-aside sets requires a different approach 
that was developed in \cite{parti}.
\begin{restatable}[Coloring Cabals]{proposition}{PropCabals}
    \label{prop:cabals}
    Suppose $\col$ is a coloring such that $\Vcabal \cap \dom\col = \emptyset$.
    Then, there exists a $\Ohat(\log\log n \cdot \log^* n)$-round algorithm coloring all nodes in $\Vcabal$ with high probability.
\end{restatable}

The inaccurate nodes have enough
slack regardless of the coloring (\cref{eqdef:inaccuracy}) and
are easily colored at the end in the same way as $\Vstar$.
\begin{proposition}[Coloring Inaccurate Nodes]
    \label{prop:inaccurate}
  We can color all vertices in $\Vin$ in $\Ohat(\log^* n)$ rounds.
\end{proposition}

\begin{proof}[Proof Sketch]
The inaccuracy means that 
each vertex $v$ in $\Vin$ has $\Omega(\eps^3 \deg(v))$ colors available in $[\deg(v)+1]$ under any (possibly partial) coloring. 
Like for $\Vstar$, we  color $\Vin$ with $O(\eps^{-12}\log \eps^{-1}) = O(1)$ iterations of \refRCT and $O(\log^* n)$ iterations of \refMCT where $\calC(v) = [\deg(v)+1]$ and $\gamma = \Theta( \eps^3 )$.
\end{proof}

Low-degree nodes are colored in \cref{sec:low-deg}.

\begin{restatable}[Coloring Low-Degree Nodes]{proposition}{PropLow}
    \label{prop:low}
    Suppose $\col$ is a coloring such that $\Vhigh = V_H \setminus \Vlow = \dom\col$.
    In $\Ohat(\log^4\log n)$ rounds, we compute a total coloring of $H$.
\end{restatable}

\begin{proof}[Proof of \cref{thm:rand-general}]
    By \cref{prop:AC}, we compute the $\eps$-almost-clique decomposition in $\Ohat(1/\eps^6) = \Ohat(1)$ rounds.
    Running Slack Generation takes $\Ohat(1)$ rounds (see \cref{alg:slack-generation}).
    By \cref{prop:color-sparse,prop:non-cabals,prop:cabals}, we extend the coloring to all vertices of $\Vhigh$ in $\Ohat(\log\log n\cdot\log^*n)$ rounds.
    By \cref{prop:low}, low-degree vertices are colored in $\Ohat(\log^4\log n)$ rounds.
    Overall, the round complexity is dominated by the coloring of low-degree vertices.
\end{proof}

\section{Slack Generation}
\label{sec:slack-generation}
We use the following algorithm for generating permanent slack in all parts of the graph except cabals.
This section contains the analysis of the amount of slack generated by \cref{alg:slack-generation}.

\begin{algorithm}[H]
    \caption{\slackgeneration\label{alg:slack-generation}}
    Each $v \in \Vsg$ joins $\Vactive$ w.p.\ $\pg=1/20$.

    Each $v\in \Vactive$ samples $c(v) \in \{r+1, r+2, \ldots, \deg(v) + 1 \}$ uniformly at random. 
    
    Let $\colsg(v) = c(v)$ if $v\in \Vactive$ and $c(v) \notin c(N_H^+(v))$. Otherwise, set $\colsg(v) = \bot$.
\end{algorithm}

Recall that $N_H^+(v) = \set{u\in N(v): \deg(u) \geq \deg(v)}$ are the neighbors of larger or equal degree.
Recall the definitions of savings, redundancy, unevenness and sparsity given in \cref{eqdef:savings,eqdef:redundancy,eqdef:unevenness,eqdef:sparsity}.
\cref{alg:slack-generation} uses only some of the vertices (activated randomly) to ensure it does not color too many vertices in each almost-clique. We state the following obvious fact for future use.
\begin{fact}
  \label{fact:colsg-AC}
  In each almost clique, w.h.p., at most $|K|/10$ vertices are colored by $\colsg$.
\end{fact}

\begin{lemma}
    Let $S \subseteq \Vsg$ and $v\in\Vsg$ such that $\deg(v) \geq 2r$.
Then $\xi_{\colsg}(v, S) \geq \gamma_{\ref{lem:gen-uneven}} \cdot \discr_v(S)$ w.p.\ $1-\exp\parens*{-\Omega(\discr_v(S))}$.
\label{lem:gen-uneven}
\end{lemma}
\begin{proof}
    The bound follows from \cite[Lemma 10]{HKNT21-arxiv}. This can be verified by observing that unevenness in deg+1-coloring falls under their definition of ``balanced discrepancy'' for the more general deg+1-list coloring. Only nodes with a palette size of at least $\deg(v)+1$ contribute to the unevenness of $v$ and therefore in their notation $\bar{\discr}_v = \bar{\discr}_v^{\mathsf{bal}}$.
The condition on the minimum degree ensures that picking colors from the range $[\deg(v)+1]\setminus [r]$ only changes the constant factor.
\end{proof}

The following lemma is stated more generally than necessary, for possible future use.
\newcommand{\hatzeta}{\hat{\spar_v}}
\newcommand{\hateta}{\hat{\discr_v}}
Let 
$$
\hatzeta =  \frac{1}{|N_H(v)\cap \Vsg|} \parens*{\binom{|N_H(v) \cap \Vsg|}{2} - \frac{1}{2}\sum_{u\in N_H(v)} \card{N_H(u)\cap N_H(v) \cap \Vsg}} 
$$
and
$$
\hateta \eqdef \sum_{u\in N_H(v)\cap \Vsg} \frac{[\deg_{\Vsg(u)}+1] \setminus [\deg_{\Vsg(v)}+1]}{[\deg_{\Vsg(u)}+1]}
$$ 
be the values of $\spar_v$ and $\discr_v$ restricted to the subgraph $\Vsg$.

\LemSlackGen*

\begin{proof}
The condition on $|N(v)\setminus \Vsg|$ implies that $\hatzeta + \hateta \ge \spar_v - |N(v)\setminus \Vsg| + \discr_v - |N(v)\setminus \Vsg| \ge (\spar_v + \discr_v)/2$.
We shall prove the desired bounds in terms of $\hatzeta + \hateta$, from which the lemma follows.

When $\hateta > \hatzeta$, the claim follows from \cref{lem:gen-uneven}.
For a bound in terms of sparsity, we consider several cases.

A node cannot have too many low-degree neighbors, by the assumption about redundancy.
If it has many high-degree neighbors, we get slack by unevenness. Otherwise, we proceed along a familiar path \cite{AA20,HKNT21-arxiv} by counting the number of same-colored pairs and ultimately applying Talagrand's concentration inequality. The main catch is to exclude the counting of such pairs that involve low-numbered colors. This is necessary to guarantee the Lipschitz property needed for Talagrand.

Let $\beta = \hatzeta/12$ and let $R_u$ denote the range of colors $[\beta, \deg(u)+1]$.
By assumption, $\beta \ge 2r$ and thus $|R_u| \ge (\deg(u)+1)/2$.

Partition the set $N(v)\cap \Vsg$ of neighbors of $v$ in $\Vsg$ according to their (pseudo-)degree:
\begin{itemize}
    \item $\Nlo$, below $2\beta$; 
\item $N$, in the range $[2\beta, |N(v)|/4)$;  
\item $N'$, in the range $\left[|N(v)|/4, 2 |N(v)|\right)$; and 
\item $\Nhi$, at least $2|N(v)|$.
\end{itemize}
We focus on the non-edges within $N \cup N'$. Each such non-edge has a good chance of contributing to the slack of $v$, so the expected slack is what we need.

We first claim that $|\Nlo| \le 4\beta$. Supposing otherwise, there are at least $4\beta$ nodes of degree below $2\beta$. 
Observe that $2\beta = \hatzeta/6 \le |N(v)|/12$ (by definition of $\spar_v$).
Then the definition of redundancy (\cref{eqdef:redundancy}) implies that there is redundancy of more than $4\beta - 2\beta = 2\beta \ge (\hatzeta + \hateta)/12$, 
which contradicts our assumptions.

Consider now the case that $|\Nhi| \ge 2\beta$. Each node $u \in \Nhi$ has degree at least $2|N(v)| \ge 4\deg(v)/3$ (by assumption on $\deg(v)$) and contributes at least $1/4$ to the unevenness of $v$. That is,
$\hateta \ge |\Nhi|/4 \ge \beta/2 \ge (\hatzeta + \hateta)/24$, which implies the lemma. 
So, we assume from now otherwise, which means that 
$|N \cup N'| \ge |N(v)| - 6\beta$.

\begin{claim}     
 $|N \cup \Nlo| \le 24\beta$.
\label{claim:n-spar}
\end{claim}
\begin{proof}  
If $\beta \ge |N(v)|/24$, the lemma is trivially true, so assume otherwise.
Each node $u$ in $N\cup \Nlo$ has at least $|N \cup N'| - (\deg(u)+1) \ge (|N(v)|- 6\beta) - |N(v)|/4 \ge |N(v)|/2$ non-neighbors in $N \cup N'$.
The number of non-edges in $H[N(v) \cap \Vsg]$ with an endpoint in $N\cup \Nlo$ 
is then at least $|N\cup \Nlo|\cdot |N(v)|/2$. Since there are only $\hatzeta |N(v)| = 12 \beta |N(v)|$ non-edges in $N(v)$, the claim follows. 
\end{proof}

Consider a non-edge $uw$ in $N \cup N'$, where without loss of generality $\deg(u) \le \deg(w)$. Let $\chi$ be a color in $R_u$.
Let $C_{u,w}^\chi$ be the event that no node in $(N(v) \setminus \{u,w\})\cup N^+(u) \cup N^+(w)$ picks the color $\chi \in R_u$, where $N^+(u)$ denotes the neighbors of $u$ in $H$ of degree at least that of $u$.
Note that $C_{u,w}^\chi$ ensures that $u$ and $w$ get to keep their color if they both pick $\chi$.
It holds for each $u \in N(v) \setminus \Nlo$ that $\deg(u) \ge 2\beta \ge 8r$.

\begin{claim}
For $u,w \in N \cup N'$ and $\chi \in R_u$ it holds that  $\Pr[C_{u,w}^\chi] \ge 2^{-108}$.
      \label{claim:cuw}
\end{claim}
\begin{proof}
We consider separately the nodes in $N(v)$ of low/high degree, and those in $N^+(u)$ or $N^+(w)$ that are not in $N(v)$.

Each node $q$ that can choose $\chi$ has degree at least $\chi-1$ and has therefore at least $\deg(q)+1-r \ge \chi - r \ge \beta - r \ge \beta/2$ colors to choose from.
By \cref{claim:n-spar}, $|\Nlo \cup N| \le 24 \beta$. 
The probability that no node in $\Nlo \cup N$ picks a given color $\chi \in R_u$ is 
then at least 
\[ \prod_{q \in N \cup \Nlo} \left(1 - \frac{1}{\deg(q)+1-r}\right) 
 \ge \left(1 - \frac{2}{\beta}\right)^{|\Nlo \cup N|} 
 \ge \left(1 - \frac{2}{\beta}\right)^{24\beta} \ge 2^{-96} \ ,
\]
using that $(1-k/x)^x \ge 2^{-2k}$ for all $x,k > 0$.

Each node $q$ in $N' \cup \Nhi$ has at least $\deg(q)+1-r \ge |N(v)|/4-r \ge |N(v)|/3$ colors to choose from.
The probability that no node in $N' \cup \Nhi$ picks $\chi$ is then at least 
\[ \prod_{q\in N'\cup \Nhi} \left(1 - \frac{1}{\deg(q)+1-r}\right) 
\ge \left(1 - \frac{3}{|N(v)|}\right)^{|N(v)|} \ge 2^{-6} \ . \]
Nodes in $N^+(u)\setminus N(v)$ are of degree at least $\deg(u) \ge 2\beta \ge 8r$. 
The probability that none of the nodes in $N^+(u)\setminus N(v)$ pick $\chi$ is at least $(1-1/(\deg(u)+1-r))^{\deg(u)} \ge (1-4/(3\deg(u)))^{\deg(u)} \ge 2^{-8/3} \ge 2^{-3}$.
Similarly, the probability that none of the nodes in 
$N^+(w)\setminus (N(v)\cup N^+(u))$ pick $\chi$ is at least $(1-1/(\deg(w)+1-r))^{\deg(w)} \ge 2^{-3}$.

Combined, the probability of $C_{u,w}^\chi$ is at least
$\Pr[C_{u,w}^\chi] \ge 2^{-96} 2^{-6} \cdot 2^{-3} 2^{-3} = 2^{-108}$.
\end{proof}

We define the following events for nodes $u,w \in N\cup N'$ and color $\chi \in R_u$.
Let $B_{u,w}^\chi$ be the event that both $u$ and $w$ pick $\chi$. 
Observe that 
\begin{equation}
\sum_{\chi \in R_u} \Pr[B_{u,w}^\chi] = \sum_{\chi \in R_u} \frac{1}{(\deg(u)+1-r) (\deg(w)+1-r)} = \frac{1}{\deg(w)+1-r} \le \frac{4/3}{\deg(w)}.
\label{eq:pr-buw}
\end{equation}
Let $A_{u,w}^\chi$ be the event that both $B_{u,w}^\chi$ and $C_{u,w}^\chi$ occur.
Let $A_{u,w} = \cup_{\chi \in R_u} A_{u,w}^\chi$ be the event the nodes both pick the same color (in the given range) and no other node in these neighborhoods also picks it.
If $A_{u,w}$ holds, then the non-edge $uw$ contributes to the slack of $v$.
Note that $A_{u,w}$ is a disjoint union of the events $A_{u,w}^\chi$, which in turn is the intersection of independent events $B_{u,w}^\chi$ and $C_{u,w}^\chi$.

\begin{claim}
    Let $X = \sum_{u,w\in N(v)} \Pr[A_{u,w}]$. Then $\Exp[X] \ge 2^{-108}\beta$.
\label{claim:exp-x}
\end{claim}
\begin{proof}
Then, using \cref{claim:cuw} and \cref{eq:pr-buw},
\[ \Pr[A_{u,w}] 
  = \Pr\left[ \bigcup_{\chi \in R_u} B_{u,w}^\chi \cap C_{u,w}^\chi\right]  
  =  \sum_{\chi \in R_u} \Pr[B_{u,w}^\chi] \cdot \Pr[C_{u,w}^\chi]  
  \ge \frac{2^{-108}}{\deg(w)+1-r} \ge \frac{2^{-110}}{|N(v)|} \ , \]
where we used that (since $w \not\in \Nhi$) $\deg(w) \le 2|N(v)|$.

Let $\overline{m}(N\cup N')$ denote the number of non-edges within $N \cup N'$.
We will show that these $\overline{m}$ non-edges generate the desired slack, w.h.p.
Since each node has at most $|N(v)|$ incident non-edges (within $N(v)$), the bulk of the sparsity remains within $N \cup N'$:\
\begin{equation}   
\overline{m}(N\cup N') 
   \ge \hatzeta |N(v)| - |N(v)| \cdot |\Nlo \cup \Nhi|
   \ge (12\beta - 6\beta) |N(v)| = 6\beta |N(v)|\ .
\label{eqn:mbar}
\end{equation} 

By the linearity of expectation and \cref{eqn:mbar},
\[ \Exp[X] = \sum_{u,w \in N\cup N', uw\not\in E(H)} \Pr[A_{u,w}] 
     \ge \overline{m}(N\cup N') \cdot \frac{2^{-110}}{|N(v)|} \ge 2^{-108}\beta\ . \]
\end{proof}

We would like to apply Talagrand's inequality to $X$ but run into the problem that $X$ is not certifiable. Instead, we apply a standard approach of decomposing $X$ into a linear combination of certifiable functions (see \cite[Lemma 4.10]{AA20} and references therein).
Observe that $X$ is a conservative estimate of the amount of sparsity slack generated, as it does not consider the potential for reusing the same color multiple times. Its value is equivalent to the following count:
the number of colors $\chi$ (larger than $\beta$) such that some non-adjacent pair in $N(v)$ picked $\chi$ but no other node in $N(v) \cup N^+(u) \cup N^+(w)$ picked $\chi$. 

Let $Y$ be the number of colors $\chi \ge \beta$ such that some non-adjacent pair in $N \cup N'$ picked $\chi$. Also, let $Z$ be the number of such colors $\chi$ picked by a non-adjacent pair $u,w$ in $N \cup N'$ such that a (different) node in $N(v) \cup N^+(u) \cup N^+(w)$ also picked $\chi$. 
Observe that $X = Y - Z$.
Note that $Z$ is 3-certifiable, since for each color $\chi$ contributing to the count, it suffices to reveal the value of three nodes: $u$, $w$, and the node in $N(v) \cup N^+(u) \cup N^+(w)$ also picked $\chi$.
Similarly, $Y$ is 2-certifiable.
Also, both variables are 2-Lipschitz, as changing the color of one node affects the contribution to $Y$ or $Z$ of at most two colors.
We now bound $\Exp[Y]$ in order to apply Talagrand.

\begin{claim}
\label{claim:xy-talagrand}
$\Exp[Y] \le 2^8 \beta$.    
\end{claim}
\begin{proof}
Let $Q = \overline{E[N \cup N']} \subseteq \binom{N \cup N'}{2}$ be the set of non-edges within $N \cup N'$. 
Using our convention that $\deg(u) \le \deg(w)$ and \cref{eq:pr-buw},
\[ \Exp[Y] = \sum_{(u,w) \in Q} \sum_{\chi \in R_u} \Pr[B_{u,w}^\chi] 
  = \sum_{(u,w) \in Q} \frac{1}{\deg(w)+1-r} 
  \le \sum_{(u,w) \in Q} \frac{4/3}{\deg(w)}\ . \]
We consider the contributions when $w$ is in $N$ and when it is in $N'$.
Applying \cref{claim:n-spar},
\[ \sum_{(u,w) \in Q, w\in N} \frac{1}{\deg(w)} \le \binom{N}{2} \frac{1}{2\beta} \le \frac{(24\beta)^2}{4\beta} = 144\beta \ . \]
Also,
\[ \sum_{(u,w) \in Q, w\in N'} \frac{1}{\deg(w)} 
  \le \sum_{(u,w) \in Q, w\in N'} \frac{4}{|N(v)|} 
  \le \frac{4}{|N(v)|} \overline{m}(N\cup N') = 4 \spar_v(N \cup N')) \le 48 \beta\ ,\]
using the definition of sparsity in the last equality.
Hence, $\Exp[Y] \le 4/3 \cdot 192\beta = 2^8 \beta$.
\end{proof}

Combining \cref{claim:exp-x,claim:xy-talagrand}, we get that $\Exp[X] = \Omega(\Exp[Y])$.
We can then apply \cref{lem:talagrand-difference} (with $h=X$, $f=Y$, $g=Z$, $\alpha = 2^{-108-8}$) to obtain that 
$\Pr[X \le \Exp[X]/2] \le \exp(-\Omega(\hatzeta))$.
\end{proof}

\section{Coloring Non-Cabals}
\label{sec:color-non-cabals}
In this section, we detail the coloring of non-cabals, Step \ref{line:d1-col-non-cabals} in \cref{alg:d1-col}. More precisely, we prove \cref{prop:non-cabals} where the input coloring $\col$ should be thought of as the coloring produced by Step \ref{line:d1-col-star} in \cref{alg:d1-col}. Also recall that $\colsg$ is the coloring produced by slack generation (Step \ref{line:SG} in \cref{alg:d1-col}).

\PropColoringNonCabals*

In words, the main assumptions on $\col$ are that it coincides with $\colsg$ for the dense vertices and that it does not use reserved colors (not even in $\Vstar$). Recall that the main characteristic of non-cabals is that $e_K > \ell$, where $\ell = \Theta(\log^{1.2} n)$ as defined in \cref{eq:params}.
\cref{alg:non-cabals} describes the steps of our algorithm.

\begin{algorithm}
    \caption{\ColoringNonCabals\label{alg:non-cabals}}

    \nonl\Input{A coloring $\col$ such as described in \cref{prop:non-cabals}}

    \nonl\Output{A coloring $\psi \succeq \col$ such that $\Vdense \setminus \Vcabal = \dom\psi$}

    \colorfulmatching when $a_K \geq \Omega(\log n)$
    \tcp*{Let $\colcm$ be the coloring produced}
    \label[line]{line:non-cabals-CM}

    \coloringoutliers with $\calC(v) = [r+1, \deg(v)+1]$
    \label[line]{line:non-cabals-outliers}

    \sct
    \label[line]{line:non-cabals-sct}

    \trycolor for $O(1)$ rounds with $\calC(v) = L_\col(K_v) \cap [r+1, \deg(v)+1]$
    \label[line]{line:non-cabals-RCT}

    \slicecolor with $\calC(v) = L_\col(K_v) \cap [r+1, \deg(v)+1]$
    \label[line]{line:non-cabals-slice}
    
    \nonl Let $\mathscr{L}_1, \ldots, \mathscr{L}_{O(\log\log n)}$ be the layers produced by \slicecolor.

    \For{$i = O(\log\log n)$ to $1$\label[line]{line:non-cabals-layers-loop}}{
        \multitrial with $\calC(v) = [r]$ in $\mathscr{L}_{i}$
    \label[line]{line:non-cabals-MCT}
    }
\end{algorithm}

\begin{proof}[{Proof of \cref{prop:non-cabals}}]
We go over each step of \cref{alg:non-cabals}.

\paragraph{Colorful Matching (Step \ref{line:non-cabals-CM}).} 
We approximate palettes using the \emphdef{clique palette} $L_\col(K) = [\Delta_K + 1] \setminus \col(K)$. The advantage of $L_\col(K)$ compared to $L_\col(v)$ is that vertices can query it, even in virtual graphs.

\begin{restatable}[Query Algorithm, \cite{FHN23,parti}]{lemma}{LemQuery}
    \label{lem:query}
    Let $\psi$ be any (possibly partial) coloring of almost-clique $K$. 
    If each $v\in K$ holds $\calC(v) \in \set{\psi(K_v), L_\psi(K_v)}$ and $1\le a_v \le b_v \le \Delta_K+1$, then, w.h.p., each $v$ can either
    \begin{enumerate}
        \item learn $|\calC(v) \cap [a_v, b_v]|$; or 
        \item if it has $i_v\in[a_v,b_v]$, learn the $i_v^{th}$ color in $\calC(v) \cap [a_v, b_v]$.
    \end{enumerate}
    The algorithm runs in $\Ohat(1)$ rounds.
\end{restatable}

The downside of using the clique palette is that, in large almost-cliques, it can run out of colors. To remedy this, we compute a \emphdef{colorful matching}. Informally, we compute a coloring $\colcm$ that saves enough colors in $K$ to ensure $L(K)$ never runs out of colors.

\begin{lemma}[Colorful Matching $a_K \geq \Omega(\log n)$, \cite{parti}]
    \label{lem:colorful-matching-high}
    There exists a $\Ohat(1/\epsilon)$-round algorithm that outputs $\colcm \succeq \col$ such that, w.h.p., in each $K \notin \Kcabal$ where $a_K \geq \Omega(\log n)$, for all $v\in K$,
    \[
        \xi_{\colcm}(v, K) \geq \xi_{\colsg}(v, K) + M_K 
        \quad\text{and}\quad 
        M_K \geq \Omega(a_K / \eps) \ .
    \]
    Further, the algorithm uses only colors in $\colcm(K) \setminus \col(K) \subseteq [|K|/10, (1-2\eps)|K|]$ and colors a vertex iff it saves a color, i.e., $|\colcm(K)| - |\col(K)| = M_K$. 
\end{lemma}

By \cref{lem:colorful-matching-high}, w.h.p., we compute in each $K \notin \Kcabal$ such that $a_K \geq \Omega(\log n)$ a colorful matching of size $\Omega(a_K/\eps)$. Using the query algorithm (\cref{lem:query}) vertices of $K$ learn $M_K$ by comparing the number of nodes in $K$ colored before and after this step. If the colorful matching has size $\geq 2\eps|K|$, then all vertices of $K$ have $\eps|K|$ slack and can be colored with high probability in $\Ohat(\log^* n)$ rounds (the same way we color $\Vstar$ and outliers). We henceforth assume $M_K < 2\eps|K|$.

\paragraph{Coloring Outliers (Step \ref{line:non-cabals-outliers}).}
In this step, we detect atypical vertices called \emphdef{outliers} and color them first. Intuitively, outliers $O_K \subseteq K$ are the few nodes that do not get enough savings from the colorful matching and will be hard to color later. Vertices $I_K = K \setminus O_K$ that are not outliers are called \emphdef{inliers}.
Formally, we define inliers/outliers as
\begin{equation}
    \label{eqdef:inliers-non-cabals}
    I_K \eqdef \set*{ v\in K : e_v \leq 20 e_K \text{ and } a_v \leq 20 a_K } 
    \quad\text{ and }\quad
    O_K \eqdef K \setminus I_K \ .
\end{equation}
Recall $a_v$ and $e_v$ are known to $v$ because they include inaccuracies. Similarly, vertices compute $a_K$ and $e_K$ in $\Ohat(1)$ rounds with a BFS on a tree spanning $K$. In particular, each vertex knows if it is an inlier or an outlier. By Markov inequality, at most $|K|/10$ vertices in $K$ can be outliers. By \cref{fact:colsg-AC} at most $|K|/10$ inliers are colored by slack generation. Outliers have slack proportional to their degrees in $G[\bigcup_{K \notin \Kcabal} O_K]$ because they are adjacent to at least $(1 - \eps - 1/10 - 1/10 - 2\eps)|K| \geq 0.5|K| \geq 0.25\deg(v)$ uncolored inliers. Since they have degree $\deg(v) \geq \Deltalow \gg r$, even after removing reserved colors, outliers have slack $0.1\deg(v)$. Hence they get colored in $\Ohat(\log^* n)$ rounds ($O(1)$ iterations of \refRCT and $\Ohat(\log^* n)$ rounds performing \refMCT with $\calC(v) = [\deg(v)+1] \setminus [r]$).

\paragraph{Accounting for All Sources of Slack.}
The \refCMHigh and \refSlackGen aimed to save enough colors so dense vertices never run out of available colors in the clique palette. Since we already colored the outliers, we focus on counting the available colors for inliers. We formalize this in \cref{lem:enough-available}. Since the proof of \cref{lem:enough-available} is a lengthy --- although not complicated --- case analysis adding up all sources of slack, we defer its proof to \cref{sec:proof-accounting}.

\begin{restatable}[Accounting Lemma]{lemma}{LemAccounting}
    \label{lem:enough-available}
    There exists a universal constant $\gamma_{\ref{lem:enough-available}} = \gamma_{\ref{lem:enough-available}}(\eps) \in (0,1)$ for which the following holds.
    For any $\col \succeq \colcm$ such that $\dom\col \subseteq \Vsg$ and $v\in I_K \subseteq K \notin \Kcabal$, the number of colors available to $v$ in the clique palette is
    \[
        |L_\col(v) \cap L_\col(K)| \geq \gamma_{\ref{lem:enough-available}} \cdot (e_K + a_K)  \ .
    \]
\end{restatable}

\paragraph{Synchronized Color Trial (Step \ref{line:non-cabals-sct}).}
For each $K\notin \Kcabal$, let $S_K$ be any set of uncolored inliers of size $|K \setminus \dom\col| - r - O(\log n)$. 
We use a smaller clique palette where $\Delta_K$ is replaced by $\Delta_{I_K} = \max_{u\in I_K} \deg(u)$ the maximum degree of an \emph{inlier}. We must do so because one high degree outlier in $K$ would add many colors in $[\Delta_K + 1]$ that no inlier can adopt due to their small degrees.
\cref{lem:sct} is more general than needed because it applies the same way to cabals.
We defer the proof to \cref{sec:sct-proof}

\begin{restatable}[Synchronized Color Trial]{lemma}{LemSCT}
    \label{lem:sct}
    Fix an almost-clique $K$ and $\col$ a partial coloring.
    Let $\alpha \in (0,1)$ and $S_K \subseteq K$ such that $[\Delta_{I_K} + 1] \setminus (\col(K) \cup [r]) \geq |S_K| \geq \alpha |K|$.
    After \cref{alg:sct}, w.h.p., $S_K$ contains $500/\alpha \cdot (e_K + a_K) + O(\log n)$ uncolored vertices. 
    This holds even if random bits outside of $K$ are adversarial.
\end{restatable}

\cref{lem:clique-palette-sct} shows there are enough colors in $[\Delta_{I_K} + 1] \setminus (\col(K) \cup [r])$ for all nodes of $S_K$. It follows from the savings given by the colorful matching. The $O(\log n)$ comes from the cases where $a_K \leq O(\log n)$ and $M_K = 0$.

\begin{restatable}{lemma}{LemPaletteSCT}
    \label{lem:clique-palette-sct}
    Suppose $\col \succeq \colcm$. 
    Then, for all $K$, we have $|[\Delta_{I_K} + 1] \setminus \col(K)| \geq |K \setminus \dom\col| - O(\log n)$.
\end{restatable}

\begin{proof}
    Let us first lower bound the number of colors in $[\Delta_{I_K} + 1]$ in terms of savings in $K$. Let $v\in I_K$. Recall that, by \cref{fact:count-degree} and definition of $\Delta_{I_K}$ we have $\Delta_{I_K} + 1 \geq \deg(v) + 1 \geq |K| - a_v \geq |K| - 20a_K$. Hence, partitioning vertices as $|K| = |K \setminus \dom \col| + |K \cap \dom\col|$, we deduce
    \begin{align*}
        |[\Delta_{I_K} + 1] \setminus \col(K)| &\geq
        \Delta_{I_K} + 1 - |\col(K) \cap [\Delta_{I_K} + 1]| \\&\geq
        |K \setminus \dom\col| + 
        \xi_\col(v, K) - 20 a_K \ . 
        \addtocounter{equation}{1}\tag{\theequation} \label{eq:reuse-CM}
    \end{align*}
    When $a_K \leq O(\log n)$, \cref{eq:reuse-CM} implies the result even if $\xi_\col(v, K) = 0$.
    If $a_K \geq \Omega(\log n)$, then, by \cref{lem:colorful-matching-high}, we have $\xi_\col(v, K) \geq \Omega(a_K / \eps) \geq 20a_K$ and \cref{eq:reuse-CM} yields the result.
\end{proof}

Since there are $0.9|K|$ inliers and at most $(0.1 + 2\eps)|K| \leq 0.11|K|$ are colored by slack generation and the colorful matching, $|S_K| \geq 0.79|K| - O(\log n) \geq 0.75|K|$. Hence, by \cref{lem:sct} and union bound, after running \cref{alg:sct} in all non-cabals in parallel, w.h.p., each $K\notin \Kcabal$ contains at most $670(e_K + a_K) + r + O(\log n) \leq O(e_K + a_K)$ uncolored inliers (including uncolored vertices in $K \setminus S_K$). Adding the $\leq 20 e_K$ external neighbors, each vertex has uncolored degree $|N_\col(v)| \leq O(e_K + a_K)$.

\newcommand{\gammaCP}{\gamma_{\ref{lem:enough-available}}}
\paragraph{Decreasing Uncolored Degrees by a Constant Factor (Step \ref{line:non-cabals-RCT}).}
Let $C_0$ be the universal constant such that for all vertices $|N_\col(v)| \leq C_0 \cdot (e_K + a_K)$.
We run \refRCT with $\calC_{\col_i}(v) = L_{\col_i}(K) \cap [r+1,\deg(v)+1]$ for $T = O(\gammaCP^{-4}C_0 \log \gammaCP^{-1}) = O(1)$ iterations, where $\col_i \succeq \col$ is the extension of the coloring obtained after the $i^{th}$ iteration.
By \cref{lem:enough-available}, we always have that
\begin{equation}
    \label{eq:RCT-available}
    |L_{\col_i}(v) \cap \calC_{\col_i}(v)| = 
    |L_{\col_i}(v) \cap L_{\col_i}(K) \setminus [r]| \geq 
    \gammaCP \cdot (e_K + a_K) - r \geq \gammaCP/2 \cdot (e_K + a_K)
\end{equation}
where the last inequality uses that $e_K \geq \ell$ and $\ell \gg r$ (\cref{eq:params}). 

Let us verify conditions of \cref{lem:try-color}. Using the \refQuery vertices can sample uniform colors in $\calC_{\col_i}(v)$ and since $K \notin \Kcabal$ we have $|\calC_{\col_i}(v)| \geq \Omega(e_K) \gg \log n$.
To verify the third condition, there are two cases. If $|\calC_{\col_i}(v)| \geq 2e_v$, then $|\calC_{\col_i}(v) \cap L_{\col_i}(v)| \geq |\calC_{\col_i}(v)| - e_v \geq |\calC_{\col_i}(v)|/2$. Otherwise, using $e_v \leq 20e_K$ and \cref{eq:RCT-available}, we get $|\calC_{\col_i}(v) \cap L_{\col_i}(v)| \geq \gammaCP/40 \cdot e_v \geq \gammaCP/80 \cdot |\calC_{\col_i}(v)|$. Finally, the fourth condition is implied by \cref{eq:RCT-available} as uncolored degrees are $|N_{\col_i}(v)| \leq C_0 (e_K + a_K) \leq 2\gammaCP^{-1}C_0 |L_{\col_i}(v) \cap \calC_{\col_i}(v)|$. Hence, w.h.p, after each iteration of the uncolored degree decreases by a $(1 - \Theta(\gammaCP^4/C_0))$ factor.
Thus, after $T$ iterations the maximum uncolored degree of a vertex in $K$ is $|N_{\col_T}(v)| \leq \gammaCP/10 \cdot (e_K + a_K)$.

\paragraph{Reducing Degrees to Logarithmic (Step \ref{line:non-cabals-slice}).}
Now that vertices have enough slack in the clique palette, we use it to reduce uncolored degrees to $O(\log n)$. We use the slice color algorithm of \cite{FHN23} that tries colors in $L_\col(v) \cap L_\col(K)$ for $O(\log\log n)$ iterations, in $\Ohat(\log \log n)$ rounds\footnote{The setting considered in \cite{FHN23} has constant $\congestion$ and $\dilation$, hence the $\Ohat$ here compared to the original statement.}.

\begin{lemma}[{Slice Color, \cite{FHN23}}]
    \label{lem:slice-color}
    Let $\alpha, \kappa > 0$ be two universal constants, $\col$ a (partial) coloring of $H$ and $S \subseteq V_H$.
    Suppose each $v\in S$ knows an upper bound on its uncolored degree $d(v) \geq |N(v) \cap S \setminus \dom\col|$ and some $s(v) \geq \alpha \cdot d(v)$.
    Suppose that for all vertices with $d(v) \geq \Omega(\log n)$, we have an algorithm that for any extension $\psi \succeq \col$ to $S$, samples $c(v) \in L_\psi(v) \cup \set{\bot}$ (where $\bot$ represents a failure) such that $\Pr[c(v) = \bot] \leq 1/\poly(n)$ and for all $c \in L_\psi(v)$,
    \[
        \Pr[c(v) = c~|~c(v) \neq \bot] \leq \frac{\kappa}{d(v) + s(v)} \ .
    \]
    Then there is a $\Ohat(\log\log n + \kappa \log(\alpha/\kappa))$-round algorithm that outputs an extension $\psi \succeq \col$ of the coloring such that, w.h.p., uncolored vertices of $S$ are partitioned into $O(\log\log n)$ layers $\mathscr{L}_1, \ldots, \mathscr{L}_{O(\log\log n)} \subseteq S \setminus \dom\psi$ satisfying
    \[
        \text{for every $i \in [O(\log\log n)]$ and $v\in \mathscr{L}_i$,} 
        \quad 
        |N_{\psi}(v) \cap \mathscr{L}_{\geq i}| \leq O(\log n) \ .    
    \]
    Moreover, the algorithm uses only colors supported by the sampler, i.e., colors $c$ such that $\Pr[c(v) = c] > 0$.
\end{lemma}

We use the sampler from \cite{FHN23}. Intuitively, we use a random hash function $h : [\Delta_K + 1] \to [O(\Delta_K)]$ to sample $\Theta(\log n)$ colors in the clique palette as $h^{-1}(\set{1, 2, \ldots, \Theta(\log n)}) \cap L_\col(K)$. It is then compared to $h^{-1}(\set{1, 2, \ldots, \Theta(\log n)}) \cap \col(E_v)$ colors used by external neighbors. To implement it efficiently, machines can aggregate bitmaps to represents both sets using $\Theta(\log n)$ bandwidth. Further, a classic trick of Newman \cite{Newman91} shows how to reduce the amount of randomness needed to describe $h$ to $O(\log n)$ bits (at the cost of exponential local computations).

\begin{lemma}[\cite{FHN23}]
    \label{lem:sampler}
    Let $\col$ be any (partial) coloring of $H$.
    There exists a sampler which succeeds w.h.p.\ ($\Pr[c(v) = \bot] \leq 1/\poly(n)$) and does not sample reserved colors $\Pr[c(v) \in [r]] = 0$. Further, for any $c \in L_\col(v) \cap L_\col(K) \setminus [r]$ we have
    \[
        \Pr[c(v) = c~|~c(v) \neq \bot] \leq \frac{O(1)}{|L_\col(v) \cap L_\col(K) \setminus [r]|} \ .
    \]
\end{lemma}

We now describe the parameters we use to apply \cref{lem:slice-color}. Let $S = \Vdense \setminus (\dom \col \cup \Vcabal)$ be the uncolored vertices in non-cabals.
Let $d(v) = \gammaCP/10 \cdot (e_K + a_K)$ and $s(v) = \gammaCP/4 \cdot (e_K + a_K)$ where $K = K_v$. After Step \ref{line:non-cabals-RCT}, we have $|N_\col(v)| \leq d(v)$, and by \cref{eq:RCT-available}, we get that $|L_\col(v) \cap L_\col(K) \setminus [r]| \geq d(v) + s(v)$ for all extension $\dom\col \subseteq \Vsg$. This suffices as $S \subseteq \Vsg$. The sampler from \cref{lem:sampler} verify all conditions of \refSlice and we therefore extend the coloring in $\Ohat(\log\log n)$ rounds such that uncolored vertices are partitioned into $O(\log\log n)$ layers.

\paragraph{Finishing The Coloring with MultiColor Trials (\cref{line:non-cabals-layers-loop,line:non-cabals-MCT}).}
Observe that reserved colors are not used by any vertex in $V_H$.
Indeed, by assumption, the coloring given in input does not use reserved colors. The \refCMHigh uses only colors in $[|K|/10, (1-\eps)|K|]$ which does not intersect with $[r]$ since $|K| \geq \Deltalow/2 \gg r$. Observe that through steps \ref{line:non-cabals-outliers} to \ref{line:non-cabals-RCT}, we were careful not to use any reserved color. 
Finally, the sampler from \cref{lem:sampler} does not include any reserved color in its support; hence, \refSlice does not use any reserved color in Step \ref{line:non-cabals-slice}.

Since $\col(V_H) \cap [r] = \emptyset$, for each $v\in \mathscr{L}_i$, the number of reserved colors available is
\[
    |L_\col(v) \cap [r]| = r = C_1 \cdot \log^{1.1} n \geq 
    C_1/4 \cdot \log^{1.1} n + 3|N_\col(v) \cap \mathscr{L}_{\geq i}| \ ,
\]
where $C_1$ is some large universal constant (\cref{eq:params}) and we use that $|N_\col(v) \cap \mathscr{L}_{\geq i}| \leq O(\log n)$. Hence, by coloring layers sequentially from $\mathscr{L}_{O(\log\log n)}$ down to $\mathscr{L}_1$, at each layer, conditions of \refMCT are verified and vertices get colored in $\Ohat(\log^* n)$ rounds with high probability. Thus, all vertices of $\Vdense \setminus \Vcabal$ are colored in $\Ohat(\log\log n \cdot \log^* n)$ rounds.
\end{proof}

\subsection{Proof of the Accounting Lemma}
\label{sec:proof-accounting}

\LemAccounting*

\begin{proof}
First, we fix $\col \succeq \colcm$ and $v$ such as described in the statement. We shall write $S = N(v) \cup K_v$ and recall the definition of savings $\xi_\col(v, S)$ from \cref{eqdef:savings}.
We begin with three cases where $\xi_\col(v, S) \geq \gamma \cdot e_K + 21 a_K - \delta_v^e$ for some universal constant $\gamma$.

\begin{itemize}
    \item \textbf{Case 1:} $a_K \geq e_K/2$.
    Since $K\not\in \Kcabal$, we have $a_K \geq e_K/2  \geq \lmin/4 \gg \log n$. Hence, the \refCMHigh saves enough colors: $\xi_\col(v, S) \geq \xi_\col(v, K) \geq M_K \geq  e_K + 21 a_K$.
\end{itemize}

In all future cases, we assume $a_K \leq e_K/2$. It then suffices to show that $\xi_{\colsg}(v, S) \geq 2\gamma \cdot e_K - \delta_v^e$ because, when $a_K \geq \Omega(\log n)$ we save an extra $\Omega(a_K/\eps) \geq 21a_K$ colors with the \refCMHigh, and when $a_K \leq O(\log n)$ then $\gamma \cdot e_K \geq \ell = \Theta(\log^{1.2} n) \gg 21 a_K$.

\begin{itemize}
    \item \textbf{Case 2:} $a_K \leq e_K/2$ and $e_v \leq e_K/6$.
    Recall that $\deg(u) + 1 = |K| + e_u - a_u$ (\cref{fact:count-degree}), for each node $u \in K$. 
    Thus, the average degree in $K$ is 
    $\deg(K)+1 = \sum_{u\in K} \frac{\deg(u)+1}{|K|} = |K| + e_K - a_K$.
    Since $a_K$ and $e_v$ are small compared to $e_K$ (i.e., $a_K + e_v \le 2 e_K/3$), the degree of $v$ must be significantly smaller than average
    $
        \deg(v) + 1 \le 
        |K| + e_v =
        \deg(K) + 1 - e_K + a_K + e_v \le
        \deg(K) + 1 - e_K/3 .
    $
    We lower bound the unevenness as 
    \begin{align*}
        e_K / 3 
        \leq \deg(K) - \deg(v) 
        &= \sum_{u\in K} \frac{\deg(u) - \deg(v)}{|K|} \\
        &\leq 2\sum_{u\in K} \frac{\deg(u) - \deg(v)}{\deg(u)+1} 
            \tag{for all $u\in K$, $|K| \geq (1-2\eps)\deg(u)$} \\
        &\leq 2\sum_{u \in S} \frac{(\deg(u) - \deg(v))^+}{\deg(u) + 1} 
        = 2\discr_v(S)
    \end{align*}
    where the last inequality holds because we only remove negative terms and add non-negative terms to the sum. By \cref{lem:gen-uneven}, w.p.\ $1 - \exp(-\Omega(\discr_v(S))) \geq 1 - \exp(-\Omega(e_K)) \geq 1-1/\poly(n)$, after slack generation we saved $\xi_{\colsg}(v, S) \geq \gamma_{\ref{lem:gen-uneven}} \cdot \discr_v(S) \geq \gamma_{\ref{lem:gen-uneven}}/6 \cdot e_K$ colors.

    \item \textbf{Case 3:} $a_K \le e_K/2$, $e_K/6 \leq e_v$, $\red_v \leq (\spar_v + \discr_v)/12$ and $|N(v) \setminus \Vsg| \leq (\spar_v + \discr_v) / 12$.
        
    If $|E_v| \leq e_K/12$, since $e_v = \delta_v^e + |E_v| \geq e_K/6$, then $\delta_v^e \geq e_K/12$. In that case, the claimed bound holds as $\xi_{\colsg}(v, S) \geq 0 \geq e_K/12 - \delta_v^e$.
    We henceforth assume $|E_v| \geq e_K / 12 \geq \Omega(\log n)$. 
    By \cref{part:ACD-ext} of \cref{prop:AC} and the assumed bound on $|N(v) \setminus \Vsg| \geq |N(v) \cap \Vin|$,
    we have $\spar_v + \discr_v \geq \Omega_\eps(|E_v|) \geq \Omega(\log n)$. Further, this implies $\spar_v + \discr_v \geq \Omega_\eps(\ell) \geq 100 r$ by our choice of $r$ (\cref{eq:params}). Hence, under our assumptions \cref{prop:slack-generation} applies and, w.h.p., $\xi_{\colsg}(v, S) \geq \CSlack \cdot (\spar_v + \discr_v) \geq \Omega_\eps(|E_v|) \geq \Omega_\eps(e_v) - \delta_v^e \geq \Omega(e_K) - \delta_v^e$. 
\end{itemize}

We now argue that in these three cases, the lower bound on savings implies the result.
Let $\gamma$ be the absolute constants such that $\xi_\col(v, S) \geq \gamma \cdot e_K + 21a_K - \delta_v^e$.
The number of colors in the clique palette available to $v$ is
\begin{align*}
    |L_\col(v) \cap L_\col(K)| &= 
    \deg(v) + 1 - |\col(S) \cap [\deg(v) + 1]| =
    \deg(v) + 1 - |S \cap \dom \col| +  \xi_\col(v, S) \ .
\end{align*}
Then, using $|K \cap \dom\col| = |K| - |K \setminus \dom\col|$, $|E_v \cap \dom\col| = |E_v| - |E_v \setminus \dom\col|$, and $\deg(v) + 1 - |K| = e_v - a_v$ (\cref{fact:count-degree}), this becomes
\begin{equation}
    \label{eq:lower-bound-available-S}
    |L_\col(v) \cap L_\col(K)| - |S \setminus \dom\col| 
    \geq \xi_\col(v, S) + (e_v - |E_v|) - a_v \geq \gamma \cdot (e_K + a_K) \ .
\end{equation}
where the last inequality uses the lower bound on $\xi_\col(v, S)$, that $a_v \leq 20a_K$ (since $v$ is an inlier) and $\gamma < 1$. 

The two remaining cases are such that $v$ has no saving slack. They are variations on \textbf{Case 3} where either the redundancy or the number of inaccurate neighbors is high (hence we cannot apply \cref{lem:slack-generation-node}).

\begin{itemize}
    \item \textbf{Case 4:} $a_K \le e_K/2$, $e_K/6 \leq e_v$ and \underline{$|N(v) \setminus \Vsg| \geq (\spar_v + \discr_v) / 12$}.
    
    By the first inequality in \cref{eq:lower-bound-available-S}, the number of available color is
    $|L_\col(v) \cap L_\col(K)| \geq |S \setminus \dom\col| + \xi_{\colcm}(v, S) + \delta_v^e - 20 a_K$.
    Recall the assumption that $\dom\col \subseteq \Vsg$, meaning that $V_H \setminus \Vsg$ remains uncolored. Hence $|S \setminus \dom\col| \geq |N(v) \setminus \Vsg|$. With the \refCMHigh when $a_K \geq \Omega(\log n)$, this gives us
    $|L_\col(v) \cap L_\col(K)| \geq |N(v) \setminus \Vsg| + \delta_v^e - O(\log n)$. Using \cref{part:ACD-ext} in \cref{prop:AC}, we have $|N(v) \setminus \Vsg| \geq \Omega_\eps(|E_v|) \geq \Omega_\eps(e_v) - \delta_v^e \geq \Omega_\eps(e_K) - \delta_v^e$ which concludes this case because $e_K \gg \log n$.

    \item \textbf{Case 5:} $a_K \le e_K/2$, $e_K/6 \leq e_v$ and \underline{$|N(v) \setminus \Vsg| \leq (\spar_v + \discr_v) / 12$ and $\red_v \geq (\spar_v + \discr_v)/12$}. 

    By definition of redundancy (\cref{eqdef:redundancy}) there exists $t \leq |N(v)|/12$ such that
    $$|S| - t \geq |\set{ u \in S: \deg(u) + 1 > t}| + \red_v \ . $$
    Observe that \cref{eqdef:redundancy} holds for $N(v)$ rather than $S$ but adding anti-neighbors increases the left-hand side at least at much as the right-hand side of the inequality.
    The number of available colors with respect to $\col$ is at least
    \begin{align*}
        |L_\col(v) \cap L_\col(K)| &\geq 
        \deg(v) + 1 - |\col(S) \cap [\deg(v) + 1]| \\ &\geq
        (|S| + \delta_v^e - a_v) - t - |\col(S) \cap [t+1, \deg(v) + 1]| \ ,
    \end{align*}
    where we used that $\deg(v) + 1 = \delta_v + |N(v)| = \delta_v + |S| - |A_v| = \delta_v^e + |S| - a_v$ (\cref{eq:def-degree-slack}).
    Plugging in the lower bound on $|S| - t$ we got from redundancy, 
    \begin{align*}
        |L_\col(v) \cap L_\col(K)| &\geq
        |\set{ u \in S: \deg(u) + 1 > t}| + \red_v + \delta_v^e - a_v -
        |\col(S) \cap [t+1, \deg(v) + 1]|  \\ &\geq
        |\set{ u \in S: \col(u) > t}| -
        |\col(S) \cap [t + 1, \deg(v) + 1]| + \red_v - a_v + \delta_v^e \ .
        \addtocounter{equation}{1}\tag{\theequation} \label{eq:available-red}
    \end{align*}
    Since the \refCMHigh uses only colors in $[|K|/12, (1-\eps)|K|] \subseteq [t+1, \deg(v)+1]$, we save $M_K$ colors. 
    Further, savings do not decrease as we extend $\colcm$ to $\col$, hence
    \begin{equation}
        \label{eq:savings-CM-red}
        |\set{ u \in K: \col(u) > t}| - |\col(K) \cap [t+1, \deg(v)+1]| \geq M_K \ .
    \end{equation}
    Together, \cref{eq:savings-CM-red,eq:available-red} give us
    \begin{align*}
        |L_\col(v) \cap L_\col(K)| &\geq 
        \red_v + M_K - a_v + \delta_v^e \\ &\geq 
        1/24 \cdot (\spar_v + \discr_v + |N(v) \setminus \Vsg|) 
            + M_K - a_v + \delta_v^e
            \tag{as $\red_v \geq \frac{1}{24}(\spar_v + \discr_v + |N(v) \setminus \Vsg|)$}\\ &\geq 
        \Omega_\eps(|E_v|) + M_K - a_v + \delta_v^e 
            \tag{\cref{part:ACD-ext} in \cref{prop:AC}}\\ &\geq 
        \Omega_\eps(e_K) + M_K - a_v \ . \tag{$|E_v| + \delta_v^e = e_v \geq e_K/6$}
    \end{align*}
    This implies the lemma since, when $a_K \geq \Omega(\log n)$ then $M_K - a_v \geq \Omega(a_K/\eps) - a_v \geq a_K$, and otherwise when $a_K \leq O(\log n)$ then $\Omega_\eps(e_K)/2 \geq \Omega_\eps(\ell) = \Theta_\eps(\log^{1.2} n) \gg 21 a_K$.
\end{itemize}
\end{proof}

\subsection{Synchronized Color Trial}
\label{sec:sct-proof}

\begin{algorithm}[ht!]
    \caption{\sct\label{alg:sct}}
    
    \nonl\Input{The coloring $\colsct \succeq \colcm$ and $S_K \subseteq K \setminus \dom\col$ such that $|[\Delta_{I_K} + 1] \setminus \col(K)| - r \geq |S_K| \geq \alpha|K|$}
    
    \nonl\Output{A coloring $\colsct$ s.t.\ each $K$ contains at most $O(e_K + a_K + \log n)$ uncolored nodes}
    
    Let $k = |S_K|$ and
    $v_1, \ldots, v_{k}$ be an arbitrary ordering of vertices in $S_K$
    
    Let $\pi$ be a uniform random permutation of $\set{1, 2, \ldots, k}$
    
    Let $c(v_i)$ be the $\pi(i)$-th color in $[\Delta_{I_K}+1] \setminus (\col(K) \cup [r])$, where $\Delta_{I_K} = \max_{u\in I_K} \deg(u)$.
    
    If $c(v) \in L_\col(v) \setminus c(N_H(v))$ then let $\colsct(v) = c(v)$ and $\colsct(v) = \bot$ otherwise.
\end{algorithm}

Implementing \cref{alg:sct} is done by adapting techniques of \cite{FHN23}. Namely, by partitioning almost-cliques into random groups and computing prefix sums on the support tree of some vertex in $K$. We refer readers to \cite[Section 4.2]{FHN23} and \cite[Appendix C.2]{parti} for more details.

\begin{lemma}[{\cite[Lemma~21]{FHN23}}]
    \cref{alg:sct} can be implemented in $\Ohat(\log\log n)$ rounds.
\end{lemma}

\LemSCT*

\begin{proof}
    Let us focus on $K$ (the result follows from union bound). Call $S = S_K$.
    Define the random variable $X_i$ indicating if $i^{th}$ vertex in $S_K$ failed to retain its color.
    By assumption, all nodes receive some color in $[\Delta_{I_K} + 1]$ to try. 
A vertex might fail to retain a color either because the color is outside of $[\deg(v)+1]$ or because it conflicts with an \emph{external} neighbor.

    Split $S$ into sets $A = \set{1, 2, \ldots, \floor{|S|/2}}$ and $B = S \setminus A$ of size at least $\floor{|S|/2}$ each. We first show that the number of nodes to fail in $A \cap S$ over the randomness of $\pi(A)$ is small. 
    Fix $i\in A$. After revealing $\pi(1), \ldots, \pi(i-1)$, the value $\pi(i)$ is uniform in a set of $|S|-|A| \geq |S|/3 \geq (\alpha/3) |K|$ values.
    By union bound, the probability the $i^{th}$ vertex fails to retain its color --- even under adversarial conditioning of $\pi(1), \ldots, \pi(i-1)$ --- because of an external neighbor is $\frac{e_v}{|S| - |A|}$.
    To bound the probability of receiving a color outside of $[\deg(v)+1]$, consider $w$ the inlier of maximal degree. \cref{fact:count-degree} implies that $\Delta_{I_K} - \deg(v) = \deg(w) - \deg(v) \le e_w + a_v \le 20e_K + a_v$. Hence, $v$ receives a color outside of $[\deg(v)+1]$ with probability $\le \frac{20 e_K + a_v}{|S| - |A|}$.
    Overall, the $i^{th}$ vertex fails to retains its color with probability 
    \[
    \Pr [X_i = 1~|~\pi(1), \pi(2), \ldots, \pi(i-1)] \le \frac{e_v + 20e_K + a_v}{|S| - |A|} \leq 3/\alpha \cdot \frac{40e_K + a_v}{|K|} \ .
    \]
    By linearity of the expectation, 
    $$\Exp \range* { \sum_{i\in A} X_i } \le 
        3/\alpha \sum_{v\in K} \frac{40e_K + a_v}{|K|} \leq 120/\alpha\cdot (e_K + a_K) \ . $$
    By the martingale inequality (\cref{lem:chernoff}), with high probability, $|(A \cap S) \setminus \dom\colsct| < 240/\alpha \cdot \max\set{e_K + a_K, \Theta(\log n)}$. The same bound holds for vertices in $B \cap S_t$. By union bound, w.h.p., both bounds hold simultaneously, and hence$|S \setminus \dom\colsct| \leq 500/\alpha \cdot \max\set{e_K + a_K, \Theta(\log n)}$. 
\end{proof}

\section{Coloring Low-Degree Nodes}
\label{sec:low-deg}

In this section, we show how to color low-degree nodes, notably proving  \cref{prop:low}.

\PropLow*

\Cref{prop:low} is achieved through \cref{alg:coloring-low-deg}. The algorithm first reduces the uncolored part of the graph to two uncolored subgraphs of maximum degree $O(\log n)$ through a sequence of random color trials. Remaining steps are performed first in one of the subgraphs, then in the other, so that the subgraphs are colored one after the other rather than in parallel. Each uncolored node $v$ in the currently processed subgraph learn $\card{N_\col(v)}+1$ colors from its palette, where the uncolored degree of the node is measured w.r.t.\ the subgraph it is in. In the final step, random color trials are used to shatter the remaining uncolored subgraphs into small connected components, which we color using a deterministic degree+1-list-coloring algorithm from the literature. Performing this deterministic algorithm is enabled by the low degree of the remaining uncolored subgraphs and the nodes' knowledge of their palettes. Coming \cref{sec:low-deg-sampling,sec:low-deg-learncolors,sec:finish-coloring-low-deg} explain each step of the algorithm. We devote a last section to a minor result improving the complexity in graphs of even smaller degree, \cref{sec:even-lower-deg}.

\begin{algorithm}[ht]
    \caption{\alg{ColoringLowDegree} \label{alg:coloring-low-deg}}

    \nonl\Input{A virtual graph with partial coloring $\col$ s.t.\ the subgraph induced by uncolored nodes has maximum pseudo-degree at most $\Deltalow$}
    
    \nonl\Output{An new coloring $\collow \succeq \col$ extending the coloring to the entire graph}

    \alg{LowDegreeReduction} \label[line]{line:low-deg-reduction} \hfill (\cref{sec:low-deg-sampling})

    Partition remaining uncolored nodes according to their uncolored pseudo-degree: \label[line]{line:low-deg-partition}
    \begin{itemize} 
    \item $V_1 \gets \set{v \in V \setminus \dom \col: \deg_\col(v) \leq C \log n}$, and
    \item $V_2 \gets (V \setminus \dom \col) \setminus V_1$.
    \end{itemize}

    \ForEach{uncolored subset of the nodes $V_1$ and $V_2$}{
        \alg{LearnPalette} \label[line]{line:low-deg-palette} \hfill (\cref{sec:low-deg-learncolors})

        \alg{FinishColoring} \label[line]{line:low-deg-deterministic} \hfill (\cref{sec:finish-coloring-low-deg})
    }
\end{algorithm}

\subsection{Sampling Colors and Degree Reduction}
\label{sec:low-deg-sampling}

\begin{algorithm}[ht]
    \caption{\alg{LowDegreeReduction}  \label{alg:low-deg-reduction}}

    \nonl \Input{An upper bound $\Dmaxcol$ on the uncolored degree of the virtual graph $H$ with partial coloring $\col$}

    \nonl \Output{An extension of the partial coloring s.t.\ uncolored nodes have $O(\log n)$ neighbors of uncolored degree $\Omega(\log n)$}

    \For{$\Theta(\log \Dmaxcol)$ iterations, at each uncolored node $v$ in parallel}{
        $v$ flips a fair coin.

        \If{$v$'s coin turned heads}{
        $v$ samples a random color $c$ using \alg{LogDegSampling}.

        If no neighbor of $v$ of higher ID sampled the same color, $v$ colors itself with $c$.
        }
    }
\end{algorithm}

For the first step, the main technical ingredient is a procedure for nodes to sample random colors from their palettes (\alg{LowDegSampling}, \cref{alg:low-deg-sampling}). This random sampling is performed using a form of binary search in the color space of the nodes, similar to recent work by Barenboim and Goldenberg notably adapting Linial's algorithm to the distance-2 setting \cite{BG23}. How random colors trials reduce the uncolored part of the graph to two subgraphs of $O(\log n)$-degree follows from earlier work~\cite{BEPSv3}. We show the following lemma:

\begin{restatable}[Adaptation of \cite{BEPSv3}]{lemma}{LowDegreeReduction}
    \label{lem:low-deg-reduction}
    After \cref{alg:low-deg-reduction}, w.h.p., nodes have at most $O(\log n)$ neighbors of uncolored degree $\Omega(\log n)$. In particular, the graph can be partitioned into two subgraphs of maximum (real) degree $O(\log n)$. Executed on a virtual graph whose uncolored nodes have maximum pseudo-degree $\Deltalow$ and at most $\Dmaxcol$ uncolored neighbors, the algorithm runs in $\Ohat(\log \Deltalow \cdot \log \Dmaxcol)$ rounds. 
\end{restatable}

Subsequent sections assume that we are in one of the two low-degree subgraphs guaranteed by this lemma. We now explain how nodes sample random colors from their palette, through \alg{LowDegSampling} (\cref{alg:low-deg-sampling}).

\begin{algorithm}[hb]
    \caption{\alg{LowDegSampling}  \label{alg:low-deg-sampling}}

    \nonl \Input{A node $v$ in a virtual graph $H$ with partial coloring $\col$}

    \nonl \Output{A random color from a $(\deg_\col(v)+1)$-sized subset of $[\deg(v)+1] \setminus \col(N_H(v))$}

    Compute $\deg_\col(v)$ using the support tree, and sample $x$ uniformly at random in $[\deg_\col(v)+1]$.

    Initialize $S = [\deg(v)+1]$.

    \While{$\card{S} > 1$}{
        Broadcast $S$ to support tree $T(v)$.

        Let $S'$ be the first half of $S$ (of size $\ceil{\card{S}/2}$), $S'' = S \setminus S'$ its second half.

        Gather from support tree $T(v)$ the count $q' = \card{\set{(u,e): e \in E_H(u,v), \col(u) \in S'}}$

        \If{$\card{S'}-q' \geq x$}{
            $S \gets S'$
        }
        \Else{
            $S \gets S'', x \gets x - \card{S'} + q'$
        }
    }

    Output the unique color $c \in S$
\end{algorithm}

\begin{lemma}
    \label{lem:low-deg-sampling}

    A node $v$ executing \cref{alg:low-deg-sampling} outputs a random color unused by its neighbors. Each color has a probability $1/(\deg_\col(v) +1)$ or less of being sampled. The algorithm takes $\Ohat(\log \Deltalow)$ rounds.
\end{lemma}
\begin{proof}
    We show that for any choice of $x \in [\deg_\col(v) +1]$ at the beginning of \cref{alg:low-deg-sampling}, the algorithm outputs a color $c$ that is free for $v$, and that for two distinct choices $x'$ and $x''$ at the beginning of the algorithm, the resulting colors $c'$ and $c''$ are also distinct. As the mapping from selected number $x\in [\deg_\col(v)+1]$ to free colors $[\deg(v)+1]\setminus \col(N(v))$ is injective, each color in the image is sampled with the same $1/(\deg_\col(v)+1)$ probability.

    Throughout loop iterations, the invariants $\card{S} \geq x + \deg(v;H\cap \col^{-1}(S))$ and $x \geq 1$ are preserved, where $\deg(v;H\cap \col^{-1}(S))$ counts the degree of $v$ in the graph induced by colors in $S$, i.e., how many of $v$'s neighbors are colored by a color in $S$, counted with multiplicity. 
    
    The invariants hold at the beginning, since then we have $\card{S} = \deg(v)+1$, $1 \leq x \leq \deg_\col(v)+1$, and $\deg(v;H\cap \col^{-1}(S)) = \deg(v;H\cap \dom \col) = \deg(v) - \deg_\col(v)$. Consider now the evolution of $S$ in a loop iteration. $S$ is updated to $S'$ iff $\card{S'} \geq x + \deg(v;H\cap \col^{-1}(S'))$, as it is dependent on the test $\card{S'} - q' \geq x$, and $q' = \deg(v;H\cap \col^{-1}(S'))$. The invariants thus still hold when $S$ becomes $S'$. When $S$ becomes $S''$, $x$ is updated to $x'' = x-\card{S'}+\deg(v;H\cap \col^{-1}(S'))$. This only occurs if we had $\card{S'} - \deg(v;H\cap \col^{-1}(S')) < x$ just before, so $x'' \geq 1$. For the other invariant:
    \begin{align*}
        \card{S''} & = \card{S} - \card{S'}\\
        & \geq (x + \deg(v;H\cap \col^{-1}(S)))  -  \card{S'}\\
        & = x  +  \deg(v;H\cap \col^{-1}(S')) + \deg(v;H\cap \col^{-1}(S'')) - \card{S'}\\
        & = (x - \card{S'} + \deg(v;H\cap \col^{-1}(S')))  +  \deg(v;H\cap \col^{-1}(S''))\\ 
        & = x'' + \deg(v;H\cap \col^{-1}(S''))\ .
    \end{align*}
    Since the set $S$ ultimately reaches a size of $1$, we end up with $x=1$ and $\deg(v;H\cap \col^{-1}(S)) = 0$. From this last equation, we know that the color in $S$ is unused in the neighborhood of $v$, i.e., it is free for $v$ to use.

    Now, to show that distinct initial values of $x$ lead to different colors, let us write a transcript of the evolution of $S$ throughout the loop iterations. Note that the algorithm is deterministic once $x$ has been randomly chosen. Consider the word $w(x) = \set{0,1}^*$ s.t.\ $w_i(x)$ indicates whether $S$ became $S'$ or $S''$ in loop iteration number $i$. Suppose two choices of initial values $x$ and $x'$ lead to the same final set $S$, then, since $S'$ and $S''$ partition $S$ in each loop iteration, it must be that $w(x) = w(x')$. However, the sequence $w(x)$ also dictates the evolution of $x$ throughout the algorithm: unchanged when $S$ becomes $S'$ in iteration $i$ ($w_i(x) = 0$), becoming $x - \card{S'} + \deg(v;H\cap \col^{-1}(S'))$ otherwise ($w_i(x) = 1$). Therefore, two values $x$ and $x'$ s.t.\ $w(x) = w(x')$ are substracted the exact same values throughout the algorithm. This leads to a contradiction, as the values are initially distinct ($x \neq x'$), and they must both reach $1$ at the end the algorithm.

    The $\Ohat(\log \Deltalow)$ runtime is due to each loop iteration roughly halving the size of $S$, which has an initial size of $\deg(v)+1 \leq \Deltalow+1$.
\end{proof}

Showing that this color sampling procedure allows us to split uncolored nodes into two sets inducing $O(\log n)$-degree subgraphs is a straightforward adaption of an argument by Barenboim, Elkin, Pettie and Schneider~\cite[Lemma~5.4 and surrounding discussion]{BEPSv3}. We sketch the argument here, and defer the full proof to \cref{sec:log-deg-beps}.

\begin{proof}[Proof sketch of \cref{lem:low-deg-reduction}]
    The cited work shows that given a degree+1-list-coloring instance of maximum degree $\Delta$, $\Theta(\log \Delta)$ iterations of each uncolored node trying a random color from its palette guarantees that each node is adjacent to at most $O(\log n)$ uncolored nodes of uncolored degree $\Omega(\log n)$. We show that the same result is achieved when changing the setting to virtual graphs, with the sampling procedure of \cref{alg:low-deg-sampling}. At a high level, the idea is that nodes with $\Omega(\log n)$ neighbors have a constant fraction of them off w.h.p., and thus, also w.h.p., a node with $\Omega(\log n)$ uncolored neighbors of degree $\Omega(\log n)$ has a constant fraction of them get colored.
\end{proof}

\begin{remark}
    The runtime of \cref{alg:low-deg-sampling} can be reduced from $\Ohat(\log \deg(v))$ to $\Ohat(\frac{\log \deg(v)}{\log \log n} +1)$ by performing a $\log_{\deg(v)}(n)$-ary search instead of a binary search in the set of colors for a free color.
\end{remark}

\subsection{Learning Colors}
\label{sec:low-deg-learncolors}

We give a procedure for nodes to learn colors from their palette using a sort of $k$-ary search in the space of colors, similar to the color sampling procedure of the previous section. Again, this sort of search in the space of colors has been used recently in a work by Barenboim and Goldenberg~\cite{BG23}. More precisely, we give a procedure \alg{GrowPalette} that:
\begin{itemize}
    \item given a set $D_v \subseteq [\deg(v) +1]$ as input,
    \item in $\Ohat\parens*{1 + \card{D_v} \frac{\log \deg(v)}{ \log n} }$ rounds,
    \item has $v$ learn $\min\set*{\deg_\col(v) - \card{D_v},\frac{\log n}{ \log \deg(v)}  }$ free colors in $[\deg(v)+1] \setminus (D_v \cup \col(N(v)))$.
\end{itemize}
Our algorithm for learning a sufficient number of colors (\alg{LearnPalette}, \cref{alg:low-deg-learnpalette}) simply consists of multiple calls to \alg{GrowPalette} (\cref{alg:low-deg-growpalette}).

\begin{algorithm}[ht]
    \caption{\alg{LearnPalette}  \label{alg:low-deg-learnpalette}}

    \nonl \Input{An uncolored node $v$ with at most $\Dmaxcol$ uncolored neighbors}

    \nonl \Output{A set of free colors $D_v$ of size at least $\card{N_\col(v)}+1$}

    $D_v \gets \emptyset$

    \While{$\card{D_v} < \Dmaxcol+1$ and $\card{D_v} < \deg_\col(v)+1$}{
        $R_v \gets \alg{GrowPalette}(D_v)$

        $D_v \gets D_v \cup R_v$
    }

    Output $D_v$
\end{algorithm}

\begin{lemma}
    \label{lem:low-deg-learnpalette}
    After executing \alg{LearnPalette} (\cref{alg:low-deg-learnpalette}),
    a node knows $\card{N_\col(v)}+1$ colors from its palette. The algorithm takes $\Ohat(\log^2 \log n)$ rounds, assuming $\card{N_\col(v)} \in O(\log n)$.
\end{lemma}
\begin{proof}
    From the condition in the while loop, \alg{LearnPalette} terminates once the number of discovered colors is $\card{D_v} \geq \min(\deg_\col(v),\Dmaxcol)+1$. By assumption, $\Dmaxcol \geq \card{N_\col(v)}$. Pseudo-degree exceeds degree, so $\deg_\col(v) \geq \card{N_\col(v)}$. Therefore, the number of discovered colors is at least $\min(\deg_\col(v),\Dmaxcol)+1 \geq \card{N_\col(v)}+1$.
    
    Each call to \alg{GrowPalette} has a node $v$ learns $x \in \Theta(\log_{\deg(v)} n)$ new colors from its palette, except possibly in its last iteration, if $\card{D_v}$ got close to $\deg_\col(v)+1$. Therefore, as the algorithm discovers a total of $\min(\deg_\col(v),\Dmaxcol)+1 \geq \Dmaxcol+1$ colors, it is guaranteed to terminate after $\Theta(\Dmaxcol / \log_{\deg(v)} n)$ loop iterations.

    Performed on nodes of pseudo-degree at most $\Deltalow \in \log^{O(1)} n$, with at most $\Dmaxcol \in O(\log n)$ uncolored neighbors as guaranteed by \cref{lem:low-deg-reduction}, each loop iteration takes at most $O(\congestion \dilation \cdot \log \log n) = \Ohat(\log \log n)$ rounds by \cref{lem:low-deg-growpalette}, and the number of iterations is at most $O(\log \log n)$. In total, \alg{LearnPalette} has runtime $\Ohat(\log^2 \log n)$ rounds.
\end{proof}

We now present the main procedure of this section, \alg{GrowPalette}, and analyze it in \cref{lem:low-deg-growpalette}.

\begin{algorithm}[ht]
    \caption{\alg{GrowPalette}  \label{alg:low-deg-growpalette}}

    \nonl \Input{A set of already discovered free colors $D_v$}

    \nonl \Output{A set of additional free colors $R_v$}
    
    Broadcast $D_v$ to the support tree $T(v)$ \label[line]{line:learnpalette-send-discovered}

    Compute $\deg_\col(v)$ using the support tree, and let $x \gets \min\set*{\deg_\col(v) - \card{D_v},\frac{\log n}{ \log \deg(v)}  }$

    Let $S = ([\deg(v)+1] \setminus D_v)$.
    
    Initialize $\calS \gets \set{(S,x)}, R_v \gets \emptyset$ \hfill (we aim to discover $x$ free colors in $S$)

    \While{$\card{R_v} < x$}{
        Broadcast $\calS$ to support tree $T(v)$
    
        Let $\calS' \gets \emptyset$
        
        \ForEach{ $(S,y) \in \calS$ in parallel}{
            \If{$\card{S} \leq 1$}{
                $R_v \gets R_v \cup S$
            }\Else{
            Let $S'$ be the first $\ceil{\card{S}/2}$ colors in $S$, $S'' = S \setminus S'$ the other colors of $S$

            Gather from support tree $T(v)$ the count $q' = \card{\set{(u,e): e \in E_H(u,v), \col(u) \in S'}}$

            \If{$\card{S'} - q' \geq y$}{
                $\calS' \gets \calS' \cup \set{(S',y)}$
            }\ElseIf{$q' \geq \card{S'}$}{
                $\calS' \gets \calS' \cup \set{(S'',y'')}$
            }
            \Else{
                $y' \gets \card{S'} - q', y'' \gets y - y'$

                $\calS' \gets \calS' \cup \set{(S',y'),(S'',y'')}$
            }
            }
        
        }
        $\calS \gets \calS'$
    }
\end{algorithm}

\begin{restatable}{lemma}{GrowPalette}
    \label{lem:low-deg-growpalette}
    \Cref{alg:low-deg-growpalette} has each node $v$ learn $\min\set*{\deg_\col(v) - \card{D_v},\frac{\log n}{ \log \deg(v)}  }$ free colors in $[\deg(v)+1] \setminus D_v$ in $\Ohat\parens*{\log \deg(v) +\frac {\card{D_v} \log \deg(v)}{\log (n)}}$ rounds.
\end{restatable}

\begin{proof}
    At a high level, the algorithm maintains a set of ranges $\calS$ of subranges of $[\deg(v)+1]\setminus D_v$ in which to discover colors. Each range $S$ is appended a number of free colors that we can expect to find in it.

    We first have each node describe the set $[\deg(v)+1] \setminus D_v$ to its support tree. Since each color of $D_v$ can be described in $O(\log \deg(v))$ bits, $\Ohat\parens*{\frac {\card{D_v} \log \deg(v)}{\log (n)}}$ rounds suffice for this step.

    Once $[\deg(v)+1] \setminus D_v$ is known to the whole support tree, describing a subrange of $[\deg(v)+1] \setminus D_v$ (i.e., the intersection of this set with an interval of $[\deg(v)+1]$) can be done in $O(\log \deg(v))$ bits.

    We maintain sets of subranges with target number of colors $\calS = \set{(S_1,y_1),\ldots, (S_1,y_k)}$ with the property that each subrange $S_i$ has at least $y_i \geq 1$ free colors to discover in it.
    As we only aim to discover $x \leq O(\log n / \log \deg(v))$ colors in total ($\sum_i y_i = x$), there are never more than $x$ subranges, and the description of a full set of subranges only requires $O(x \cdot \log \deg(v)) \leq O(\log n)$ bits. As such, the broadcast of such a set $\calS$ to the whole support tree only takes $O(\congestion\dilation) = \Ohat(1)$ rounds.

    We show that for each $(S,y)$ contained in $\calS$, there always are $y$ free colors from $[\deg(v)+1] \setminus D_v$ to discover in $S$. This is true at initialization, since $x \leq \deg_\col(v) - \card{D_v}$. We show that the property of having $y$ extra colors in $S$ compared to the number of edges to nodes colored with a color from $S$ is maintained as we consider smaller sets. I.e., we always have $\card{S} - \card{\set{(u,e): e \in E_H(u,v), \col(u) \in S}} \geq y$.
    
    Consider now how the set $\calS'$ is constructed from $\calS$. For each $(S,y) \in \calS$, we split $S$ into two halves $S'$ and $S''$ of $S$, so $S = S' \sqcup S''$. We count how many edges $v$ has to neighbors with a color in $S'$. If that number $q'$ is small enough ($\card{S'} \geq q' + y$), $S'$ alone is guaranteed to contain $y$ colors so we only add $S'$ to the set of subranges for the next iteration (we add $(S',y)$ to $\calS'$). If $q'$ is larger than $\card{S'}$, then $S''$ is guaranteed to contain $y$ free colors, and we only add $S''$to the next iteration
    (we add $(S'',y)$ to $\calS'$). Otherwise, we add both subranges $S'$ and $S''$, using that $S'$ is guaranteed to contained at least $y' = \card{S'} - q'$ free colors (so we add $(S',y')$ to $\calS'$), and $S''$ is guaranteed to contain $y'' = y - y'$ free colors (so we add $(S'',y'')$ to $\calS'$).

    Let us now reason about the sizes of subranges. When we process a subrange of size $1$, we simply add the color it contains to the set of newly discovered colors $R_v$. When the subrange is larger, we split it into two and add either one or two of its halves to the next set of subranges. As a result, each iteration of the loop processing all subranges in $\calS$ halves the size of the largest subrange contained in $\calS$. Hence, the algorithm terminates after $O(\log \deg(v))$ iterations of this loop. 

    Finally, since we maintain the invariant that when we process $(S,y)$, $S$ contains at least $y>0$ free colors, when $S$ reaches size $1$, the unique color it contains is necessarily free for $v$. Since the sum $\sum_i y_i$ of the values $y_i$ s.t.\ $\calS = \set{(S_1,y_1),\ldots,(S_k,y_k)}$ only decreases by how many colors are added to the set $R_v$, in total, we discover $x$ colors in this process.
\end{proof}

\subsection{Finishing the Coloring}
\label{sec:finish-coloring-low-deg}

We now show how to finish the coloring in an uncolored subgraph of maximum degree $O(\log n)$, whose nodes $v$ all know at least one color more than they have uncolored neighbors in the subgraph, as guaranteed by \cref{lem:low-deg-reduction,lem:low-deg-learnpalette}. Our exact statement is provided in \cref{lem:low-deg-final}. The proof is essentially the same as that of \cite[Proposition~2, Proposition~1.1 in the full version]{FHN23}, adapted to our more general setting.

\begin{restatable}{lemma}{LowDegFinal}
    \label{lem:low-deg-final}
    Consider an uncolored virtual graph whose nodes $v$ have at most $\card{N_\col(v)} \leq \Dmaxcol \in O(\log n)$ uncolored neighbors in the graph, and s.t.\ each node knows a list $L'(v) \subseteq L_\col(v)$ of colors from its palette of size $\card{L'(v)} \geq \card{N_\col(v)}+1$, with $L'(v) \subseteq [\Deltalow+1]$ where $\Deltalow \in \log^{O(1)} n$.
    Then, there exists a $\Ohat(\log^3 \log \Deltalow \cdot \log \log n)$-round algorithm for coloring all low-degree uncolored nodes. 
\end{restatable}

\begin{algorithm}[ht]
    \caption{\alg{FinishColoring} \label{alg:finish-low-deg}}
    
    \nonl\Input{A coloring $\col$ and a set of nodes to color $V' \subseteq \Vlow \setminus \dom \col$ of maximum pseudo-degree $\Deltalow \in \log^{O(1)} n$, maximum uncolored degree $\max_{v \in V'} \card{N_{H[V']}(v)} \leq \Dmaxcol \in O(\log n)$, and s.t.\ $v$ knows a list of colors of size at least $\card{N_{H[V']}(v)}+1$ from its palette for all nodes $v \in V'$.}
    
    \nonl\Output{An extension $\collow \succeq \col$ s.t.\ $V' \subseteq \dom \collow$}

    Linial's algorithm adapted to the virtual graph setting. \label[line]{line:low-deg-Linial} \hfill (Adapted from~\cite{linial92,BG23})

    $O(\log \log n)$ random color trials. \label[line]{line:low-deg-shattering} \hfill (Shattering~\cite{BEPSv3})

    Simulation of the algorithm by Ghaffari and Kuhn \cite{GK21} on our virtual graph. \label[line]{line:low-deg-gk}
\end{algorithm}

\paragraph{Ultrafast Deterministic Coloring (Step \ref{line:low-deg-Linial}).}
Our first step is to compute an auxiliary $(\log^{O(1)} n)$-coloring of $V' \subseteq \Vlow\setminus \dom \col$, which will be useful later when running the deterministic algorithm by Ghaffari and Kuhn \cite{GK21}.

The seminal algorithm by Linial~\cite{linial92} for $O(\Delta^2)$-coloring in $O(\log^* n)$ rounds assumes that nodes have access to their neighbors' IDs. In the context of virtual graphs, in particular in power graphs as studied in \cite{BG23}, this is requires significant communication. Nonetheless, Barenboim and Goldenberg~\cite{BG23} showed that a $O(\Delta^4)$-coloring of $G^2$ can be computed in $O(\log \Delta \log^* n)$. The same idea also works in our setting of virtual graphs.

\begin{lemma}[Adaptation of {\cite[Theorem~3.3]{BG23}}]
    In a graph of maximum pseudo-degree $\Deltalow$, a proper $O(\Deltalow)$-coloring of its vertices can be computed in $\Ohat(\log \Deltalow \log^*n)$ \CONGEST rounds.
\end{lemma}
\begin{proof}
    Linial's algorithm works by phases, in which colors of a current coloring are mapped to sets of smaller colors, with the property that each node always has a color in its set of colors that is not in any of its neighbors' sets. Each node then update its color to one such color. When the graph to color is the same as the communication graph, the nodes can simply receive the current colors of all their neighbors to compute the associated sets, and find a new color for itself. The obstacle to performing this in virtual graphs (notably in power graphs) is that nodes cannot learn the current colors of all their neighbors.
    
    The key idea of \cite{BG23} is that a phase of Linial's algorithm can instead be performed by doing a binary search in one's set of colors. The binary search is done by having each node compute how many times colors from the first half of their set appear at neighboring nodes. The node then continues its search for a color to adopt in the first or second half of its set depending on which one had less colors appearing in the sets of neighboring nodes.

    In the distance-2 setting, nodes can perform this binary search by asking their direct neighbors about colors that appear in the sets of their distance-2 neighbors. In our setting, we simply do the same thing on the support tree. Indeed, in $O(\congestion \dilation) = \Ohat(1)$ rounds, each node can inform its support tree of its current color or set of colors, and in the same time it can aggregate from nodes in its support tree how many times a slice of its set of colors intersects with the sets of neighboring nodes. Since \cite{BG23} showed that the binary search only takes $O(\log \Delta)$ steps (doing some simple optimization once the sets from which colors are chosen are of size $O(\log n)$), as Linial's algorithm has $O(\log^* n)$ phases, and since our pseudo-degree is bounded by $\Deltalow$, we get the desired complexity.
\end{proof}

\paragraph{Shattering (Step \ref{line:low-deg-shattering}).}
The second step of our algorithm to finish the coloring of low-degree nodes is a few rounds of random color trials so as to \emph{shatter} the graph. In \cite{BEPSv3}, Barenboim, Elkin, Pettie and Schneider showed that after $\Theta(\log \Delta)$ random color trials, the uncolored part of a graph consists of connected components of size $O(\Delta^2 \log_\Delta n)$. Applied to our setting, we have the following lemma.

\begin{lemma}[Consequence of {\cite[Lemma 5.3]{BEPSv3}}]
    \label{lem:shattering}
    After $\Theta(\log \Dmaxcol) \leq O(\log \log n)$ random color trials by each uncolored node, the uncolored part of a graph of maximum degree $\Dmaxcol\in O(\log n)$ has connected components of size $O((\Dmaxcol)^2 \log n) \leq O(\log^3 n)$.
\end{lemma}

Each random color trial requires $O(\congestion \dilation) = \Ohat(1)$ rounds in our case, for each node to communicate the color it tries to its neighbors through its support tree, and to update its palette by aggregating over its support tree which of its $O(\log n)$ colors are no longer available. In total, this step takes $\Ohat(\log \log n)$ rounds.

\paragraph{Coloring the Small Connected Components (Step \ref{line:low-deg-gk}).}
In this final stage of the algorithm, the uncolored nodes of $V' \in \set{V_1,V_2}$ form connected components of size $\poly(\log n)$ and have at most $O(\log n)$ uncolored neighbors. Their colors are also all from $[\Deltalow + 1]$ where $\Deltalow \in \poly \log n$, and thus fit in $O(\log \log n)$ bits. They also have an auxiliary $O(\Deltalow^2)$-coloring from Step \ref{line:low-deg-Linial}, i.e., an auxiliary $\poly(\log n)$-coloring.

We now run a deterministic algorithm by Ghaffari and Kuhn \cite{GK21}. In the standard setting where the graph to color is also the communication network, given an auxiliary $O(\Delta^2)$-coloring of the nodes, this algorithm runs in $O(\log^2\calC \cdot \log N)$ rounds of bandwidth $O(\log \calC)$ on a graph of $N$ nodes, where $\calC$ is the size of the color space\footnote{Note that if the auxiliary $O(\Delta^2)$-coloring is not provided, their algorithm actually requires a few rounds of higher bandwidth to perform Linial's algorithm at the beginning. Also, note that $\calC \geq \Delta+1$}. The low degree of our virtual nodes allows us to simulate their algorithm in $\Ohat(\log^4 \log n)$ rounds.

\begin{lemma}[Simulation of \cite{GK21} on our small low-degree virtual subgraph]
    The uncolored connected components of size $\poly(\log n)$ and maximum degree $\Dmaxcol \in O(\log n)$ are colored in $\Ohat( \log^4 \log n)$ round during Step \ref{line:low-deg-gk} of \cref{alg:finish-low-deg}.
\end{lemma}
\begin{proof}
    Note that our connected components have size $N = \poly \log n$ and colors are of order $\calC \leq \Deltalow +1 = \log^{O(1)} n$. As such, the algorithm by Ghaffari and Kuhn \cite{GK21} would consist of $O(\log^2 \calC \log N) = O(\log^3 \log n)$ rounds if run directly on a connected component of $V'$ with the virtual graph as communication graph. In fact, the algorithm consists of $O(\log N)$ iterations of a procedure which colors a constant fraction of the graph, doing so in $O(\log^2 \calC)$ rounds.

    We simulate this algorithm on the virtual graph by spending $O(\congestion \dilation \cdot \log \calC)$ = $\Ohat(\log \calC)$ rounds on the communication graph for each of its rounds on the virtual graph. Indeed, as previously observed in \cite{FGHKN23}, the algorithm also works in the \emph{Broadcast} \CONGEST model, where in each round, each node sends the same message to all its neighbors. As a result, the nodes can send messages to their neighbors without specifying the intended recipient, and do not need to learn to which of their neighbors they have multiple links. When it comes to receiving the messages sent by neighbors, the support tree of a node may receive up to $\Theta(\Deltalow) \geq \omega(\log n)$ messages, but those contain at most $O(\log n)$ distinct messages. With at most $O(\log n)$ messages to receive, each of size $O(\log \calC)=O(\log \log n)$, forwarding all  messages to the root of the support tree can be done in $\Ohat( \log \log n)$ rounds, even if tagging each message with how many times it was received.

    While in a slightly different setting than that in which the algorithm was presented, the guarantee of the original algorithm that a constant fraction of the graph gets colored in $O(\log^2 \calC)$ rounds still holds, only with an updated runtime of $O(\congestion \dilation \cdot \log^3 \Deltalow) = \Ohat(\log^3 \log n)$. As $O(\log \log n)$ iterations of this suffice to color an entire connected component of size $\poly(\log n)$, we get the stated runtime of $\Ohat(\log^4 \log n)$.
\end{proof}

This completes the coloring of low-degree nodes, in a total of $\Ohat(\log^4 \log n)$ rounds.

\subsection{Faster algorithm on Graphs of Lower Degrees}
\label{sec:even-lower-deg}
Before closing this section, we remark that the algorithm presented here becomes faster on graphs of lower degrees, e.g., if the initial graph has maximum pseudo-degree $\Deltalow \in o(\log n)$. In particular:
\begin{itemize}
    \item Having each node deliver a $O(\log \Deltalow)$-sized message to all its up to $\Deltalow$ neighbors, can be done in $O((\frac{\congestion \Deltalow \log \Deltalow}{\log n}+1)\dilation)$ rounds.
    \item As a result, each node can receive the colors $\leq \Deltalow+1$ of all its neighbors in this runtime. This allows each node to learn its palette in $O(\congestion \dilation)$ rounds when $\Deltalow \in O(\log n / \log \log n$, and even faster for smaller $\Deltalow$.
    \item Shattering the graph (\cref{lem:shattering}) uses $O(\log \Deltalow)$ random color trials, i.e., a number shrinking with $\Deltalow$. Those can be done in the palette by learning the new palette between two trials. The remaining uncolored components have size $\poly(\Deltalow)\log n$, i.e., also shrinking with $\Deltalow$.
    \item The last step, the simulation of the deterministic algorithm by Ghaffari and Kuhn~\cite{GK21}, is also sped up. Each phase of the algorithm that colors a constant fraction of each connected component takes $O(\log^2 \Deltalow)$ rounds, during which a message of size $O(\log \Deltalow)$ has to be delivered from each node to all its neighbors.
\end{itemize}

This yields the following result:
\begin{remark}
    \label{rem:lower-degree}
    Virtual graphs of maximum pseudo-degree $\Deltalow \in O(\log n)$ can be colored faster than stated in \cref{prop:low}, in
    $O\parens*{\parens*{\frac{\congestion \Deltalow \log \Deltalow}{\log n}+1}\dilation \log^2 \Deltalow \log \log n}$ rounds.
\end{remark}

In particular, on graphs of maximum pseudo-degree $\Deltalow \in O(\log n / \log \log n)$ and constant congestion and dilation, the algorithm only takes $O(\log^3 \log n)$ rounds, i.e., is as fast as state-of-the-art algorithms in the standard \CONGEST setting.

\section{Open Problems}
\label{sec:open}

The most natural immediate question following our work is:
\begin{problem}
    Can  we color virtual graphs in $\congestion\dilation \cdot \poly(\log\log n)$ rounds using lists $\set{1, 2, \ldots, |N_H(v)|+1}$ for each $v \in V_H$? 
\end{problem}
The issue is with dense vertices whose anti-degree is hard to approximate accurately. In \cite{parti}, we show that it is possible to $\Delta+1$-color in $\Ohat(\log^* n)$ rounds when $\Delta = \max_{v} |N_H(v)|$ is the maximum number of neighbors (and $\Delta \gg \log^{21} n$). However, whether the technique used to approximate anti-degrees can be generalized to $|N(v)|+1$-coloring is unclear.
Using MultiColor Trials, it is possible to $(1+\eps)|N_H(v)|$-color in $\congestion\dilation \cdot \poly(\log\log n)$ rounds.

\begin{problem}
When is it possible to compute low-congestion support trees efficiently?
\end{problem}
We assumed that a support tree was given in $G$ for each node of $H$ (or could be easily deduced, as in the case of distance-2 coloring).
It is easy, per se, to find some support tree for each node, e.g., by BFS, but this could significantly affect the congestion. It is known \cite{ghaffari2016distributed,HaeuplerIZ16,KoganP21,GhaffariH21} that for some families of graph, one can compute embeddings with low congestion. Conversely, for some problems (such as MST), on general graphs $\Omega(\sqrt{n})$ congestion is unavoidable \cite{dassarma11}.
It is a highly interesting question whether low-congestion support trees could be computed efficiently for local problems.

\begin{problem}
    \label{open-problem:scheduling}
    Can we color virtual graphs in $O((\congestion + \dilation)\poly(\log\log n))$ rounds?
\end{problem}

Throughout the paper, our main goal was showing that coloring can be achieved in $\poly(\log \log n)$ rounds of broadcast and aggregation over the supports of the virtual nodes. We mostly ignored the runtime of these broadcast and aggregation operations, known to be achievable in $O(\congestion\dilation)$ rounds, and requiring $\Omega(\congestion+\dilation)$. The naive runtime is already optimal in some restricted cases (when $\congestion \in O(1)$ or $\dilation \in O(1)$), but not in general.
While $O(\congestion + \dilation)$ schedules are known to exist for standard packet routing (with fixed paths), our problem is a proper generalization of the usual routing scenario. We also need schedules that are distributedly computable.
\Cref{open-problem:scheduling} asks whether our subroutines can be performed faster, possibly also pipelined, certainly an exciting open question. It is essentially an independent scheduling question, despite its implications for the main results of our paper.

\begin{problem}
Can we $\Delta^{O(t)}$-color $G^t$ in $O(\Delta^{\floor{(t-1)/2}-\Omega(1)}\poly\log n)$ rounds of \congest?
\end{problem}
We showed that the complexity of coloring needs to grow linearly with the congestion, but this was only shown existentially for a specific class of instances. Can this dependence on congestion be avoided? In particular, the complexity of distance-3 coloring is a major open question, where congestion is necessarily linear in $\Delta$.

\newpage
\bibliographystyle{alpha}
\bibliography{partiirefs}

\newcommand{\etalchar}[1]{$^{#1}$}
\begin{thebibliography}{BYCHM{\etalchar{+}}20}

\bibitem[AA20]{AA20}
Noga Alon and Sepehr Assadi.
\newblock Palette sparsification beyond ({$\Delta+1$}) vertex coloring.
\newblock In {\em {Approximation, Randomization, and Combinatorial
  Optimization. Algorithms and Techniques {(APPROX/RANDOM)}}}, volume 176 of
  {\em LIPIcs}, pages 6:1--6:22. {LZI}, 2020.

\bibitem[ACK19]{ACK19}
Sepehr Assadi, Yu~Chen, and Sanjeev Khanna.
\newblock {Sublinear algorithms for {$(\Delta + 1)$} vertex coloring}.
\newblock In {\em the Proceedings of the ACM-SIAM Symposium on Discrete
  Algorithms (SODA)}, pages 767--786, 2019.
\newblock Full version at arXiv:1807.08886.

\bibitem[ALH{\etalchar{+}}23]{ALHZG_dc23}
Ioannis Anagnostides, Christoph Lenzen, Bernhard Haeupler, Goran Zuzic, and
  Themis Gouleakis.
\newblock Almost universally optimal distributed {Laplacian} solvers via
  low-congestion shortcuts.
\newblock {\em Distributed Computing}, 36(4):475--499, 2023.

\bibitem[BE13]{barenboimelkin_book}
Leonid Barenboim and Michael Elkin.
\newblock {\em Distributed Graph Coloring: Fundamentals and Recent
  Developments}.
\newblock Morgan \& Claypool Publishers, 2013.

\bibitem[BEPS16]{BEPSv3}
Leonid Barenboim, Michael Elkin, Seth Pettie, and Johannes Schneider.
\newblock The locality of distributed symmetry breaking.
\newblock {\em Journal of the ACM}, 63(3):20:1--20:45, 2016.

\bibitem[BG23]{BG23}
Leonid Barenboim and Uri Goldenberg.
\newblock Speedup of distributed algorithms for power graphs in the {CONGEST}
  model.
\newblock {\em CoRR}, abs/2305.04358, 2023.

\bibitem[Bra17]{Braverman_siamrev17}
Mark Braverman.
\newblock Interactive information complexity.
\newblock {\em {SIAM} Rev.}, 59(4):803--846, 2017.

\bibitem[BYCHM{\etalchar{+}}20]{BCMPP20}
Reuven Bar-Yehuda, Keren Censor-Hillel, Yannic Maus, Shreyas Pai, and Sriram~V
  Pemmaraju.
\newblock Distributed approximation on power graphs.
\newblock In {\em Proceedings of the 39th Symposium on principles of
  distributed computing}, pages 501--510, 2020.

\bibitem[CCDM23]{CCDM_icalp23}
Sam Coy, Artur Czumaj, Peter Davies, and Gopinath Mishra.
\newblock Optimal (degree+1)-coloring in congested clique.
\newblock In Kousha Etessami, Uriel Feige, and Gabriele Puppis, editors, {\em
  50th International Colloquium on Automata, Languages, and Programming,
  {ICALP} 2023, July 10-14, 2023, Paderborn, Germany}, volume 261 of {\em
  LIPIcs}, pages 46:1--46:20. Schloss Dagstuhl - Leibniz-Zentrum f{\"{u}}r
  Informatik, 2023.

\bibitem[CDP21]{CDP21}
Artur Czumaj, Peter Davies, and Merav Parter.
\newblock Simple, deterministic, constant-round coloring in congested clique
  and {MPC}.
\newblock {\em SIAM J. Comput.}, 50(5):1603--1626, 2021.

\bibitem[CFG{\etalchar{+}}19]{CFGUZ19}
Yi-Jun Chang, Manuela Fischer, Mohsen Ghaffari, Jara Uitto, and Yufan Zheng.
\newblock The complexity of {($\Delta+ 1$)} coloring in congested clique,
  massively parallel computation, and centralized local computation.
\newblock In {\em Proceedings of the 2019 ACM Symposium on Principles of
  Distributed Computing}, pages 471--480, 2019.
\newblock Full version at arXiv:1808.08419.

\bibitem[CLP20]{CLP20}
Yi-Jun Chang, Wenzheng Li, and Seth Pettie.
\newblock Distributed {($\Delta+1$)}-coloring via ultrafast graph shattering.
\newblock {\em SIAM Journal of Computing}, 49(3):497--539, 2020.

\bibitem[DHK{\etalchar{+}}11]{dassarma11}
Atish {Das Sarma}, Stephan Holzer, Liah Kor, Amos Korman, Danupon Nanongkai,
  Gopal Pandurangan, David Peleg, and Roger Wattenhofer.
\newblock Distributed verification and hardness of distributed approximation.
\newblock In {\em Proc.\ 43rd Symp. on Theory of Computing (STOC)}, pages
  363--372, 2011.

\bibitem[Doe20]{Doerr2020}
Benjamin Doerr.
\newblock {\em Probabilistic Tools for the Analysis of Randomized Optimization
  Heuristics}, pages 1--87.
\newblock Springer International Publishing, 2020.

\bibitem[DP09]{DP09}
Devdatt~P. Dubhashi and Alessandro Panconesi.
\newblock {\em Concentration of Measure for the Analysis of Randomized
  Algorithms}.
\newblock Cambridge University Press, 2009.

\bibitem[EPS15]{EPS15}
Michael Elkin, Seth Pettie, and Hsin{-}Hao Su.
\newblock (2{\(\Delta-1\)})-edge-coloring is much easier than maximal matching
  in the distributed setting.
\newblock In {\em Proceedings of the Twenty-Sixth Annual {ACM-SIAM} Symposium
  on Discrete Algorithms, {SODA} 2015, San Diego, CA, USA, January 4-6, 2015},
  pages 355--370, 2015.

\bibitem[FGG{\etalchar{+}}23]{FGGKR23}
Salwa Faour, Mohsen Ghaffari, Christoph Grunau, Fabian Kuhn, and V{\'{a}}clav
  Rozhon.
\newblock Local distributed rounding: Generalized to {MIS}, matching, set
  cover, and beyond.
\newblock In {\em Proceedings of the 2023 {ACM-SIAM} Symposium on Discrete
  Algorithms, {SODA} 2023, Florence, Italy, January 22-25, 2023}, pages
  4409--4447. {SIAM}, 2023.

\bibitem[FGH{\etalchar{+}}23]{FGHKN23}
Maxime Flin, Mohsen Ghaffari, Magn\'us~M. Halld\'orsson, Fabian Kuhn, and
  Alexandre Nolin.
\newblock Coloring fast with broadcasts.
\newblock In {\em the Proceedings of the ACM Symposium on Parallelism in
  Algorithms and Architectures (SPAA)}, pages 455--465. {ACM}, 2023.

\bibitem[FGH{\etalchar{+}}24]{FGHKN24}
Maxime Flin, Mohsen Ghaffari, Magn\'us~M. Halld\'orsson, Fabian Kuhn, and
  Alexandre Nolin.
\newblock A distributed palette sparsification theorem.
\newblock In {\em the Proceedings of the ACM-SIAM Symposium on Discrete
  Algorithms (SODA)}, 2024.

\bibitem[FGL{\etalchar{+}}21]{FGLPSY21}
Sebastian Forster, Gramoz Goranci, Yang~P. Liu, Richard Peng, Xiaorui Sun, and
  Mingquan Ye.
\newblock Minor sparsifiers and the distributed laplacian paradigm.
\newblock In {\em 62nd {IEEE} Annual Symposium on Foundations of Computer
  Science, {FOCS} 2022}, 2021.

\bibitem[FHK16]{fraigniaud16}
Pierre Fraigniaud, Marc Heinrich, and Adrian Kosowski.
\newblock {Local Conflict Coloring}.
\newblock In {\em the Proceedings of the Symposium on Foundations of Computer
  Science (FOCS)}, pages 625--634, 2016.

\bibitem[FHN20]{FHN20}
Pierre Fraigniaud, Magn{\'{u}}s~M. Halld{\'{o}}rsson, and Alexandre Nolin.
\newblock Distributed testing of distance-k colorings.
\newblock In {\em Structural Information and Communication Complexity - 27th
  International Colloquium, {SIROCCO} 2020}, 2020.

\bibitem[FHN23]{FHN23}
Maxime Flin, Magn{\'{u}}s~M. Halld{\'{o}}rsson, and Alexandre Nolin.
\newblock Fast coloring despite congested relays.
\newblock In {\em 37th International Symposium on Distributed Computing, {DISC}
  2023, October 10-12, 2023, L'Aquila, Italy}, 2023.

\bibitem[FHN24]{parti}
Maxime Flin, Magn{\'{u}}s~M. Halld{\'{o}}rsson, and Alexandre Nolin.
\newblock Decentralized distributed graph coloring: {Cluster} graphs, May 2024.
\newblock arXiv 2405.07725, in submission.

\bibitem[Fis17]{Fischer17}
Manuela Fischer.
\newblock Improved deterministic distributed matching via rounding.
\newblock In Andr{\'{e}}a~W. Richa, editor, {\em 31st International Symposium
  on Distributed Computing, {DISC} 2017, October 16-20, 2017, Vienna, Austria},
  volume~91 of {\em LIPIcs}, pages 17:1--17:15. Schloss Dagstuhl -
  Leibniz-Zentrum f{\"{u}}r Informatik, 2017.

\bibitem[FK23]{FK23}
Marc Fuchs and Fabian Kuhn.
\newblock List defective colorings: Distributed algorithms and applications.
\newblock In {\em 37th International Symposium on Distributed Computing, {DISC}
  2023, October 10-12, 2023, L'Aquila, Italy}, LIPIcs, 2023.

\bibitem[GG23]{GG23}
Mohsen Ghaffari and Christoph Grunau.
\newblock Faster deterministic distributed {MIS} and approximate matching.
\newblock In {\em Proceedings of the 55th Annual {ACM} Symposium on Theory of
  Computing, {STOC} 2023, Orlando, FL, USA, June 20-23, 2023}, pages
  1777--1790. {ACM}, 2023.

\bibitem[GGH{\etalchar{+}}23]{GHIR23}
Mohsen Ghaffari, Christoph Grunau, Bernhard Haeupler, Saeed Ilchi, and
  V{\'{a}}clav Rozhon.
\newblock Improved distributed network decomposition, hitting sets, and
  spanners, via derandomization.
\newblock In {\em Proceedings of the 2023 {ACM-SIAM} Symposium on Discrete
  Algorithms, {SODA} 2023, Florence, Italy, January 22-25, 2023}, pages
  2532--2566. {SIAM}, 2023.

\bibitem[GGR21]{GGR20}
Mohsen Ghaffari, Christoph Grunau, and Václav Rozhoň.
\newblock Improved deterministic network decomposition.
\newblock In {\em the Proceedings of the ACM-SIAM Symposium on Discrete
  Algorithms (SODA)}, 2021.

\bibitem[GH16]{ghaffari2016distributed}
Mohsen Ghaffari and Bernhard Haeupler.
\newblock Distributed algorithms for planar networks {II}: Low-congestion
  shortcuts, {MST}, and {Min-Cut}.
\newblock In {\em Proceedings of the twenty-seventh annual ACM-SIAM symposium
  on Discrete algorithms}, pages 202--219. SIAM, 2016.

\bibitem[GH21]{GhaffariH21}
Mohsen Ghaffari and Bernhard Haeupler.
\newblock Low-congestion shortcuts for graphs excluding dense minors.
\newblock In Avery Miller, Keren Censor{-}Hillel, and Janne~H. Korhonen,
  editors, {\em {PODC} '21: {ACM} Symposium on Principles of Distributed
  Computing, Virtual Event, Italy, July 26-30, 2021}, pages 213--221. {ACM},
  2021.

\bibitem[Gha15]{G15}
Mohsen Ghaffari.
\newblock Near-optimal scheduling of distributed algorithms.
\newblock In {\em Proceedings of the 2015 {ACM} Symposium on Principles of
  Distributed Computing, {PODC} 2015, Donostia-San Sebasti{\'{a}}n, Spain, July
  21 - 23, 2015}, pages 3--12, 2015.

\bibitem[GHZ21]{GHZ_stoc21}
Mohsen Ghaffari, Bernhard Haeupler, and Goran Zuzic.
\newblock Hop-constrained oblivious routing.
\newblock In Samir Khuller and Virginia~Vassilevska Williams, editors, {\em
  {STOC} '21: 53rd Annual {ACM} {SIGACT} Symposium on Theory of Computing,
  Virtual Event, Italy, June 21-25, 2021}, pages 1208--1220. {ACM}, 2021.

\bibitem[GK13]{ghaffari2013cut}
Mohsen Ghaffari and Fabian Kuhn.
\newblock Distributed minimum cut approximation.
\newblock In {\em International Symposium on Distributed Computing}, pages
  1--15. Springer, 2013.

\bibitem[GK21]{GK21}
Mohsen Ghaffari and Fabian Kuhn.
\newblock Deterministic distributed vertex coloring: Simpler, faster, and
  without network decomposition.
\newblock In {\em 62nd {IEEE} Annual Symposium on Foundations of Computer
  Science, {FOCS} 2021, Denver, CO, USA, February 7-10, 2022}, pages
  1009--1020. {IEEE}, 2021.

\bibitem[GKK{\etalchar{+}}18]{GKKLP18}
Mohsen Ghaffari, Andreas Karrenbauer, Fabian Kuhn, Christoph Lenzen, and Boaz
  Patt{-}Shamir.
\newblock Near-optimal distributed maximum flow.
\newblock {\em {SIAM} J. Comput.}, 47(6):2078--2117, 2018.

\bibitem[GZ22]{GZ22}
Mohsen Ghaffari and Goran Zuzic.
\newblock Universally-optimal distributed exact min-cut.
\newblock In {\em {PODC} '22: {ACM} Symposium on Principles of Distributed
  Computing, Salerno, Italy, July 25 - 29, 2022}, pages 281--291. {ACM}, 2022.

\bibitem[HIZ16]{HaeuplerIZ16}
Bernhard Haeupler, Taisuke Izumi, and Goran Zuzic.
\newblock Low-congestion shortcuts without embedding.
\newblock In George Giakkoupis, editor, {\em Proceedings of the 2016 {ACM}
  Symposium on Principles of Distributed Computing, {PODC} 2016, Chicago, IL,
  USA, July 25-28, 2016}, pages 451--460. {ACM}, 2016.

\bibitem[HKM20]{HKM20}
Magn{\'{u}}s~M. Halld{\'{o}}rsson, Fabian Kuhn, and Yannic Maus.
\newblock Distance-2 coloring in the {CONGEST} model.
\newblock In {\em {PODC} '20: {ACM} Symposium on Principles of Distributed
  Computing, Virtual Event, Italy, August 3-7, 2020}, pages 233--242, 2020.

\bibitem[HKMN20]{HKMN20}
Magn{\'{u}}s~M. Halld{\'{o}}rsson, Fabian Kuhn, Yannic Maus, and Alexandre
  Nolin.
\newblock Coloring fast without learning your neighbors' colors.
\newblock In {\em 34th International Symposium on Distributed Computing, {DISC}
  2020, October 12-16, 2020, Virtual Conference}, pages 39:1--39:17, 2020.

\bibitem[HKMT21]{HKMT21}
Magn{\'{u}}s~M. Halld{\'{o}}rsson, Fabian Kuhn, Yannic Maus, and Tigran
  Tonoyan.
\newblock Efficient randomized distributed coloring in {CONGEST}.
\newblock In {\em the Proceedings of the ACM Symposium on Theory of Computing
  (STOC)}, pages 1180--1193. {ACM}, 2021.
\newblock Full version at CoRR abs/2105.04700.

\bibitem[HKNT21]{HKNT21-arxiv}
Magn{\'{u}}s~M. Halld{\'{o}}rsson, Fabian Kuhn, Alexandre Nolin, and Tigran
  Tonoyan.
\newblock Near-optimal distributed degree+1 coloring.
\newblock {\em CoRR}, abs/2112.00604, 2021.
\newblock In STOC 2022.

\bibitem[HKNT22]{HKNT22}
Magn{\'{u}}s~M. Halld{\'{o}}rsson, Fabian Kuhn, Alexandre Nolin, and Tigran
  Tonoyan.
\newblock Near-optimal distributed degree+1 coloring.
\newblock In {\em STOC}, pages 450--463. {ACM}, 2022.

\bibitem[HMKS16]{disc16_coloring}
D.~Hefetz, Y.~Maus, F.~Kuhn, and A.~Steger.
\newblock A polynomial lower bound for distributed graph coloring in a weak
  {LOCAL} model.
\newblock In {\em Proc.\ 30th Symp.\ on Distributed Computing (DISC)}, pages
  99--113, 2016.

\bibitem[HN23]{HN23}
Magn{\'{u}}s~M. Halld{\'{o}}rsson and Alexandre Nolin.
\newblock Superfast coloring in {CONGEST} via efficient color sampling.
\newblock {\em Theor. Comput. Sci.}, 948:113711, 2023.

\bibitem[HNT22]{HNT22}
Magn{\'{u}}s~M. Halld{\'{o}}rsson, Alexandre Nolin, and Tigran Tonoyan.
\newblock Overcoming congestion in distributed coloring.
\newblock In {\em the Proceedings of the ACM Symposium on Principles of
  Distributed Computing (PODC)}, pages 26--36. {ACM}, 2022.

\bibitem[HPR{\etalchar{+}}24]{HPRSZ_arxiv24}
Bernhard Haeupler, Shyamal Patel, Antti Roeyskoe, Cliff Stein, and Goran Zuzic.
\newblock Polylog-competitive deterministic local routing and scheduling.
\newblock {\em CoRR}, abs/2403.07410, 2024.

\bibitem[HSS18]{HSS18}
David~G. Harris, Johannes Schneider, and Hsin-Hao Su.
\newblock {Distributed {($\Delta + 1$)}-coloring in sublogarithmic rounds}.
\newblock {\em Journal of the ACM}, 65:19:1--19:21, 2018.

\bibitem[KLL{\etalchar{+}}15]{KLLRX_siamcomp15}
Iordanis Kerenidis, Sophie Laplante, Virginie Lerays, J{\'{e}}r{\'{e}}mie
  Roland, and David Xiao.
\newblock Lower bounds on information complexity via zero-communication
  protocols and applications.
\newblock {\em {SIAM} J. Comput.}, 44(5):1550--1572, 2015.

\bibitem[KP21]{KoganP21}
Shimon Kogan and Merav Parter.
\newblock Low-congestion shortcuts in constant diameter graphs.
\newblock In Avery Miller, Keren Censor{-}Hillel, and Janne~H. Korhonen,
  editors, {\em {PODC} '21: {ACM} Symposium on Principles of Distributed
  Computing, Virtual Event, Italy, July 26-30, 2021}, pages 203--211. {ACM},
  2021.

\bibitem[KRZ21]{KRZ_arxiv21}
Christian Konrad, Peter Robinson, and Viktor Zamaraev.
\newblock Robust lower bounds for graph problems in the blackboard model of
  communication.
\newblock {\em CoRR}, abs/2103.07027, 2021.

\bibitem[Len13]{Lenzen_podc13}
Christoph Lenzen.
\newblock Optimal deterministic routing and sorting on the congested clique.
\newblock In Panagiota Fatourou and Gadi Taubenfeld, editors, {\em {ACM}
  Symposium on Principles of Distributed Computing, {PODC} '13, Montreal, QC,
  Canada, July 22-24, 2013}, pages 42--50. {ACM}, 2013.

\bibitem[Lin92]{linial92}
Nati Linial.
\newblock Locality in distributed graph algorithms.
\newblock {\em SIAM Journal on Computing}, 21(1):193--201, 1992.

\bibitem[LMR94]{LMR_combinatorica94}
Frank~Thomson Leighton, Bruce~M. Maggs, and Satish Rao.
\newblock Packet routing and job-shop scheduling in \emph{O}(congestion +
  dilation) steps.
\newblock {\em Combinatorica}, 14(2):167--186, 1994.

\bibitem[LMR99]{LMR_combinatorica99}
Frank~Thomson Leighton, Bruce~M. Maggs, and Andr{\'{e}}a~W. Richa.
\newblock Fast algorithms for finding \emph{O}(congestion + dilation) packet
  routing schedules.
\newblock {\em Combinatorica}, 19(3):375--401, 1999.

\bibitem[LMRR94]{LMRR_jal94}
Frank~Thomson Leighton, Bruce~M. Maggs, Abhiram~G. Ranade, and Satish Rao.
\newblock Randomized routing and sorting on fixed-connection networks.
\newblock {\em J. Algorithms}, 17(1):157--205, 1994.

\bibitem[LPSP06]{LPP06}
Zvi Lotker, Boaz Patt-Shamir, and David Peleg.
\newblock Distributed mst for constant diameter graphs.
\newblock {\em Distributed Computing}, 18(6):453--460, 2006.

\bibitem[Lub86]{luby86}
M.~Luby.
\newblock A simple parallel algorithm for the maximal independent set problem.
\newblock {\em SIAM Journal on Computing}, 15:1036--1053, 1986.

\bibitem[MPU23]{MPU23}
Yannic Maus, Saku Peltonen, and Jara Uitto.
\newblock Distributed symmetry breaking on power graphs via sparsification.
\newblock In {\em Proceedings of the 2023 {ACM} Symposium on Principles of
  Distributed Computing, {PODC} 2023, Orlando, FL, USA, June 19-23, 2023},
  pages 157--167. {ACM}, 2023.

\bibitem[MT20]{MT20}
Yannic Maus and Tigran Tonoyan.
\newblock Local conflict coloring revisited: Linial for lists.
\newblock In {\em the Proceedings of the International Symposium on Distributed
  Computing (DISC)}, pages 16:1--16:18, 2020.

\bibitem[Nao91]{Naor_siamdm91}
Moni Naor.
\newblock A lower bound on probabilistic algorithms for distributive ring
  coloring.
\newblock {\em {SIAM} J. Discret. Math.}, 4(3):409--412, 1991.

\bibitem[New91]{Newman91}
Ilan Newman.
\newblock Private vs. common random bits in communication complexity.
\newblock {\em Inf. Process. Lett.}, 39(2):67--71, 1991.

\bibitem[NS95]{naor95}
Moni Naor and Larry Stockmeyer.
\newblock What can be computed locally?
\newblock {\em SIAM J.\ on Comp.}, 24(6):1259--1277, 1995.

\bibitem[Ree98]{Reed98}
Bruce~A. Reed.
\newblock {\(\omega\)}, {\(\Delta\)}, and {\(\chi\)}.
\newblock {\em J. Graph Theory}, 27(4):177--212, 1998.

\bibitem[RG20]{RG20}
V{\'{a}}clav Rozho\v{n} and Mohsen Ghaffari.
\newblock Polylogarithmic-time deterministic network decomposition and
  distributed derandomization.
\newblock In {\em the Proceedings of the ACM Symposium on Theory of Computing
  (STOC)}, pages 350--363, 2020.

\bibitem[RGH{\etalchar{+}}22]{RozhonGHZL22}
V{\'{a}}clav Rozhon, Christoph Grunau, Bernhard Haeupler, Goran Zuzic, and
  Jason Li.
\newblock Undirected (1+\emph{{\(\epsilon\)}})-shortest paths via
  minor-aggregates: near-optimal deterministic parallel and distributed
  algorithms.
\newblock In {\em {STOC} '22: 54th Annual {ACM} {SIGACT} Symposium on Theory of
  Computing, Rome, Italy, June 20 - 24, 2022}, pages 478--487, 2022.

\bibitem[RY20]{RY_book20}
Anup Rao and Amir Yehudayoff.
\newblock {\em Communication Complexity: and Applications}.
\newblock Cambridge University Press, 2020.

\bibitem[Sha48]{Shannon48}
Claude Shannon.
\newblock A mathematical theory of communication.
\newblock {\em Bell System Technical Journal}, 1948.

\bibitem[SS13]{SS13}
Stefan Schmid and Jukka Suomela.
\newblock Exploiting locality in distributed {SDN} control.
\newblock In {\em Proceedings of the second ACM SIGCOMM Workshop on Hot Topics
  in Software-Defined Networking}, pages 121--126, 2013.

\bibitem[SW10]{SW10}
Johannes Schneider and Roger Wattenhofer.
\newblock A new technique for distributed symmetry breaking.
\newblock In {\em the Proceedings of the ACM Symposium on Principles of
  Distributed Computing (PODC)}, pages 257--266. {ACM}, 2010.

\end{thebibliography}

\appendix

\section{Concentration Bounds}
\label{app:concentration}
\begin{lemma}[Chernoff bounds]\label{lem:basicchernoff}
Let $\{X_i\}_{i=1}^r$ be a family of independent binary random variables with $\Pr[X_i=1]=q_i$, and let $X=\sum_{i=1}^r X_i$. For any $\delta>0$, 
\[
\Pr \range*{ |X-\Exp[X]|\ge \delta\Exp[X] } \le 2\exp(-\min(\delta,\delta^2) \Exp[X]/3) \ .
\]
\end{lemma}

We use the following variants of Chernoff bounds for dependent random variables. The first one is obtained, e.g., as a corollary of Lemma 1.8.7 and Thms.\ 1.10.1 and 1.10.5 in~\cite{Doerr2020}.

\begin{lemma}[Martingales \cite{Doerr2020}]\label{lem:chernoff}
Let $\{X_i\}_{i=1}^r$ be binary random variables, and $X=\sum_i X_i$.
If $\Pr[X_i=1\mid X_1=x_1,\dots,X_{i-1}=x_{i-1}]\le q_i\le 1$, for all $i\in [r]$ and $x_1,\dots,x_{i-1}\in \{0,1\}$ with $\Pr[X_1=x_1,\dots,X_r=x_{i-1}]>0$, then for any $\delta>0$,
\[\Pr\event*{X\ge(1+\delta)\sum_{i=1}^r q_i}\le \exp\parens*{-\frac{\min(\delta,\delta^2)}{3}\sum_{i=1}^r q_i}\ .\]
If $\Pr[X_i=1\mid X_1=x_1,\dots,X_{i-1}=x_{i-1}]\ge q_i$, $q_i\in (0,1)$, for all $i\in [r]$ and $x_1,\dots,x_{i-1}\in \{0,1\}$ with $\Pr[X_1=x_1,\dots,X_r=x_{i-1}]>0$, then for any $\delta\in [0,1]$,
    \begin{equation}\label{eq:chernoffmore}
    \Pr \range*{ X\le(1-\delta)\sum_{i=1}^r q_i } 
    \le \exp\parens*{ -\frac{\delta^2}{2}\sum_{i=1}^r q_i }\ .
    \end{equation}
\end{lemma}

A function $f(x_1,\ldots,x_n)$ is  \emph{$c$-Lipschitz} iff changing any single $x_i$ affects the value of $f$ by at most $c$, and $f$ is  \emph{$r$-certifiable} iff whenever $f(x_1,\ldots,x_n) \geq s$ for some value $s$, there exist $r\cdot s$ inputs $x_{i_1},\ldots,x_{i_{r\cdot s}}$ such that knowing the values of these inputs certifies $f\geq s$ (i.e., $f\geq s$ whatever the values of $x_i$ for $i\not \in \{i_1,\ldots,i_{r\cdot s}\}$).
\begin{lemma}[Talagrand's inequality~\cite{DP09}]
\label{lem:talagrand}
Let $\{X_i\}_{i=1}^n$ be $n$ independent random variables and $f(X_1,\ldots,X_n)$ be a $c$-Lipschitz $r$-certifiable function; then for $t\geq 1$,
\[\Pr\event*{\abs*{f-\Exp[f]}>t+30c\sqrt{r\cdot\Exp[f]}}\leq 4 \cdot \exp\parens*{-\frac{t^2}{8c^2r\Exp[f]}}\]
\end{lemma}

This implies the following lemma, as shown in \cite{HKNT21-arxiv}.

\begin{lemma}
\label{lem:talagrand-difference}
Let $\set*{X_i}_{i=1}^n$ be $n$ independent random variables. Let $\set*{A_j}_{j=1}^k$ and $\set*{B_j}_{j=1}^k$ be two families of
events that are functions of the $X_i$'s. Let $f=\sum_{j\in[k]} \mathbb{I}_{A_j}$, $g=\sum_{j\in[k]} \mathbb{I}_{A_j \cap \overline{B}_j}$,\footnote{$\mathbb{I}$ denotes the indicator random variable of an event.} and $h=f-g$ be such that $f$ and $g$ are $\Theta(1)$-Lipschitz and $\Theta(1)$-certifiable w.r.t.\ the $X_i$'s, and $\Exp[h] \geq \alpha \Exp[f]$ for some constant $\alpha \in (0,1)$. Let $\delta \in (0,1)$. Then for $\Exp[h] \geq \Omega(\delta^{-2}\alpha^{-1})$ large enough:
\[\Pr\event*{\abs*{h - \Exp[h]} > \delta \Exp[h]} \leq \exp(-\Omega(\Exp[h]))\]
\end{lemma}

\section{Background}
\label{sec:background}

\subsection{Color Trials}
We state here some classical symmetry breaking algorithms with additional assumptions for virtual graphs. The main difference with \local or \congest version of these algorithm is that vertices sample in color space $\calC(v)$ instead of their palette.

\begin{restatable}[Random Color Trial]{lemma}{LemTryColor}
    \label{lem:try-color}
    Let $\gamma \in (0,1)$ be universal constants known to all nodes.
    Let $\col$ be a coloring, $S \subseteq V \setminus \dom\col$ a set of uncolored nodes, and sets $\calC(v) \subseteq [\deg(v)+1]$ for each $v\in S$ such that
    \begin{enumerate}
        \item $v$ can sample a uniform color in $\calC(v)$ in $O(1)$ rounds,
        \item $|\calC(v)| \geq \Theta(\gamma^{-1}\log n)$, 
        \item\label[part]{part:rct-clique-palette} 
        $|L_\col(v) \cap \calC(v)| \ge \gamma |\calC(v)|$, and
        \item\label[part]{part:rct-min-palette}
        $|L_\col(v) \cap \calC(v)| \ge \gamma |N_\col(v) \cap S|$.
    \end{enumerate}
    Let $\psi \succeq \col$ be the coloring produced by \trycolor.
    Then, w.h.p., each $w \in V_H$ has uncolored degree in $S$
    \[
    |N_{\psi}(w) \cap  S| \le 
    \max\set*{ 
        (1-\gamma^4/64) |N_\col(w) \cap S| ,~
        \Theta(\gamma^{-4} \log n )  
    } \ .
    \]
    The algorithm ends after $\Ohat(1)$ rounds and $\psi(v) \in \calC(v)$ for all $v\notin\dom\col$.
\end{restatable}

The MultiColorTrial in \cref{lem:mct} is adapted from \cite{HN23} to sample colors from a restricted known color space. 

\begin{lemma}[MultiColorTrial, adapted from~\cite{HN23}]
    \label{lem:mct}
    Let $\col$ be a (partial) coloring of $H$, $S \subseteq V_H \setminus \dom\col$, and $\mathcal{C}(v) \subseteq [\deg(v)+1]$ be a color space for each node. 
Suppose that there exists some constant $\gamma > 0$ known to all nodes such that
    \begin{enumerate}
        \item $\mathcal{C}(v)$ is known to all machines in $V(v)$; and
        \item\label[cond]{part:mct-slack} 
        $|L_\col(v) \cap \calC(v)| - |N_\col(v) \cap S| 
        \ge \max\set{2 |N_\col(v) \cap S|, \Theta(\log^{1.1} n)} + \gamma \card{\calC(v)}$.
    \end{enumerate}
    Then, there exists an algorithm computing a coloring $\psi \succeq \col$ such that, w.h.p., all nodes of $S$ are colored and $\psi(v) \in \calC(v)$ for each $v\in S$. The algorithm runs in $\Ohat(\gamma^{-1} \log^* n)$ rounds.
\end{lemma}

\subsection{Fingerprinting}

\begin{lemma}[\cite{parti}]
    \label{lem:concentration-fingerprint}
    \newcommand{\kstar}{K_{\star}}
        Consider $t, d\geq 1$ integers and $t \times d$ independent geometric random variables $(X_{i,j})_{i \in [t], j \in [d]}$ of parameter $1/2$. For each $i \in [t]$, let $Y_i = \max_{j \in [d]}(X_{i,j})$. For each integer $k$, let $Z_k = \card{\set{i \in [t]: Y_i < k }}$.
Let $\kstar = \min \set{k : Z_k \geq (27/40)t}$ and define:
        \[\hat{d} \eqdef \frac
        {\ln\parens*{Z_{\kstar} / t}}
        {\ln(1-2^{-\kstar})}
        \ .\]
        Then, for any $\xi \in (0,1/4)$, 
        $\Pr[\abs{ d- \hat{d} }
        \leq \xi d] \geq 1-6\exp(-{\xi^{2}t}/{50})$.
\end{lemma}

\begin{lemma}[\cite{parti}]
    \label{lem:fingerprint-encoding}
    Let $Y_i = \max_{j\in[d]} X_{i,j}$ where $(X_{i,j})_{i\in[t], j\in[d]}$ are independent geometric variables of parameter $1/2$.
    With probability $1-2^{-t/10+1}$,  the sequence of values $(Y_i)_{i\in[t]}$ 
can be described in $O(t + \log \log d)$ bits.
\end{lemma}

\begin{lemma}[\cite{parti}]
    \label{lem:fingerprinting}
    Let $\xi \in (0,1/4)$ and, for each $v\in V_H$, predicates $P_v: N(v) \to \set{0,1}$ such that if $P_v(u)=1$, there exists $w\in V(v) \cap V(u)$ that knows it. There is a $\Ohat(\xi^{-2})$-round algorithm for all nodes to estimate $\card{N_H(v) \cap P_v^{-1}(1)}$ with high probability within a multiplicative factor $(1\pm \xi)$.
\end{lemma}

\section{Coloring Cabals}
\label{sec:color-cabals}
We sketch the algorithm for coloring cabals. This results in \cref{prop:cabals}, which only assumes that nodes of $\Vcabal$ are uncolored.

\PropCabals*

The algorithm for coloring cabals (mostly) follows the structure of \cite{parti}, and hence we merely sketch the proof and refer the reader to \cite{parti} for more details. Intuitively, the reason cabals are simpler to color than non-cabals is because they do not rely on slack generation. Instead, vertices get slack from \emph{put-aside sets}. The main difference from \cite{parti} is Step \ref{line:cabals-slice-color} where we reduce uncolored degrees to $O(\log n)$. In \cite{parti}, as we used $\Delta+1$-colors, we could run MultiColorTrial directly after the synchronized color trial (Step \ref{line:cabals-sct}).

\begin{algorithm}
    \caption{Cabals\label{alg:coloring-cabals}}

    \nonl Let $r' \eqdef 150 \ell$, where $\ell = C_1\log^{1.2} n$ is as described in \cref{eq:params}.

    \alg{ColorfulMatching}.
    \label[line]{line:cabals-CM}

    \coloringoutliers with $\calC(v) = [\deg(v) + 1] \setminus [r']$.
    \label[line]{line:cabals-outliers}

    \computePutAside $P_K \subseteq I_K$.
    \label[line]{line:cabals-compute-put-aside}

    \sct with $S_K = K \setminus (\dom \col \cup P_K)$
    \label[line]{line:cabals-sct}

    \slicecolor with $\calC(v) = [\deg(v)+1] \setminus [r']$
    \label[line]{line:cabals-slice-color}

    \nonl Let $\mathscr{L}_1, \ldots, \mathscr{L}_{O(\log\log n)}$ be the layers produced by \slicecolor

    \For{$i = O(\log\log n)$ to $1$\label[line]{line:cabals-loop-layers}}{
        \multitrial with $\calC(v) = [r']$ in $\mathscr{L}_i$
        \label[line]{line:cabals-mct}
    }

    \colorPutAside
    \label[line]{line:cabals-put-aside-sets}
\end{algorithm}

\begin{proof}[Proof Sketch of \cref{prop:cabals}]
We emphasize that the number of reserved colors $r' \eqdef 150\ell$ \emph{in cabals} is different from the number of reserved colors $r$ we used on non-cabals (\cref{eq:params}).

Let us go over steps of \cref{alg:coloring-cabals}.

\paragraph{Colorful Matching (Step \ref{line:cabals-CM}).}
In cabals where $a_K \geq \Omega(\log n)$, we run the colorful matching algorithm from \cref{lem:colorful-matching-high}. When $a_K \leq O(\log n)$, we compute a coloring such that at least $(1-O(\eps))|K|$ vertices $v\in K$ have $a_v\leq M_K \eqdef |K\cap\dom\col| - |\col(K)|$. See \cite[Section 6]{parti} for more details.

\paragraph{Inliers \& Outliers (Step \ref{line:cabals-outliers}).}
Observe that since $\Vcabal$ was uncolored, the vertices can compute $M_K$ exactly using the \refQuery. 
We define inliers slightly differently from non-cabals when $a_K \leq O(\log n)$.
\begin{itemize}
    \item When $a_K \geq \Omega(\log n)$, let $I_K$ be as in \cref{eqdef:inliers-non-cabals}.
    \item When $a_K \leq O(\log n)$, let $I_K = \set{ v \in K: e_v \leq 20 e_K \text{ and } a_v \leq M_K }$.
\end{itemize}
The difference in the case where $a_K \leq O(\log n)$ is that the colorful matching is not ensured to be large compared to $a_K$. However, using the put-aside sets, vertices do not need the extra $a_K$ slack given by the colorful matching. Hence it suffices that $a_v \leq M_K$ to ensure the clique palette contains enough colors.

Since each vertex $v$ knows $e_v$, $a_v$ (thus $e_K$ and $a_K$), and $M_K$, it knows whether it is an inlier or outlier. Also note that by Markov's inequality, each cabal contains at least $(0.95 - 10\eps)|K|$ inliers. Outliers are colored in $O(\log^* n)$ rounds exactly like those in non-cabals (using the $\Omega(|K|)$ temporary slack from uncolored inliers).

\paragraph{Computing Put-Aside Sets (Step \ref{line:cabals-compute-put-aside}).}
Like in \cite{HKNT22,parti}, we compute put-aside sets $P_K$ in cabals through sampling. Their key properties are that: $P_K \subseteq I_K \setminus \dom\col$ (they are uncolored inliers); there are 
$|P_K| = r' + \ell$ slightly more than reserved colors; there are no edges between $P_K$ and $P_{K'}$ for $K'\neq K$; and at most $|K|/100$ nodes of $K$ have neighbors in $\bigcup_{K'\in \Kcabal\setminus \set{K}} P_{K'}$. The only difference with \cite{parti} is that we use slightly larger put-aside sets, which is needed to run slice color (Step \ref{line:cabals-slice-color}).

\paragraph{\refSCT (Step \ref{line:cabals-sct}).}
\cref{lem:clique-palette-sct,lem:sct} apply also in cabals. Hence, after this step each cabal contains $O(e_K + a_K) + O(\log n) \leq O(a_K + \ell)$ uncolored vertices. Adding the $e_v \leq O(\ell)$ external neighbors (since $e_K \leq \ell$), we get that uncolored degrees are $O(a_K + \ell)$.

\paragraph{Reducing Uncolored Degrees (Step \ref{line:cabals-slice-color}).}
Let $V' = \Vcabal \setminus \bigcup_K P_K$. Steps \ref{line:cabals-slice-color} and \ref{line:cabals-mct} color all vertices in $V'$ as we explain now.
By \cref{lem:slice-color,lem:sampler}, it suffices to show that $|L_\col(v) \cap L_\col(K) \setminus [r']| \geq |N_\col(v) \cap V'| + \Theta(a_K + \ell)$.
Let us first consider inliers $v\in I_K$ such that $a_K \leq O(\log n)$.
Then, $a_v \leq M_K$ and one can show that
\[
    |L_\col(K) \cap L_\col(v)| 
    \geq |(K \cup N(v)) \setminus \dom \col| 
    = |N_\col(v) \cap V'| + |P_K|  \ .
\]
If $a_K \geq \Omega(\log n)$, then inliers get additional slack because $a_v \leq 20 a_K$ while the colorful matching has size $M_K \geq \Omega(a_K/\eps)$. Thus, one can show
\[
|L_\col(K) \cap L_\col(v)|
\geq |(K \cup N(v)) \setminus \dom \col| + a_K
= |N_\col(v) \cap V'| + a_K + |P_K|  \ .
\]
Since $|P_K| = r' + \ell$, 
one can verify that in all cabals
\[
    |L_\col(K) \cap L_\col(v) \setminus [r']| 
    \geq |N_\col(v) \cap V'| + (\ell + a_K)/2 \ .
\]
Let each $v\in V'$ compute $\tilde{d}(v)$ a 2-approximation of $N_\col(v) \cap V'$ using the fingerprinting algorithm (\cref{lem:fingerprinting}) and let $d(v) \eqdef 2\tilde{d}(v) \geq |N_\col(v) \cap V'| \geq \tilde{d}(v)/2 = d(v)/4$. Let $s(v) = 2(\ell + a_K)$ and note that since uncolored degrees are $O(a_K + \ell)$ we have $s(v) \geq \Omega(d(v))$. Finally, we have 
$
    |L_\col(K) \cap L_\col(v) \setminus [r']| \geq (d(v) + s(v))/4
$, which means conditions of Slice Color with $\kappa = 4$ are verified when we use the sampler from \cref{lem:sampler}.
By \cref{lem:slice-color}, after $\Ohat(\log\log n)$ rounds, we obtain layers $\mathscr{L}_1, \ldots, \mathscr{L}_{O(\log\log n)} \subseteq V'$ such that for every vertex $v\in \mathscr{L}_i$, we have $|N_\col(v) \cap \mathscr{L}_{\geq i}| \leq O(\log n)$.

\paragraph{\refMCT (Steps \ref{line:cabals-loop-layers} and \ref{line:cabals-mct}).}
We color the remaining vertices in $V'$.
Since we did not use any reserved color in cabals, i.e., $\col(\Vcabal) \cap [r'] = \emptyset$, the only reason a vertex might lose colors in $[r']$ is when it is used by an external neighbor. Hence, each $v\in V'$ has
\[
    |L_\col(v) \cap [r]| \geq r' - e_v \geq 100\ell \geq 3 |N_\col(v) \cap \mathscr{L}_{\geq i}| + \Theta(\log^{1.2} n) \ .
\]
Hence, running \refMCT layer by layer colors all remaining vertices in $V'$ in $\Ohat(\log\log n \cdot \log^* n)$ rounds.

\paragraph{Coloring Put-Aside Sets (Step \ref{line:cabals-put-aside-sets}).}
Only put-aside sets remain to color. This can be done in $\Ohat(1)$ rounds using the recoloring procedure described in \cite[Section 7]{parti}.
This completes the coloring of $\Vcabal$.
\end{proof} 
\section{Lower Bounds: Details}
\label{sec:lower-bound}
In this section, we fill in the missing details in our lower bounds arguments sketched in \cref{sec:lower-bound-overview}. Our end goal is the following theorem involving the congestion, dilation, and bandwidth:

\TheoremLowerBound*

We claim little novelty here, and provide these lower bounds chiefly to paint a fuller picture of the relevance of our algorithms. 
The idea to study information complexity through zero-communication protocols has been in the literature for a while now~\cite{KLLRX_siamcomp15}, and recently, lower bounds were proved in a different model\footnote{The model consists of $k$ machines with a shared blackboard, that receive an input graph whose description is randomly distributed among them.} using essentially the same arguments as our congestion lower bound \cite[Theorem 1]{KRZ_arxiv21}.
The type of argument deployed in the dilation lower bound has also appeared multiple times in the past, e.g., to connect the complexity of computing ruling sets to that of computing maximal independent sets.

\subsection{Lower-bound with respect to congestion}
\label{sec:lower-bound-congestion}

We prove the congestion-related part of our lower bound through a reduction from communication complexity.

Recall the definition of $\MCOL_k$ given in \cref{def:congestion-gadget-task}.
Alice has a set of $2k$ nodes $\set{\vL{1},\ldots,\vL{2k}}$, and Bob has a set of nodes $\set{\vR{1},\ldots,\vR{2k}}$, such that for each $i$ there is an edge between $\vL{i}$ and $\vR{i}$. Alice and Bob each receive as input a perfect matching over their nodes, and must output a $3$-coloring of their nodes such that the overall coloring is proper.
We denote by $X$ and $Y$ the sets of input of Alice and Bob, i.e, each element $x \in X$ represents a unique perfect matching over the nodes $\vL{1},\ldots,\vL{2k}$.

We show that without communication, the players necessarily fail the task with some non-negligible probability.

\NoCommunicationError*

\begin{proof}
    Consider any deterministic protocol for this task -- since the error is measured w.r.t.\ the input distribution, for any randomized protocol achieving error $\eps$, there exists a deterministic protocol achieving the same (or lower) error. We show that there necessarily is an index $i\in[8]$ such that both Alice and Bob do not always output the same color on their $i$th node.
Since there are only $3$ colors to choose from, this means that the players necessarily output the same color at this index on some pairs of inputs, and thus fail to properly color the graph.

    Without loss of generality, we can assume that Alice and Bob always assign colors s.t.\ the coloring is at least valid w.r.t\ the matchings they received, and the error only comes from the edges $\vL{i}\vR{i}$ for each $i \in [8]$. Indeed, for any protocol without this property, a protocol with the property that has at most the same error probability can always be constructed.

    Let us denote by $\pA(i,c)$ the probability (over her random input $X$) that Alice outputs $c$ at a given index $i$. For any input $x \in X$ of Alice, let $\pA(i,c \mid x)$ be this same probability conditioned on Alice's input being $x$. Note that since the protocol is deterministic, $\pA(i,c \mid x) \in \set{0,1}$. We have $\pA(i,c) = \sum_{x \in X} \pA(i,c \mid x) \Pr[X = x] = \frac{1}{\card{X}}\sum_{x \in X} \pA(i,c \mid x)$. Let $\pB(i,c)$ and $\pB(i,c \mid y)$ be similarly defined from Bob's behavior in the protocol. For any pair of indices $i\neq j$, the vertices $\vL{i}$ and $\vL{j}$ are connected in $1/7$ of input matchings of Alice. Therefore, the probability that Alice outputs the same color on both indices is bounded as follows
    \begin{align}
        \forall i,j \in \binom{[8]}{2},
        &&\frac{1}{\card{X}}\sum_{x\in X}\sum_{c=1}^3 \pA(i,c \mid x) \cdot \pA(j,c \mid x) &\leq \frac{6}{7}\ . \label{eq:not-too-often-equal}
    \end{align}

    Suppose there exists $4$ distinct indices $i$ s.t.\ $\max_{c\in[3]} \pA(i,c) > \frac{13}{14}$. This would imply that there exist two indices $i\neq j$ and a color $c$ s.t.\ $\pA(i,c) > \frac{13}{14}$ and $\pA(j,c) > \frac{13}{14}$. But then, for those two indices, we would have $\frac{1}{\card{X}}\sum_{x\in X}\sum_{c=1}^3 \pA(i,c \mid x) \cdot \pA(j,c \mid x) > \frac{6}{7}$, contradicting \cref{eq:not-too-often-equal}. 
    
    Consequently, there exists at least $5$ indices s.t.\ $\max_{c\in[3]} \pA(i,c) \leq \frac{13}{14}$. The same argument can be done on Bob's side, yielding that on at least $5$ indices, we have $\max_{c\in[3]} \pB(i,c) \leq \frac{13}{14}$. Since there are only $8$ indices, there are at least $2$ indices for which both $\max_{c\in[3]} \pA(i,c)$ and $\max_{c\in[3]} \pB(i,c)$ are bounded by $\frac{13}{14}$.

    Consider such an index $i$. The probability of a color conflict at this index is $\sum_{c =1}^3\pA(i,c)\cdot \pB(i,c)$. For an inner product of positive terms like this, the sum is minimized if the terms $\pA$ and $\pB$ in reverse order with respect to each other, i.e., if $\pA(i,c+1) \geq \pA(i,c)$ and $\pB(i,c+1) \leq \pB(i,c)$ for each $c \in [2]$, and if the weights are as unequal as possible. In our setting, this means the value is minimized with $\parens*{\pA(i,c)}_{c\in[3]} = (0,\frac{1}{14},\frac{13}{14})$ and $\parens*{\pB(i,c)}_{c\in[3]} = (\frac{13}{14},\frac{1}{14},0)$, which gives a probability of error of at least $\frac{1}{196}$.
\end{proof}

\paragraph{Non-negligible information complexity.}
The study of interactive information complexity in communication complexity has intuitively been an enterprise of generalizing seminal results for non-interactive communication (notably Shannon's~\cite{Shannon48}), and an attempt to understand the flow of information in interactive protocols. Information complexity intuitively captures the amount of information that is revealed about the inputs in an interactive protocol, either between players (internal information cost) or to the outside world (external information cost). For an introduction to information theory concepts and their relevance to the study of communication complexity protocols, we recommend \cite[Chapter 6]{RY_book20}.

\paragraph{Background on Information Theory.}
Let $A,B,C$ be three random variables of respective support $\calA,\calB,\calC$. Let us slightly abuse notation by taking the convention that $0 \cdot \log (1 / 0) = 0$, $0 \cdot \log (0 / 0) = 0$ and $c\cdot\log(1/0) = \sign(c)\infty$ for $c \neq 0$.
\begin{description}
    \item[Entropy] The \emph{entropy} $H(A)$ of $A$ is defined as
    \[
    H(A)
    \eqdef
    \sum_{a \in \calA} \Pr[A = a] 
    \cdot \log \parens*{ \frac{1}{\Pr[A = a]}}
    \ .\]
    
    \item[Conditional entropy] The \emph{entropy of $A$ conditioned on $B$} is defined as 
    \[H(A \mid B) 
    \eqdef \sum_{b \in \calB} \Pr[B = b] 
    \cdot \sum_{a \in \calA} \Pr[A = a \mid B = b] 
    \cdot \log\parens*{ \frac{1}{\Pr[A = a \mid B = b]}}\ .\]
    
    \item[Joint distribution] The \emph{joint distribution} $AB$ of $A$ and $B$ is the distribution over $\calA \times \calB$ s.t.\ $\forall (a,b) \in \calA \times \calB, \Pr[AB=(a,b)] = \Pr[A = a \wedge B=b]$.

    As immediate properties, we have $H(AB) \leq H(A) + H(B)$, with equality when $A$ and $B$ are independent, and $H(A \mid B) = H(AB) - H(B)$.
    
    \item[Mutual information] The \emph{mutual information} between $A$ and $B$ is defined as
    \[I(A:B) \eqdef H(A) + H(B) - H(AB)\ .\]
    As with entropy, we define mutual information conditioned on a third variable $C$ as
    \[I(A:B \mid C) \eqdef H(A \mid C) + H(B \mid C) - H(AB \mid C)\ .\]
    As immediate property, we have $I(A:B \mid C) + H(C) = I(AC : BC)$.
    
    \item[Kullback-Leibler divergence] The \emph{(Kullback-Leibler) divergence} between two distributions $p(x)$ and $q(x)$ over a common support $\calX$ is
    \[\KL{p}{q} \eqdef \sum_{x \in \calX} p(x) \log\frac{p(x)}{q(x)} = \Exp_{p(x)}\bracks*{\log \frac{p(x)}{q(x)}}\ .\]
    Note that the divergence is $+\infty$ when for some element $x$ we have $q(x) = 0$ but $p(x) \neq 0$. 
    \item[Total variation distance] The \emph{total variation distance} between two distributions $p(x)$ and $q(x)$ over a common support $\calX$ is
    \[\norm{p - q}_1 = \sum_{x \in \calX} \abs{p(x) - q(x)} = 2 \max_{S \subseteq \calX}(p(S) - q(S))\ .\]
    The notation $\norm{.}_1$ highlights that this distance is equivalent to the $\ell^1$-distance between $p$ and $q$, defined as the $\ell^1$-norm of the vector $p-q$.
\end{description}

\begin{lemma}[Pinsker's inequality]
    \label{lem:pinsker}
    \[ 
    \norm*{p - q}_1
    \leq \sqrt{\frac{\ln(2)}{2} \cdot \KL{p}{q}}\ .\]
\end{lemma}
\begin{corollary}
    \label{cor:pinsker-information}
    Let $A$ and $B$ be two random variables of distributions $p_A$ and $p_B$, with $p_{A \mid B=b}$ the distribution of $A$ conditioned on $B=b$. We have
    \[
    \Exp_{b \sim B}\bracks*{\norm*{p_A - p_{A \mid B=b}}_1}
    \leq \sqrt{\frac{\ln(2)}{2} \cdot I(A:B)}\ .
    \]
\end{corollary}
The corollary follows from Pinsker's inequality, together with $I(A:B) = \Exp_{b \sim B}\bracks*{\KL{p_{A \mid B=b}}{p_A}}$, the concavity of $x \to \sqrt{x}$, and Jensen's inequality. \begin{definition}
    \label{def:information-cost}
    Consider a protocol $\pi$ accepting inputs from $\calX \times \calY$. Let $\mu$ be the distribution of $X$ and $Y$. Then, the \emph{internal information cost} of $\pi$ over $\mu$ is
    \[\ICint{\mu}(\pi) \eqdef I(\Pi:X \mid Y) + I(\Pi:Y \mid X)\ ,\]
    and the \emph{external information cost} of $\pi$ over $\mu$ is
    \[\ICext{\mu}(\pi) \eqdef I(\Pi:XY)\ .\]
\end{definition}
\begin{proposition}
    For a protocol $\pi$ and distribution of input $\mu$,

    \begin{itemize}
        \item $\ICint{\mu}(\pi) \leq \ICext{\mu}(\pi)$.
        \item When $\mu$ is a product distribution, $\ICint{\mu}(\pi) = \ICext{\mu}(\pi)$.
        \item Information cost lower bounds communication cost $\CC_\mu(\pi)$, the maximum number of bits that can be sent in an execution of $\pi$.
    \end{itemize}
\end{proposition}

\begin{definition}
    \label{def:information-complexity}
    For a communication complexity task $T$, its \emph{internal information complexity} $\ICint{}(T)$ is the infimum taken over all internal information costs of protocols performing $T$. Its \emph{external information complexity} $\ICext{}(T)$ is the infimum taken over all external information costs of protocols performing $T$.

    When the task can be performed with error $\eps$ over an input distribution $\mu$, the quantities $\ICint{\mu}(T,\eps)$ and $\ICext{\mu}(T,\eps)$ are infimums taken over the internal/external information costs of protocols performing $T$ with error at most $\eps$ on input distribution $\mu$.
\end{definition}

Note that as internal and external information costs coincide when inputs taken from a product distribution, internal and external information complexities also coincide in that situation. As a result, we omit the superscripts $\mathsf{int}$ and $\mathsf{ext}$ and simply write $\IC_{\mu}(T,\eps)$ when $\mu$ is a product distribution.

\begin{lemma}
    \label{lem:error-to-information}
    Let $T$ be some communication complexity task for some product distribution of input $\mu = \mu_X \times \mu_Y$, and $\eps,\rho > 0$ be s.t.\ $\IC_{\mu}(T,\eps) < \rho$. Then, there exists a zero communication protocol for performing $T$ with error at most $\eps+\sqrt{\frac{\ln(2)}{8} \cdot \rho}$ over $\mu$.
\end{lemma}
\begin{proof}
    There exists a protocol $\pi$ s.t.\ $\IC_{\mu}(\pi) \leq \rho$. From the definition of external information cost, there is little mutual information between the inputs $XY$ and the protocol's transcript $\Pi$. Consider the following protocol: Alice and Bob sample random inputs from public randomness $X'$ and $Y'$, completely independent from their actual inputs $XY$. Then, using again public randomness to sample random bits as needed, Alice and Bob locally execute $\pi$ over the sampled inputs $X'$ and $Y'$ without any communication, and use the transcript of this protocol to decide on their output for their communication complexity task $T$. Let $\Pi'$ be the distribution of transcripts from the execution of $\pi$ without any communication using inputs and private randomness sampled from public randomness. $\Pi'$ has the same distribution as $\Pi$, and from \cref{cor:pinsker-information} (Pinsker's inequality), we have
    \[
    \Exp_{(x,y) \sim XY}\bracks*{\norm*{p_\Pi - p_{\Pi \mid XY=(x,y)}}_1}
    \leq \sqrt{\frac{\ln(2)}{2} \cdot I(\Pi:XY)}\ .
    \]
    For each input $(x,y)$, let $S_{x,y}$ be the subset of transcripts that lead to correctly performing $T$. From the definition of total variation distance, the probability of hitting this subset of the transcripts is lowered by at most $\frac{1}{2}\norm*{p_\Pi - p_{\Pi \mid XY=(x,y)}}_1$ when using the locally computed $\Pi'$ as transcript instead of running the protocol on the actual input $(x,y)$. Hence, when sampling inputs over $\mu$, the error probability of this new protocol is $\eps + \sqrt{\frac{\ln(2)}{8} \cdot \rho}$.
\end{proof}

\begin{corollary}
    \label{cor:mcol-ic}
    Computing $\MCOL_4$ with error at most $1/1000$ over the uniform distribution has information complexity $\ICint{\mu}(\MCOL_4,1/1000) > 0.0001942$. 
\end{corollary}
\begin{proof}
From \cref{lem:gadget-no-communication-error} and \cref{lem:error-to-information}, we get that $\ICint{\mu}(\MCOL_4,\eps) \geq \frac{8}{\ln(2)}\parens*{\frac{1}{196} - \eps}^2$. The result follows from plugging in $\eps = 1/1000$.
\end{proof}

\paragraph{Increasing the size of the communication task.}
For a set of input distributions $\calM$ and a communication complexity task $T$ (possibly allowing for some error), define $\ICint{}(T,\calM)$ as the infimum of the information costs of all communication complexity protocols that solve task $T$ on any distribution from the set $\calM$.

\begin{lemma}[Direct sum for internal information complexity~{\cite[Theorem~4.2]{Braverman_siamrev17}}]
    \label{lem:information-direct-sum}
    
    Let $T_1$ and $T_2$ be two tasks over input spaces $\calX_1 \times \calY_1$ and $\calX_2 \times \calY_2$, $\calM_1$ and $\calM_2$ be two sets of distributions over $\calX_1 \times \calY_1$ and $\calX_1 \times \calY_2$. Let $T = T_1 \times T_2$. Then 
\[
    \ICint{}(T,\calM_1 \times \calM_2) 
    = \ICint{}(T_1,\calM_1) + \ICint{}(T_2,\calM_2)
    \ .\]
\end{lemma}

\begin{lemma}
    \label{lem:mcol-copies}
    Solving $k$ independent instances of $\MCOL_4$ with error at most $1/1000$ for each instance has information complexity $\Omega(k)$, and a fortiori, communication complexity $\Omega(k)$.
\end{lemma}
\begin{proof}
    Let $T$ be the task of solving an individual instance of $\MCOL_4$ with error at most $1/1000$. From \cref{cor:mcol-ic}, $T$ has information complexity $\Omega(1)$ on inputs from the uniform distribution. From \cref{lem:information-direct-sum}, solving $k$ instances of $\MCOL_4$ given $k$ inputs from the uniform distribution has information complexity $\Omega(k)$. As information complexity lower bounds communication complexity, this gives the stated lower bound on communication complexity.
\end{proof}

\begin{figure}[hb]
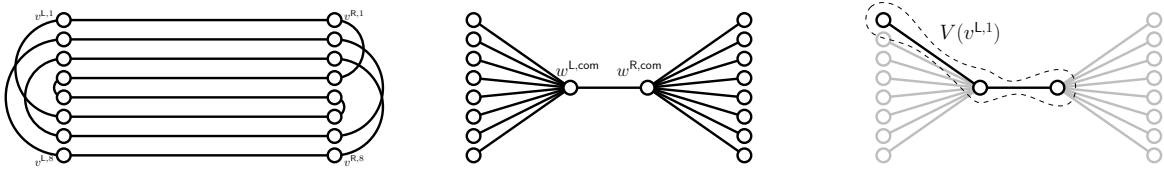

    \centering
    \includegraphics[page=3,width=0.33\textwidth]{congest_lb_virtual.pdf}\includegraphics[page=1,width=0.33\textwidth]{congest_lb_communication.pdf}\includegraphics[page=8,width=0.33\textwidth]{congest_lb_communication.pdf}\caption{Examples of a virtual graph $H_{1,x,y}$ (left), a communication network $G_{1,x,y}$ (middle) in which it can be embedded, and the support of the top left virtual node (right).}
    \label{fig:congestion-lb-graph-copy}
\end{figure}

\paragraph{Putting the communication problem into a virtual graph.} We now map our communication complexity problem into our virtual graph framework. Intuitively, the graphs involved in our communication complexity gadget problem become our virtual graph, while the communication graph is simply a graph with the same perfect matchings on both sides with a single edge connecting all the nodes in the matching of Alice to all the nodes in the matching for Bob, resulting in all the congestion happening over this central edge.
For any integer $k$, $x \in [105]^k$ and $y \in [105]^k$, consider the following virtual and communication graphs $H_{k,x,y}$ and $G_{k,x,y}$:

\begin{description}
    \item[Virtual graph] $H_{k,x,y}$ contains $16k$ nodes and $16k$ edges. Each node is indexed by two numbers $(i,j) \in [k]\times[8]$ and is either in the left or the right part of the graph, i.e., we have nodes $\vL{1}_1,\ldots,\vL{8}_1,\vL{1}_2,\ldots,\vL{8}_k$ and $\vR{1}_1,\ldots,\vR{8}_1,\vR{1}_2,\ldots,\vR{8}_k$. For each $(i,j) \in [k]\times[8]$, there is an edge between $\vL{j}_i$ and $\vR{j}_i$. For each $i \in [k]$, $x_i \in [105]$ describes a perfect matching over the nodes $\vL{1}_i,\ldots,\vL{8}_i$, and similarly, $y_i \in [105]$ describes a perfect matching over the nodes $\vR{1}_i,\ldots,\vR{8}_i$ (recall that there are $\frac{(2x)!}{2^x\cdot x!}$ perfect matchings over a set of $2x$ nodes, which is $105$ in the case of $8$ nodes).
    
    \item[Communication graph] $G_{k,x,y}$ contains $16k+2$ machines. Two central machines, $\wLcom$ and $\wRcom$, are each the root of a star with $8k$ leaves. The leaves are $(\wL{j}_i)_{i\in [k],j \in [8]}$ for the star rooted at $\wLcom$, and $(\wR{j}_i)_{i\in [k],j \in [8]}$ for the other one. $\wLcom$ and $\wRcom$ are linked together.

\item[Embedding]
    For each $(i,j) \in [k]\times[8]$, the virtual node $\vL{j}_i$ has support $V(\vL{j}_i)=\set{\wL{j}_i,\wLcom,\wRcom}$, while $\vR{j}_i$ has the smaller support $V(\vR{j}_i)=\set{\wR{j}_i,\wRcom}$. Each edge of the form $\vL{j}_i\vR{j}_i$ is handled by $\wRcom$. Edges between left nodes (resp., right nodes) are handled by $\wLcom$ (resp., $\wRcom$). That is, $\wLcom$ and $\wRcom$ each know one of the two perfect matchings.
    
\end{description}

For each $i\in [k]$, we refer to the virtual subgraph over the nodes $\vL{1}_i,\ldots,\vL{8}_i$ and $\vR{1}_i,\ldots,\vR{8}_i$ as the $i$th gadget. \Cref{fig:congestion-lb-graph-copy} (copy of \cref{fig:congestion-lb-graph}) shows a virtual graph and a communication graph with $k=1$, i.e., with a single gadget.

\begin{lemma}
\label{lem:congestion-lb}
$\deg+1$ coloring the virtual graph $H_{k,x,y}$ with $G_{k,x,y}$ as communication graph with error at most $1/1000$ requires $\Omega(\congestion / \bandwidth+1)$ communication rounds of bandwidth $\bandwidth$.
\end{lemma}
\begin{proof}
    Any $T$-round distributed protocol for this problem immediately implies a $T\cdot \bandwidth$ communication complexity protocol for 
    $\MCOL_4^{\otimes \congestion}$, which we have shown requires $\Omega(\congestion)$ communication (\cref{lem:mcol-copies}), even with a less stringent error requirement. With the other constraint that $0$ communication is insufficient (\cref{lem:gadget-no-communication-error}), i.e., at least $1$ distributed round is needed, we get the $\Omega(\congestion / \bandwidth+1)$ lower bound.
\end{proof}

\subsection{Lower-bound with respect to dilation}
\label{sec:lower-bound-dilation}

We obtain the dilation component of our lower bound as an easy consequence of the seminal $\Omega(\log^* n)$ lower bound for $3$-coloring paths by Linial~\cite{linial92}, and its extension to randomized algorithms by Naor~\cite{Naor_siamdm91}.

\begin{lemma}
    \label{lem:dilation-lower-bound}
    A $T(n)$-round algorithm for $\deg+1$-coloring virtual graphs of dilation $\dilation$ implies a $O(T(n)/\dilation)$-round algorithm for $3$-coloring paths.
\end{lemma}
\begin{corollary}[\Cref{lem:dilation-lower-bound} + $\Omega(\log^* n)$ lower bound for $3$-coloring paths~\cite{linial92,Naor_siamdm91}]
    $\deg+1$-coloring virtual graphs of dilation $\dilation$ has complexity at least $\Omega(\dilation \cdot \log^*n)$, both for deterministic and randomized algorithms.
\end{corollary}
\begin{proof}
    Let $\calA$ be the $T(n)$-round algorithm for $\deg+1$-coloring virtual graphs of dilation $\dilation$. Let $P=u_1,\ldots,u_n$ be the path that we will $3$-color using $\calA$, and let its nodes have infinite bandwidth, i.e., we are in the \local model. The path is taken as the virtual graph $H$ in the context of our $\deg+1$-coloring algorithm, and the nodes in the path simulate $\calA$ on an imaginary communication graph $G$. $G$ is constructed as follow: each node $u_i$ in the path $P$ is replaced by a path of $2\dilation + 1$ machines (of edge-length $2\dilation$), that acts as the support tree $T(u_i)$ of $u_i$, with the root in the middle of the path. The support trees are made to overlap at their extremities, i.e., one of the 2 leaves of $T(u_2)$ is also a leaf of $T(u_1)$, and the other is also a leaf of $T(u_3)$. The space of IDs of $G$ is slightly increased from that of $P$ (by a factor $\Theta(\dilation)$) to give each machine a unique ID. Each $u_i$ holds the information about the machines in $T(u_i)$.

    After $x$ rounds of communication over the original path $P$, each node in $P$ has learned everything within distance $x$ of itself. This means it knows everything within distance $x\cdot 2 \dilation$ from any imaginary machine $v \in T(u_i)$, and in particular, it can simulate the full behavior of these machines in algorithm $\calA$ for $x\cdot 2 \dilation$ rounds. With $x \geq T(n) / (2 \dilation)$, $u_i$ can simulate an entire run of $\calA$ for the machine $v \in T(u_i)$, which means that the machines of the imaginary communication graph must reach a state in which they correctly $\deg+1$-colored their virtual graph. Since the virtual graph of the imaginary communication graph is the path we started with, in which the maximum degree is $2$, this yields the claimed $O(T(n)/\dilation)$ algorithm for $3$-coloring a path.
\end{proof}

 \newpage

\section{Degree Reduction in Low Degree Setting}
\label{sec:log-deg-beps}
\LowDegreeReduction*

As explained in the main text (\cref{sec:low-deg-sampling}), the following lemma is essentially an adaptation of \cite[Lemma 5.4]{BEPSv3}.

\begin{proof}
    Consider a $u$ node of uncolored pseudo-degree $\deg_\col(u) \geq C \ln n$, for some constant $C$ in some given loop iteration. Each of its $\card{N_\col(v)} \leq \deg_\col(u)$ neighbors has a probability $1/2$ of trying a random color. By \cref{lem:chernoff} (Chernoff bound), with probability at least $1-n^{-C/12}$, there are no more than $3\deg_\col(u)/4$ uncolored neighbors of $u$ trying a random color in this iteration.

    Let $v$ be a node with at least $C \ln n$ uncolored neighbors with uncolored degree at least $C\ln n$ in some loop iteration. Let $U$ be this high-degree subset of $v$'s neighborhood, i.e., $U = \set{ u \in N_\col(v) : \deg_\col(u) \geq C\ln n}$, and $\card{U} \geq C \ln n$. With probability at least $1-n^{1-C/12}$, each high-degree neighbor $u\in U$ of $v$ has no more than $3\deg_\col(u)/4$ uncolored neighbors trying a color in this iteration. Additionally, again by \cref{lem:chernoff} (Chernoff bound), at least $\card{N_\col(v)}/4$ nodes in $U$ try a random color in this iteration with probability at least $1-n^{-C/8}$. Let $W \subseteq U$ be the set of high-degree neighbors of $v$ that try a random color in this iteration.
    
    Let us condition on the previous high probability random events, so $v$ has many high-degree neighbors trying a random color ($W$), which all have not too many neighbors trying a color in this iteration. All these nodes sample a random color in their palettes using \alg{LowDegSampling}. The color sampled by a node $u\in W$ has probability at most $1/(\deg_\col(u)+1)$ of being chosen. Therefore, regardless of colors sampled by higher ID nodes, each node in $W$ has a probability at least $1/4$ of sampling color not taken by higher degree neighbor. In expectation, at least $\card{W}/4 \geq \card{U}/16$ high-degree neighbors of $v$ get colored, and by \cref{lem:chernoff} (Chernoff bound), with probability at least $1-n^{-C/128}$, $v$ has at most $31\card{U}/32$ high-degree uncolored neighbors in the next iteration of the loop.

    Setting $C$ high enough, we get that with high probability, all nodes with many high degree neighbors have at least a constant fraction of them get colored. After $\log_{1+1/31} \Dmaxcol \in \Theta(\log \Dmaxcol)$ iterations, all nodes have at most $C \ln n$ uncolored neighbors with uncolored pseudo-degree of $C \ln n$ or more.

    To partition uncolored vertices of the graph into two sets each inducing a subgraph of maximum degree $O(\log n)$, it then suffices to partition nodes according to their uncolored pseudo-degrees. The subgraph induced by uncolored nodes of uncolored pseudo-degree $2 C \ln n$ or less trivially has maximum (real) degree $O(\log n)$, and the graph induced by the remaining nodes has $O(\log n)$ maximum (real) degree by the current lemma.
\end{proof} 
\section{Computing the Almost-Clique Decomposition}
\label{sec:ACD}
\label{sec:additional}
\LemACD*

\subsection{Balanced \& Friendly Edges}

Let $\theta \in (0, 1/20)$ be a small parameters.
Using \cref{lem:fingerprinting}, w.h.p., each vertex computes $\tilde{d}(v) \in (1\pm\theta^3)|N(v)|$ in $\Ohat(1/\theta^6)$ rounds. Define 
\[ 
    \Vin = \set{v\in V_H : \deg(v) \geq (1 + 2\theta^3) \tilde{d}(v)  }
\]

\begin{claim}
    \label{lem:inacc}
    All $v\in \Vin$ are $\theta^3/2 \cdot |N(v)|$-inaccurate and $v\in V_H \setminus \Vin$ have $\deg(v) \leq (1 + 5\theta^3)|N(v)|$.
\end{claim}

We classify each pair of adjacent vertices $u,v \in V_H$ as follow:
\begin{itemize}
    \item \textbf{$\theta$-balanced:} 
        $\min\set{\deg(u),\deg(v)} \geq (1-2\theta)\max\set{\deg(u), \deg(v)}$
    \item \textbf{$\theta$-friendly:}
        $|N(u) \cap N(v)| > (1-\theta)\min\set{|N(u)|, |N(v)|}$
\end{itemize}
Note that the definition of balanced is in term of list-size (pseudo-degrees) while friendliness is in term of number of neighbors. For accurate vertices, being pseudo-degree balanced is equivalent to being have balanced neighborhoods up to a constant factor.
\begin{fact}
    \label{fact:balanced-neighborhood}
    A $\theta$-balanced pair $u,v \in V_H \setminus \Vin$ has
    $\min\set{|N(u)|,|N(v)|} \geq (1-\theta)\max\set{|N(u)|, |N(v)|}$.
\end{fact}

Each machine in $V(v) \cap V(u)$ can test locally if the pair $u,v$ is balanced since vertices know pseudo-degrees. Testing if a pair is friendly needs approximating the number of neighbors, which we do in \cref{lem:buddy}.

\begin{lemma}
    \label{lem:buddy}
There exists a $\Ohat(1/\theta^2)$-round algorithm computing a set $F$ of adjacent vertices such that, w.h.p.,
    for each pair of adjacent vertices $u,v \in V_H$,
    \begin{itemize}
        \item if $\set{u,v}$ is $\theta$-friendly, then $\set{u,v}\in F$;
        \item if $\set{u,v}$ is \emph{not} $16\theta$-friendly, then $\set{u,v} \notin F$.
    \end{itemize}
    Pairs of adjacent vertices that are $16\theta$-friendly but not $\theta$-friendly are classified arbitrarily by the algorithm. For each $\set{u,v} \in F$, at least one machine in $V(v) \cap V(u)$ knows it.
\end{lemma}

\begin{proof}
    Each vertex $v \in V_H$ samples $t = \Theta(\theta^{-2}\log n)$ independent geometric random variables $X_{v, 1}, \ldots, X_{v,t}$. By \cref{lem:fingerprint-encoding}, vectors $Y_{v,i} = \max\set{X_{u', i},~u\in N(v)}$ are disseminated to all machines in $V(v)$ in $\Ohat(\theta^{-2})$ rounds.
    Each $w\in V(v) \cap V(u)$ therefore received both $Y_{u,i}$ and $Y_{v,i}$ for all $i\in[t]$, thus can compute $Y_{uv, i} = \max\set{Y_{v,i}, Y_{u,i}} = \max\set{X_{u', i},~u'\in N(u) \cup N(v)}$ for each $i\in [t]$. By \cref{lem:concentration-fingerprint}, w.h.p., each machine $w \in V(u) \cap V(v)$ computes $f_w \in (1 \pm \theta)|N(v) \cup N(u)|$, an approximation of the size of the joint neighborhood. 

    Let $u, v \in V_H$ be adjacent vertices. Rename $w\in\set{u,v}$ and $\bar{w} \in \set{u,v} \setminus \set{w}$ such that $\tilde{d}(w) \geq \tilde{d}(\bar{w})$.
    We argue that machines in $V(v) \cap V(u)$ can tell the difference between the case where $|N(v) \cap N(u)| \geq (1 - \theta)\min\set{|N(v)|,|N(u)|}$ and $|N(v) \cap N(u)| < (1 - 16\theta)\min\set{|N(v)|,|N(u)|}$ based on estimates $\tilde{d}(w)$, $\tilde{d}(\bar{w})$ and $f = f_w$.

    \begin{itemize}
        \item \textbf{Case 1:} $\set{u,v}$ is a $\theta$-friendly edge. 
By definition, their joint neighborhood can only be slightly larger than $\tilde{d}(w)$:
        \begin{align*}
            |N(v) \cup N(u)| 
            &= |N(v)| + |N(u)| - |N(v) \cap N(u)| \\
            &\leq |N(v)| + |N(u)| - (1 - \theta)\min\set{|N(u)|, |N(v)|} \\
            &\leq \tilde{d}(v) + \tilde{d}(u) - (1 - \theta)\tilde{d}(\bar{w}) + 3\theta^3 \cdot \max\set{|N(v)|, |N(u)|}\\
            &\leq \tilde{d}(w) + \theta\tilde{d}(\bar{w}) + 6\theta^3 \cdot \tilde{d}(w) 
            \leq (1 + 2\theta)\tilde{d}(w) \ .
        \end{align*}
        Hence, the estimation verifies
        $
            f \leq (1 + \theta)|N(u) \cup N(v)| \leq (1+5\theta) \tilde{d}(w) 
        $ since $\theta < 1$.

    \item \textbf{Case 2:} $\set{u,v}$ is \emph{not} $16\theta$-friendly. By definition, their neighborhoods overlap only on a small fraction of $\tilde{d}(\bar{w})$:
        \begin{align*}
            |N(u) \cap N(v)| 
            &< (1 - 16\theta) \min\set{|N(u)|, |N(v)|} \\
            &\leq (1 - 16\theta)\tilde{d}(\bar{w}) + \theta^3 \min\set{|N(u)|, |N(v)|}\\
            &\leq (1 - 16\theta)\tilde{d}(\bar{w}) + 2\theta^3 \tilde{d}(\bar{w})
            = (1 - 15\theta)\tilde{d}(\bar{w}) \ .
        \shortintertext{Which implies the joint neighborhood must be large }
            |N(u) \cup N(v)| &= 
            |N(v)| + |N(u)| - |N(u) \cap N(v)| \\ &>
            \tilde{d}(v) + \tilde{d}(u) - (1 - 15\theta)\tilde{d}(\bar{w}) - 6\theta^3 \tilde{d}(w)
            \\ &\geq
            (1 + 14\theta)\tilde{d}(w) 
        \end{align*}
        Hence, the estimate 
        $
            f \geq (1-\theta)|N(u) \cap N(v)| >
            (1 + 5\theta) \tilde{d}(w)
        $,
        for $\theta < 1/10$. 
    \end{itemize}
    Each machine in $V(v) \cap V(u)$ can therefore tell apart the two cases with high probability.
\end{proof}

\subsection{Proof of 
\texorpdfstring {\cref{prop:AC}}{Lemma~\ref*{prop:AC}}}

Let $F$ be the set of adjacent pairs selected by the algorithm in \cref{lem:buddy}. We construct the \emph{simple} graph $D = (V_D, F_D)$ on vertex set $V_D = V_H \setminus \Vin$ and $\theta$-balanced pairs of $F$ \emph{with both endpoints accurate}, i.e., $F_D = \set{ \set{u,v}\in F: u,v\in V_D \text{ and is $\theta$-balanced} }$.

\begin{lemma}
    \label{lem:is-sparse}
    Let $v \in V_D$.
    Then $\spar_v + \discr_v + |N(v) \cap \Vin| \geq \theta/6 \cdot |N(v) \setminus N_D(v)|$.
\end{lemma}

\begin{proof}
    \newcommand{\ov}[1]{\overline{#1}}
    Partition unfriendly neighbors $u \in N(v) \setminus N_D(v)$ into three sets: $I(v)$, those in $\Vin$; $\ov{B}(v)$, where $u\notin \Vin$ and $\set{u,v}$ is $\theta$-unbalanced; and $R(v)$, the remaining ones. The first two cases are adaptations of \cite[Claim 4.8]{AA20} to factor in inaccuracies. We add \textbf{Case 3} to handle vertices with many inaccurate neighbors.
    \begin{itemize}
        \item\textbf{Case 1:} $|R(v)| \geq |N(v) \setminus N_D(v)|/3$. Since $\set{u,v}$ is balanced with both endpoints accurate, it must be that $|N(v) \cap N(u)| < (1 - \theta)|N(v)|$ (otherwise it would be $F_D$). Hence, each $u \in R(v)$ contributes to at least $\theta|N(v)|$ non-edges between neighbors of $v$. This gives the results as
        $
            \spar_v 
            \geq \theta |R(v)|
            \geq \theta/3 \cdot |N(v) \setminus N_D(v)|
        $.
    
        \item\textbf{Case 2:} $|\ov{B}(v)| \geq |N(v) \setminus N_D(v)|/3$. Let $\ov{B}_+(v)$ denote the vertices $u\in \ov{B}(v)$ where $\deg(v) \leq (1 - 2\theta)\deg(u)$ and $\ov{B}_-(v)$ the ones where $\deg(u) \leq (1 - 2\theta)\deg(v)$. 
        When $|\ov{B}_+(v)| > |\ov{B}_-(v)|$, since each $u \in \ov{B}_+(v)$ contributes $\theta$ to the unevenness of $v$, we get
        $
            \discr_v \geq \theta \cdot |\ov{B}_+(v)| \geq \theta/6 \cdot |N(v) \setminus N_D(v)|
        $.
        Suppose $|\ov{B}_-(v)| \geq |\ov{B}_+(v)|$. Since $v \notin \Vin$, each $u\in \ov{B}_-(v)$ has a small degree compared to $v$ 
        \[
            |N(u)| \leq \deg(u) \leq (1 - 2\theta)\deg(v) \leq (1 - 2\theta)(1 + 5\theta^3) |N(v)| 
            \leq (1 - \theta) |N(v)| \ ,
        \]
        hence contributes $\theta \cdot |N(v)|$ non-edges to the sparsity of $v$. Thus, $\spar_v \geq \theta/6 \cdot |N(v) \setminus N_D(v)|$.

        \item\textbf{Case 3:} $|I(v)| \geq |N(v) \setminus N_D(v)|/3$. Since no pair $\set{u,v}$ where $u\in\Vin$ is selected in $F_D$, we get that $N(v) \cap \Vin = I(v)$ and the bound follows.
    \end{itemize}
\end{proof}

We say a vertex is \emphdef{$\theta$-highly-dense}\footnote{they call it $\theta$-dense in \cite{AA20}} if it is incident to at least $(1-\theta)|N(v)|$-friendly edges which have $\theta$-balanced neighborhoods, i.e., $\min\set{|N(u)|,|N(v)|} \geq (1-\theta)\max\set{|N(u)|, |N(v)|}$.

\begin{proof}[Proof of \cref{prop:AC}]
    Let $\theta = \eps/100$.

    By \cref{lem:inacc}, we compute $\Vin$ in $\Ohat(1/\theta^6)$ rounds.
    Each vertex uses the \hyperref[lem:fingerprinting]{fingerprinting algorithm} to estimate $|N_D(v)|$ up to error $\theta/2 \cdot |N(v)|$. We label as highly-dense each vertex $v\in V_D$ whose estimated number of neighbors in $D$ is at least $(1-\theta)\deg(v)$. Such a vertex has $|N_D(v)| \geq (1 - 1.5\theta)|N(v)|$ neighbors with whom it is (at least) $(1-16\theta)$-friendly. 
    Further, both endpoints are accurate, and hence, the neighborhood of each pair in $F$ is balanced (\cref{fact:balanced-neighborhood}).
    Each labeled vertex is, therefore, $16\theta$-highly-dense.

    Let $K_1, \ldots, K_k$ be the connected components of $D$ where we labeled at least one vertex. Define $\Vdense = K_1 \cup \ldots \cup K_k$; we claim that sets $K_1, \ldots, K_k$ are $100\theta$-almost-cliques. Indeed, \cite[Claim 4.7]{AA20} implies the first two bounds of \cref{part:ACD-dense} and \cref{lem:is-sparse} implies \cref{part:ACD-ext} because $N(v) \setminus K = N(v) \setminus N_D(v)$.

    Let $\Vstar \eqdef V_H \setminus (\Vin \cup K_1 \cup \ldots \cup K_k)$.
    Each vertex $v \notin \Vin$ that is not in one of these components must have estimated $< (1 - \theta)\deg(v)$ neighbors in $D$ hence 
    \begin{align*}
        |N_D(v)| 
        &< (1 - \theta/2)\deg(v) 
        \leq (1 - \theta/2)(1 + 5\theta^3)|N(v)|
        \leq (1 - \theta/4) |N(v)| \ ,
    \end{align*}
    for $\theta < 1/4$. Hence, by \cref{lem:is-sparse} they have 
    \[ 
        \spar_v + \discr_v + |N(v) \cap \Vin| 
        \geq \theta/6 \cdot |N(v) \setminus N_D(v)|
        \geq \theta^2/24 \cdot |N(v)| \geq \theta^2/50 \cdot \deg(v) \ .
    \]
    Vertices of $\Vstar$ thus verify \cref{part:ACD-star}.

    As almost-cliques have diameter two, they can be identified in $O(\congestion\dilation)$ rounds. This gives the $\eps$-almost-clique decomposition. Let $\Vlow$ be all vertices with $\deg(v) \leq \Deltalow$. If some vertex $v\in K$ joins $\Vlow$, we add all vertices of $K$. Hence every vertex in $\Vhigh$ has degree $\geq \Deltalow$ and every vertex in $\Vlow$ has degree $(1 + \eps)\Deltalow \leq 2\Deltalow$. We also remove from $\Vin$, $\Vstar$ all low degree vertices.
\end{proof}

\end{document}